\documentclass[10pt]{article}
\usepackage[left=0.8in,top=0.8in,right=0.8in,bottom=0.8in]{geometry}
\usepackage{graphicx}
\usepackage{subfig}
\usepackage{amsmath}
\usepackage{amssymb}
\usepackage{amsthm}
\usepackage{latexsym}
\usepackage{bm}
\usepackage{array,supertabular}
\usepackage{color}
\usepackage{framed}
\usepackage{setspace}
\usepackage{fancyhdr}
\usepackage{stmaryrd}
\usepackage{mathrsfs}
\usepackage{mdframed}
\usepackage{url}
\usepackage{multirow}
\usepackage{algorithm}
\usepackage{algpseudocode}
\usepackage{pifont}
\usepackage{longtable}
\usepackage{booktabs}
\usepackage{graphicx}
\usepackage{rotating}
\usepackage{thmtools}
\usepackage{xcolor}
\usepackage{makecell}

\makeatletter
\newcommand*{\rom}[1]{\expandafter\@slowromancap\romannumeral #1@}
\makeatother

\newcolumntype{P}[1]{>{\centering\arraybackslash}p{#1}}

\SetSymbolFont{stmry}{bold}{U}{stmry}{m}{n}

\newtheorem{definition}{Definition}
\newtheorem{lemma}{Lemma}

\newtheorem{proposition}{Proposition}

\newtheorem{remark}{Remark}
\newtheorem{assumption}{Assumption}

\usepackage[many]{tcolorbox}

\numberwithin{equation}{section}

\begin{document}
\title{A framework for finite-strain viscoelasticity based on rheological representations}
\author{Chongran Zhao, Hongyan Yuan, and Ju Liu\\
\textit{\small Department of Mechanics and Aerospace Engineering,}\\
\textit{\small Southern University of Science and Technology,}\\
\textit{\small 1088 Xueyuan Avenue, Shenzhen, Guangdong 518055, China}\\
\textit{\small Email: chongranzhao@outlook.com, yuanhy3@sustech.edu.cn, liuj36@sustech.edu.cn}
}
\date{}
\maketitle

\section*{Abstract}
This work presents a new constitutive and computational framework based on strain-like internal variables belonging to Sym(3) and two representative rheological configurations. The generalized Maxwell and generalized Kelvin-Voigt models are considered as prototypes for parallelly and serially connected rheological devices, respectively. For each configuration, distinct kinematic assumptions are introduced. The constitutive theory is derived based on thermomechanical principles, where the free energies capture recoverable elastic responses and dissipation potentials govern irreversible mechanisms. The evolution equations for the internal variables arise from the principle of maximum dissipation. A key insight is the structural distinction in the constitutive laws resulted from the two rheological architectures. In particular, the Kelvin–Voigt model leads to evolution equations with non-equilibrium processes coupled, which pose computational challenges for the constitutive integration. To address this, we exploit the Sherman–Morrison–Woodbury formula and extend it to tensorial equations to design an efficient strategy during constitutive integration. With that strategy, the integration can be performed based on an explicit update formula, and the algorithmic complexity scales linearly with the number of non-equilibrium processes. This framework offers both modeling flexibility and computational feasibility for simulating materials with multiple non-equilibrium processes and complex rheological architectures under finite strain.

\vspace{5mm}

\noindent \textbf{Keywords:} 
Viscoelasticity, Generalized Kelvin-Voigt model, Generalized standard materials, Micromechanical model, Constitutive integration, Sherman-Morrison-Woodbury formula

\section{Introduction}
\label{sec:introduction}
Rheological representations, or mechanical analogs, motivate and guide constitutive model development \cite{Steinmann2021}. These analogs use basic devices, such as springs, dashpots, and friction devices, and arrange them in various series and parallel configurations to characterize complex material behavior, especially inelastic phenomena \cite{Simo2006,Wineman2009,Flaschel2023,Carvalho2023}. In the realm of viscoelasticity, two fundamental building blocks are the Maxwell and Voigt elements \cite{Ferry1980}. The former, consisting of an elastic spring connected in series with a viscous dashpot, captures fluid-like stress relaxation behavior; the latter models solid-like creep, in which the viscous component causes gradual strain accumulation under sustained load. The Zener \cite{Zener1949} and Poynting-Thomson models \cite{Poynting1909} were further introduced as the three-parameter formulation capable of simultaneously describing creep and relaxation. More advanced approaches, such as Prony series, generalize these concepts by incorporating discrete relaxation or creep spectra \cite{Tschoegl1989,Ferry1980}, providing a more accurate representation of complex viscoelasticity. Although they encode the same material information in principle, the generalized Maxwell and Kelvin-Voigt models are naturally suited to describe different viscoelastic responses, i.e., relaxation and creep, respectively. Their distinct practical implications warrant the study of both models.

At the micromechanical level, the generalized Maxwell model admits a variety of physical interpretations. Polymeric materials are typically represented by polymer chain networks, comprising a static network of chemically cross-linked polymer chains that governs elastic behavior, along with multiple transient networks formed by chains connected through weak, reversible bonds or entanglements. In the transient network theory \cite{Green1946,Tobolsky1945}, the continuous bond detachment and reattachment give rise to the viscous, rate-dependent response \cite{Tanaka1992,Linder2011,Vernerey2018}. Alternatively, in the tube theory \cite{De1971,Doi1988}, the viscous effects are attributed to the chain reptation along the tube-like region formed by surrounding chains \cite{Zhou2018,Zhao2022}. The parallel network architecture aligns naturally with the generalized Maxwell model, where total stress results from the additive contributions of elastic and viscous networks. In contrast, viscous behavior may also originate from interchain slippage in polymer strands containing multiple chains, while the elastic response stems from robust crosslinks that preserve network integrity. This microscopic mechanism, as evidenced in hydrogels \cite{Cai2024,Suriano2014}, is captured by the parallel spring-dashpot configuration in the Kelvin-Voigt model. The different micromechanical origins justify a detailed study of both viscoelastic models.

To interpret the rheological representation into a sound continuum model, the bridge is the \textit{kinematic assumption}. It posits the decomposition of the total deformation into elastic and viscous components, enabling consistent definitions of both deformation and rate of deformation within each rheological element. A commonly employed kinematic assumption is the multiplicative decomposition \cite{Lee1969,Sidoroff1974}, which has been extensively applied for the generalized Maxwell model \cite{Reese1998,Tallec1993,Lion1997a,Gouhier2024,Oskui2025}. This decomposition is indeed directly applicable to the generalized Maxwell model, as each Maxwell element involves only two serially connected components. However, we point out that this assumption lacks a strong physical basis for polymeric materials, rendering the notion of the intermediate configuration unclear. Moreover, since the internal variable evolves in GL$_+$(3),  modeling anisotropic materials requires careful specification of the evolution of material symmetry groups  \cite{Sadik2024,Bahreman2022,Ciambella2021}. The generalized Kelvin-Voigt model is less commonly employed in finite viscoelasticity. An exception is the Poynting-Thomson model \cite{Huber2000a,Laiarinandrasana2003,Meo2002,Cai2024}, which involves a single Voigt element. When applied to a model with multiple Voigt elements connected in series, an apparent challenge is the emergence of multiple intermediate configurations. Given the aforementioned issues related to the intermediate configuration, handling multiple such hypothetical configurations is by no means trivial, as it demands distributing the total deformation across all elements. Nevertheless,  rheological models involving multiple serially connected components may arise in practical applications. In this context, taking the generalized Kelvin-Voigt model as a prototype offers a meaningful representative case for investigating more complex rheological models.

A kinematic assumption was introduced for finite viscoelasticity \cite{Liu2024}, drawing inspiration from the Green-Naghdi decomposition in plasticity \cite{Green1965,Naghdi1990} and Miehe's concept of plastic metric \cite{Miehe2000,Miehe1998a}. Unlike the original Green-Naghdi decomposition, the use of generalized strains offers a more flexible framework for kinematically separating elastic and viscous deformations. The concept of generalized strains \cite{Hill1979} has been continuously enriched over time \cite{Liu2024}, thereby offering a rich set of tools for the kinematic decomposition in inelastic theories. The internal variable was taken as a rank-two Lagrangian tensor $\bm \Gamma \in$ Sym$_+$(3) \cite{Liu2024}, akin to the deformation tensor $\bm C$. That choice was largely inspired by the model of Simo \cite{Simo1987,Holzapfel1996b,Liu2021b} and the plastic metric concept \cite{Miehe2000,Miehe1998a}. Importantly, the resulting constitutive theory does not involve the controversial and ambiguous notion of an intermediate configuration, thereby eliminating the need to prescribe additional kinematic laws for a purely hypothetical concept. This feature represents a key conceptual advantage of the proposed formulation.

\subsection*{Contribution}
The choice of the internal variable plays a pivotal role in shaping the constitutive theory. In this work, we develop a theory based on strain-like internal variables belonging to Sym(3). It will be shown that, when paired with quadratic free energies (i.e., hyperelasticity of Hill's class), this formulation presents a unified family of finite deformation linear viscoelasticity, encompassing the models of Simo \cite{Simo1987,Liu2021b}, of Green and Tobolsky \cite{Green1946,Lubliner1985}, and of Miehe and Keck \cite{Miehe2000} as speicial cases. More general constitutive relations can be systematically constructed by adopting a broader class of generalized strains. A key feature of our approach is its ability to calibrate the strain parameters directly from experimental data. As a result, the kinematic decomposition is not imposed a priori but is instead adapted to the specific material under consideration. More importantly, by adopting coercive strain measures, we enable the consistent construction of elastic and viscous deformation tensors. This generalization permits the use of more general hyperelastic energy functions, such as the micromechanically motivated eight-chain model \cite{Arruda1993,Bischoff2001}, and leads to a more general constitutive framework.

We construct our theory based on two representative rheological models: the generalized Maxwell model, characterized by multiple parallel elements, and the generalized Kelvin-Voigt model, composed of rheological devices connected in series. The two configurations serve as prototypes for more complex rheological architectures, representing the essential features of parallel and serial organizations. Owing to their distinct arrangements, we propose two kinematic assumptions. The construction of the constitutive theory follows the fundamental thermomechanical principles. The free energy is employed to characterize the reversible phenomena, while the dissipation potential is introduced to depict the irreversible, dissipative mechanism. In particular, the evolution of the internal variables is derived from the principle of maximum dissipation, placing the theory within the category of generalized standard materials \cite{Ziegler1983,Ziegler1987,Halphen1975,Martyushev2006}. The introduction of the dissipation potential provides a rational basis for handling viscous dissipation within a thermodynamically consistent framework, offering a theoretical foundation for modeling non-Newtonian viscous behaviors.

The resulting structure of the constitutive equations reflects the underlying rheological architecture. While the generalized Maxwell model admits independent evolution equations for each internal variable, the generalized Kelvin–Voigt model exhibits a coupled system of flow rules, in which all non-equilibrium processes are mutually dependent. From a computational perspective, this coupling appears to result in a larger local system during constitutive integration. However, for the finite linear viscoelasticity with the generalized Kelvin–Voigt representation, the matrix arising in the constitutive integration can be expressed as a rank-one modification of a diagonal matrix. This structure allows the matrix inverse to be computed efficiently using the Sherman–Morrison–Woodbury formula \cite{Sherman1950,Hager1989}, enabling an effective and efficient treatment of the coupled evolution equations. For nonlinear models, we draw inspiration from this strategy and develop a decoupling approach to facilitate efficient constitutive updates. As a result, the algorithm complexity of the constitutive integration for the generalized Kelvin–Voigt model scales linearly with the number of non-equilibrium processes, similar to the generalized Maxwell model.

The remainder of this article is organized as follows. Section \ref{sec:continuum_basis} outlines the theoretical foundations. Section \ref{sec:constitutive_theory} develops the constitutive framework for viscoelasticity using the two rheological representations. Detailed numerical formulations are provided in Section \ref{sec:numerical_formula}. Section \ref{sec:results} presents model calibration and finite element analysis results. Concluding remarks are given in Section \ref{sec:conclusion}.

\section{Continuum basis}
\label{sec:continuum_basis}
Before delving into the theoretical discussion, we want to emphasize that, in this work, the summation convention is not adopted. We use a summation sign explicitly for all summations. The algebraic operations of tensors follow the definitions provided in the textbook \cite{Holzapfel2000}.

\subsection{Kinematics}
\label{sec:kinematics}
Consider the reference configuration of the body denoted by $\Omega_{\bm X} \subset \mathbb{R}^3$ with particles labeled by $\bm X \in \Omega_{\bm X}$. Its motion over time $t$ is described by the mapping $\bm x = \bm \varphi(\bm X, t) = \bm \varphi_t(\bm X)$, which establishes a one-to-one correspondence between each material point $\bm X$ and its current spatial position $\bm x$. The deformation gradient $\bm F$ and its Jacobian $J$ are defined as $\bm F := \partial \bm \varphi_t/\partial \bm X$ and $J := \det(\bm F)$, respectively. The right Cauchy-Green deformation tensor $\bm C := \bm F^\mathrm{T} \bm F$ admits the following spectral decomposition
\begin{align}
\label{eq:C_spectral}
\bm C = \sum_{a=1}^3 \lambda_a^2 \bm M_a,
\end{align}
where $\lambda_a$ are the principal stretches, and $\bm M_a := \bm N_a \otimes \bm N_a$ is the self-dyads of the principal referential directions $\bm N_a$, for $a=1,2,3$.

The concept of generalized strains plays a fundamental role in characterizing finite deformation and separating inelastic effects. The  Lagrangian generalized strain is defined as
\begin{align}
\label{eq:def_E}
\bm{E} := \sum_{a=1}^{3} E (\lambda_a) \bm{M}_a,
\end{align}
where $E: (0,\infty) \rightarrow \mathbb{R}$ is a scale function that maps the principal stretches to the principal values of the generalized strain \cite{Hill1968}. To make the definition \eqref{eq:def_E} mathematically sound and physically meaningful, it is demanded that the scale function $E$ has to be at least twice continuously differentiable, monotonically increasing (i.e., $E' > 0$), vanish at the reference state (i.e., $E(1)=0$), and satisfy the normality condition (i.e., $E'(1)=1$) \cite{Hill1979}. Here, we use a prime to denote the first derivative of the univariate function $E$. The generalized strain induces the following rank-four and rank-six tensors:
\begin{align}
\label{eq:def_Q&K}
\mathbb Q  := 2 \frac{\partial \bm E}{\partial \bm C},
\quad
\mathbb Q^{-1} := \frac12 \frac{\partial \bm C}{\partial \bm E},
\quad
\bm{\mathcal L} := 2 \frac{\partial \mathbb Q}{\partial \bm C} = 4 \frac{\partial^2 \bm E}{\partial \bm C \partial \bm C},
\quad
\bm{\mathcal{K}}:= \frac12 \frac{\partial \mathbb Q^{-1}}{\partial \bm E}= \frac14 \frac{\partial^2\bm C}{\partial\bm E\partial \bm E}.
\end{align}
These tensors will be useful in the subsequent discussion. Their explicit forms in terms of the principal stretches and directions can be found in \cite{Miehe2001b,Liu2024}. We cautiously mention that $\mathbb Q$ is not a projection because it does not satisfy the idempotent property (i.e., $\mathbb Q : \mathbb Q \neq \mathbb Q$), although it was previously referred to as such in earlier literature. 

For the Seth-Hill strains, the scale function is $E(\lambda) = (\lambda^{m}-1)/m$, where $m \neq 0$ is a strain parameter. A major drawback of this strain family is that the strain value approaches a finite limit under extreme compression or tension. To rectify this unphysical behavior, an additional property, known as \textit{coerciveness}, is additionally imposed on the scale function,
\begin{align}
\label{eq:E_coerciveness}
\lim_{\lambda \rightarrow 0} E(\lambda)= -\infty \quad \mbox{and} \quad \lim_{\lambda \rightarrow \infty} E(\lambda) = +\infty.
\end{align}
For a coercive strain, its scale function becomes a bijection from $\mathbb{R}_+$ to $\mathbb{R}$. From a mathematical standpoint, strains that lack coerciveness suffer from a compatibility issue in their eigenstructure, as not every symmetric rank-two tensor can be represented in the form of \eqref{eq:def_E} when the range of $E$ is a strict subset of $\mathbb{R}$. For example, a symmetric tensor with eigenvalues less than $-1/2$ cannot be represented as a Green-Lagrange strain. Recognizing the importance of coerciveness, generalized strains developed after the Seth-Hill strain family all incorporate this feature \cite{Liu2024}. Based on the above discussion, we have the following results regarding the representation of a symmetric tensor as a coercive strain.
\begin{lemma}
\label{lemma:representation}
Given a symmetric rank-two tensor $\bm W$ and a coercive scale function $E$, there exists a unique set of positive scalars $w_a \in \mathbb R_{+}$ and an orthonormal basis $\tilde{\bm N}_{a}$ such that
\begin{align*}
\bm W = \sum_{a=1}^{3} E(w_a) \tilde{\bm N}_a \otimes \tilde{\bm N}_a.
\end{align*}
Moreover, the symmetric positive semi-definite tensor
\begin{align*}
\sum_{a=1}^{3} w^2_a \tilde{\bm N}_a \otimes \tilde{\bm N}_a
\end{align*}
serves as a generalized deformation tensor associated with $\bm W$.
\end{lemma}
This lemma claims that any symmetric tensor can be represented as a coercive strain, and the positive scalars can be interpreted as the stretches associated with the tensor. Importantly, this construction enables the introduction of elastic and viscous deformation measures, which in turn facilitate the formulation of more general constitutive models beyond those based on the classical Green-Naghdi kinematic assumption.

\subsection{Kinematic assumptions}
\label{sec:green_naghdi_assumption}
Kinematic assumption distinguishes elastic and viscous kinematics from the total deformation state and plays a fundamental role in shaping the constitutive theory. In this section, we discuss kinematic assumptions for viscoelasticity using basic mechanical analogs. Our focus is on the kinematic assumption of the Green-Naghdi type \cite{Liu2024,Naghdi1990}, and we generalize it through the notion of generalized strains to flexibly capture elastic and viscous deformations. This assumption forms the foundation of our model construction for multiple non-equilibrium processes. Different from the choice made in \cite{Liu2024}, where the internal variable is deformation-like and belongs to Sym(3)$_+$, the internal variable invoked here is strain-like and belongs to Sym(3). It will be demonstrated that this choice enjoys appealing attributes in both theory and computation. 

\begin{figure}[h]
\centering
\includegraphics[trim=0 370 440 0, clip,  scale=0.45]{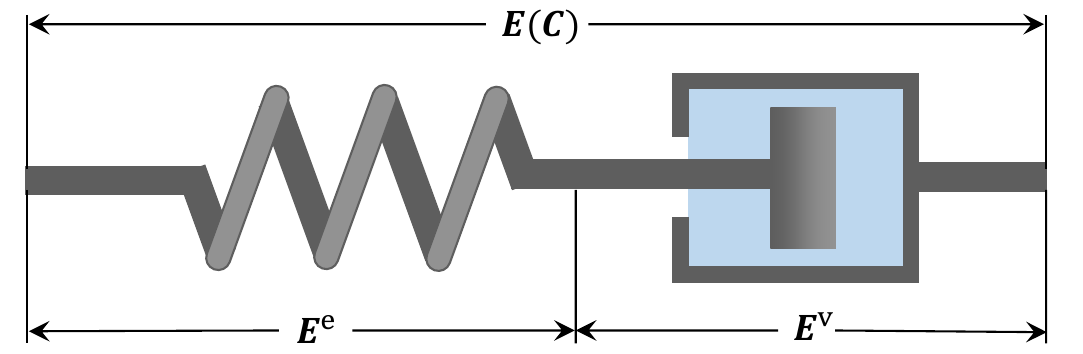}
\caption{Maxwell element.}
\label{fig:single_Maxwell_element}
\end{figure}

To begin with, we consider the Maxwell element, a conceptual device consisting of a spring and a dashpot connected in series (see Figure \ref{fig:single_Maxwell_element}). A key implication of this particular device is that the mechanical response of the spring equals that in the connected-in-series dashpot, where the spring is elastic and the dashpot obeys a viscous law. Given the total deformation of the Maxwell element, a central question lies in distinguishing the deformation of each component, as their separate kinematics directly determine the overall mechanical response. Leveraging the concept of generalized strains, we may give the kinematic assumption of the Green-Naghdi type for the Maxwell element as follows.
\begin{center}
\emph{The internal variable $\bm E^{\mathrm v} \in$ Sym(3) enters into the constitutive relation through the term $\bm E - \bm E^{\mathrm v}$.}
\end{center}
The internal variable $\bm E^{\mathrm v}$ can be conceptually interpreted as the strain in the dashpot, and the strain $\bm E$ is a generalized strain based on the spectral representation of $\bm C$. We may thus introduce
\begin{align}
\label{eq:def_Ee}
\bm E^\mathrm{e} := \bm E(\bm C) - \bm E^\mathrm{v},
\end{align}
and regard it as the strain in the spring (Figure \ref{fig:single_Maxwell_element}). Since $\bm E^\mathrm{v}$ evolves in Sym(3), the above elastic strain $\bm E^\mathrm{e}$ is also symmetric and thus admits a spectral decomposition. In the original proposal by Green and Naghdi, they used the Green-Lagrange strain for $\bm E$ \cite{Green1965}. As a non-coercive strain, the eigenvalues of $\bm E^{\mathrm e}$ do not guarantee that it can be represented as a Green-Lagrange strain, as was pointed out by Lee in \cite{Lee1969}. Thereafter, Green and Naghdi asserted that $\bm E - \bm E^{\mathrm v}$ enters into the stress response as a whole and avoided the concept of additive decomposition of the strain \cite{Naghdi1990}. This makes their approach somewhat abstract without a clear intuitive background. Due to the lack of a physical interpretation, the inelasticity theories based on the original Green-Naghdi assumption rely exclusively on the hyperelasticity of Hill's class, which has a quadratic form for the strain energy \cite{Papadopoulos1998,Schroeder2002}. Following our discussion based on Lemma \ref{lemma:representation}, it becomes clear that $\bm E^{\mathrm e}$ can always be expressed as a strain with a coercive scale function. As long as the scale function associated with $\bm E^{\mathrm e}$ is coercive, Lemma \ref{lemma:representation} guarantees the existence of elastic stretches $\lbrace \lambda_a^{\mathrm e} \rbrace$ such that
\begin{align*}
\bm E^\mathrm{e} = \sum_{a=1}^{3} E^{\mathrm e}(\lambda^\mathrm{e}_a) \hat{\bm M}_a,
\end{align*}
where $\hat{\bm M}_a$ is the self-dyads of the principal directions of $\bm E^{\mathrm e}$. The inverse of $E^{\mathrm e}$ does not always have a closed-form expression, and the values of $\lambda_a^\mathrm{e}$ are thus obtained via an iterative procedure in practice. The concept of elastic stretches naturally leads to the introduction of an elastic deformation tensor as
\begin{align}
\label{eq:def_Ce}
\bm C^\mathrm{e} := \sum_{a=1}^{3} \lambda_a^{\mathrm{e}\:2} \hat{\bm M}_a.
\end{align}
The primary benefit of introducing $\bm C^{\mathrm e}$ is that we may design the strain energy in a more flexible manner by invoking the modern theory of hyperelasticity modeling. In the later part of this work, we demonstrate this flexibility by using the eight-chain model \cite{Bischoff2001,Arruda1993} to characterize the elastic stress responses.

\begin{figure}[h]
\centering
\includegraphics[trim=0 280 440 0, clip,  scale=0.45]{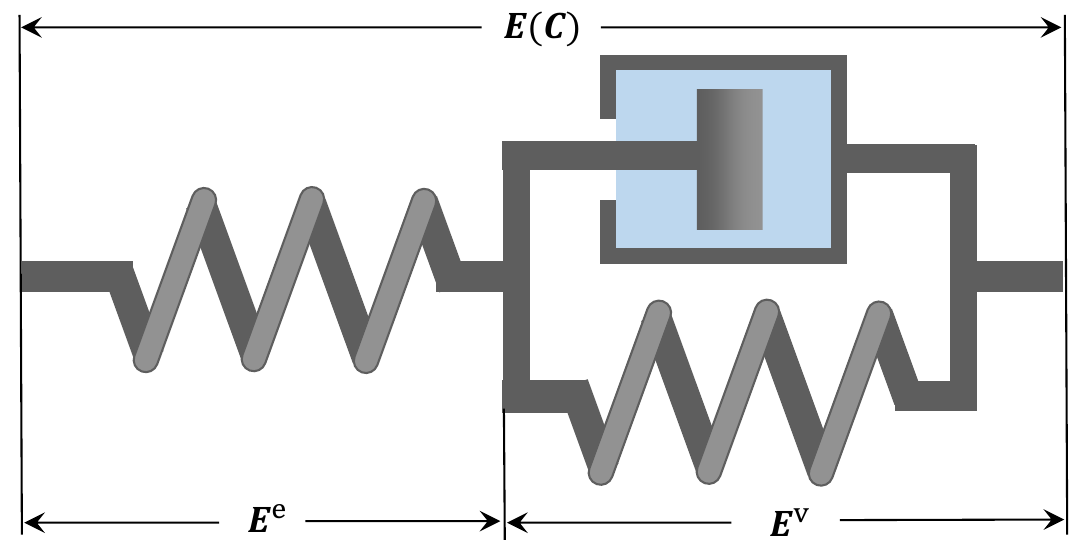}
\caption{The Poynting-Thomson model.}
\label{fig:Poynting_Thomson_model}
\end{figure}

In addition to the Maxwell element, the Poynting-Thomson model, also known as the classical Kelvin-Voigt model, is a critical rheological model for viscoelastic modeling \cite{Thomson1865}. It augments the Maxwell element by introducing an additional spring in parallel with the dashpot (Figure \ref{fig:Poynting_Thomson_model}). The resulting combination of a spring and a dashpot in parallel is referred to as the Voigt element \cite{Ferry1980}. Within the Voigt element, the spring is referred to as the \textit{non-equilibrium spring}, while the spring connected in series with the Voigt element is known as the \textit{equilibrium spring}, which is responsible for carrying the instantaneously recoverable deformation. An underlying implication of the Voigt element is that the non-equilibrium spring and the dashpot share the same deformation state characterized by the internal variable $\bm E^{\mathrm v}$. In order to describe the non-equilibrium spring using a general hyperelasticity theory, it is necessary to give a detailed characterization of the internal variable. If the scale function $E^{\mathrm v}$ associated with $\bm E^{\mathrm v}$ is coercive, Lemma \ref{lemma:representation} ensures the existence of uniquely defined viscous stretches $\lbrace \lambda^{\mathrm v}_a \rbrace$ such that
\begin{align*}
\bm E^\mathrm{v} = \sum_{a=1}^{3} E^{\mathrm v}(\lambda^\mathrm{v}_a) \bar{\bm M}_a,
\end{align*}
where $\bar{\bm M}_a$ is the self-dyads of the principal direction of $\bm E^{\mathrm v}$.  As a result, the viscous deformation tensor $\bm C^\mathrm{v}$ is defined as
\begin{align}
\label{eq:def_Cv}
\bm C^\mathrm{v} := \sum_{a=1}^{3} \lambda^\mathrm{v\:2}_a \bar{\bm M}_a,
\end{align}
which facilitates the modeling of the non-equilibrium spring by a broader class of strain energy designs.

As mentioned at the beginning of this section, the choice of the internal variable and the related kinematic assumption dictates the constitutive theory. In \cite{Liu2024}, our constitutive theory was formulated using an internal variable $\bm \Gamma \in$ Sym(3)$_+$, which is analogous to the viscous deformation tensor $\bm C^{\mathrm v}$ introduced here. That choice was motivated by earlier works, including the viscoelastic model of Simo \cite{Simo1987}, the identical polymer chain model \cite{Govindjee1992,Holzapfel1996b}, and Miehe's notion of plastic metric \cite{Miehe1998,Miehe1998a,Miehe2000}. Based on the spectral decomposition of $\bm \Gamma$, a generalized viscous strain was defined as $\bm E^{\mathrm v}(\bm \Gamma)$. The model was then constructed under a similar assumption with $\bm E(\bm C) - \bm E^{\mathrm v}(\bm \Gamma)$ entering the strain energy. In contrast to the theory of \cite{Liu2024}, where the viscous strain was introduced via a nonlinear tensorial function of the internal variable $\bm \Gamma$, the present theory directly takes the strain-like tensor as the internal variable from the outset. As will be shown in Section \ref{sec:Gen_Maxwell_model_FLV}, this choice ensures a linear flow rule when the free energy is quadratic, thus yielding a family of \textit{finite linear viscoelastic} models \cite{Liu2021b}. This attribute, however, does not generally hold for the previous theory of \cite{Liu2024}. Moreover, integrating an internal variable of Sym(3)$_+$ can be non-trivial due to the need to preserve positive definiteness \cite[p.~74]{Simo1992e}. This often necessitates formulating the evolution equation in a manner that the exponential integrator can be applied \cite{Miehe2000,Simo1992e}.

The multiplicative decomposition of the deformation gradient $\bm F = \bm F^{\mathrm e} \bm F^{\mathrm v}$ is another widely adopted kinematic assumption in modeling viscoelasticity \cite{Sidoroff1974}. This formulation introduces a conceptual intermediate configuration. The deformation gradient $\bm F^\mathrm{v}$ maps the reference configuration to the intermediate state, which is then elastically deformed to the current configuration via $\bm F^\mathrm{e}$. Under this setting, the elastic deformation tensor $\hat{\bm C}^\mathrm{e}:= \bm F^{\mathrm{e}\:\mathrm T} \bm F^\mathrm{e}$ and the viscous deformation tensor $\hat{\bm C}^\mathrm{v} := \bm F^{\mathrm{v}\:\mathrm T} \bm F^\mathrm{v}$ are then defined. These quantities characterize the deformation of the spring within the Maxwell element \cite{Reese1998} and the Voigt element  \cite{Huber2000a}, and are analogous to $\bm C^{\mathrm e}$ and $\bm C^{\mathrm v}$ defined in \eqref{eq:def_Ce} and \eqref{eq:def_Cv}. Here we used the hat symbol to distinguish the two from the deformation tensors introduced in our framework. The multiplicative decomposition found a sound micromechanical interpretation for polycrystalline materials. The inelastic deformation gradient is due to the sliding of crystal blocks along slip planes, and $\bm F^{\mathrm e}$ describes the distortion and rotation of the crystal lattice \cite{Asaro2006}. However, this physical interpretation does not generally extend to other materials, such as polymers or biological tissues, where viscous effects stem from completely different mechanisms.

A deeper issue with the multiplicative decomposition is the non-uniqueness of the intermediate configuration, since its internal variable belongs to GL(3)$_+$. For isotropic materials, this indeterminacy does not cause trouble and is usually fixed by absorbing all rotation into $\bm F^{\mathrm e}$ in practice \cite{Reese1998,LeTallec1994}. The ambiguity is exacerbated in anisotropic materials, where the evolution of material symmetry groups must be explicitly specified \cite{Ciambella2021,Sadik2024}. To mitigate the issue caused by the intermediate configuration, a reverse decomposition $\bm F = \bm F^{\mathrm v} \bm F^{\mathrm e}$ was proposed \cite{Latorre2016}, which allows formulating the anisotropic free energy completely on the reference configuration. Interestingly, it was recently demonstrated that the reverse decomposition enjoys better thermomechanical consistency than the original Sidoroff decomposition \cite{Bahreman2022}. While the multiplicative decomposition is the prevailing kinematic assumption, it is not without controversy and should not be considered the sole admissible choice for modeling general inelastic materials. This work aims to develop a kinematic decomposition that allows flexible treatment of complex rheological representations while providing effective descriptions of inelastic material behavior.

\section{Modeling viscoelasticity based on rheological representations}
\label{sec:constitutive_theory}
This study focuses on two representative rheological representations for viscoelasticity. The generalized Kelvin-Voigt model has been rarely discussed in the literature. The decomposition of generalized strains offers a robust basis for model construction, which lays the groundwork for more elaborate inelastic models composed of multiple serially connected rheological elements. Before discussing the two models in detail, we construct a thermomechanical foundation for viscoelasticity with multiple non-equilibrium processes.

Consider a material characterized by $M$ non-equilibrium processes. We assume the existence of $M$ internal variables $\lbrace \bm E^\mathrm{v}_\alpha \rbrace_{\alpha=1}^M$. Based on our discussion in Section \ref{sec:green_naghdi_assumption}, the internal variables are symmetric rank-two Lagrangian tensors and have the same invariance property as the strain $\bm E$ under superimposed rigid body motion. Under the isothermal condition, the Helmholtz free energy $\Psi$ is postulated to depend on the deformation tensor $\bm C$ and the internal variables:
\begin{align}
\label{eq:def_energy}
\Psi = \Psi(\bm C, \bm E^\mathrm{v}_1, \dots, \bm E^\mathrm{v}_M).
\end{align}
The invariance requirement ensures the objectivity of the free energy. The Clausius-Plank inequality takes the form
\begin{align*}
\mathcal{D} := \left( \bm S - 2 \frac{\partial \Psi}{\partial \bm C} \right) : \frac{1}{2} \dot{\bm C} + \sum_{\alpha=1}^{M} \bm T_{\alpha} : \dot{\bm E}^\mathrm{v}_\alpha \geq 0,
\end{align*}
where $\mathcal D$ stands for the internal dissipation, and $\bm T_\alpha$ represents the thermodynamic force conjugate to the $\alpha$-th internal variable $\bm E^{\mathrm v}_{\alpha}$, i.e.,
\begin{align}
\label{eq:def_T_alpha}
\bm T_\alpha := -\frac{\partial \Psi}{\partial \bm E^\mathrm{v}_\alpha}.
\end{align}
We adopt the second Piola-Kirchhoff stress $\bm S$ as a non-dissipative stress, meaning
\begin{align}
\label{eq:S-constitution}
\bm S = 2 \frac{\partial \Psi}{\partial \bm C}.
\end{align}
In the modeling of viscoelasticity, collective experimental observation indicates the presence of both equilibrium and non-equilibrium material responses \cite{Bergstrom1998,Wang2018,Xiang2019}. Correspondingly, the free energy and the stress can be decomposed into two parts as
\begin{align}
\label{eq:S_eq_S_neq}
\Psi = \Psi^\infty + \Psi^\mathrm{neq}.
\end{align}
The superscript ‘$\infty$' denotes the equilibrium response, while ‘neq’ refers to the dissipative, non-equilibrium response. For instance, the associated stress components are introduced as
\begin{align*}
\bm S^\infty = 2\frac{\partial \Psi^\infty}{\partial \bm C}
\quad \mbox{and} \quad
\bm S^\mathrm{neq} = 2\frac{\partial \Psi^\mathrm{neq}}{\partial\bm C}.
\end{align*}
With the relation \eqref{eq:S-constitution} for $\bm S$, the internal dissipation reduces to 
\begin{align*}
\mathcal{D} = \sum_{\alpha=1}^{M} \bm T_{\alpha} : \dot{\bm E}^\mathrm{v}_\alpha.
\end{align*}
To characterize the dissipative material behavior, a dissipation potential $\Phi(\dot{\bm E}_1^{\mathrm v}, \cdots, \dot{\bm E}_{M}^{\mathrm v})$ is introduced, which is required to be non-negative and convex. The principle of maximum entropy production has been employed as a general framework for describing the evolution of non-equilibrium systems. This principle can be traced back to the variational principle due to Onsager \cite{Onsager1931} and was generalized by Ziegler to account for continuum problems \cite{Ziegler1983}. It provides a thermodynamically motivated guideline for specifying the evolution of internal variables. Specifically, it claims that the evolution of internal variables should maximize the dissipation $\mathcal D$ while subjected to the constraint
\begin{align}
\label{eq:dissipation_constraint}
\sum_{\alpha=1}^{M} \bm T_{\alpha} : \dot{\bm E}^\mathrm{v}_\alpha = \Phi(\dot{\bm E}_1^{\mathrm v}, \cdots, \dot{\bm E}_{M}^{\mathrm v}).
\end{align}
The non-negativity of $\Phi$ ensures the satisfaction of the Clausius-Plank inequality. The above concept can be described by the constrained optimization problem,
\begin{align*}
\max_{(\delta, \dot{\bm E}^\mathrm{v}_\alpha)} \left\{ \sum_{\alpha=1}^{M}\bm T_\alpha : \dot{\bm E}_\alpha^\mathrm{v} - \delta \Big( \Phi(\dot{\bm E}_1^{\mathrm v}, \cdots, \dot{\bm E}_{M}^{\mathrm v}) - \sum_{\alpha=1}^{M}\bm T_\alpha : \dot{\bm E}^\mathrm{v}_\alpha \Big) \right\},
\end{align*}
where $\delta$ is a Lagrange multiplier. Examining the optimality condition leads to
\begin{align}
\label{eq:opt_optimality_cond}
\bm T_\alpha = \frac{\delta}{1 + \delta} \frac{\partial \Phi}{\partial \dot{\bm E}^\mathrm{v}_\alpha},
\end{align}
which suggests the thermodynamic force $\bm T_\alpha$ is orthogonal to the plane tangent to the iso-surface of the dissipation potential $\Phi(\dot{\bm E}_1^{\mathrm v}, \cdots, \dot{\bm E}_{M}^{\mathrm v}) = \Phi_0$. This condition is analogous to the normality rule in plasticity \cite{Ziegler1958}, which is typically given in the stress space by the dual dissipation function \cite{Miehe2002a,Moreau1976}. It is worth noting that the convexity of the dissipation potential guarantees that the optimization problem admits a unique solution and the resulting relation \eqref{eq:opt_optimality_cond} is well-defined. Combining the optimality condition with the definition of the thermodynamic forces in \eqref{eq:def_T_alpha}, we arrive at the general evolution equation for all internal variables as
\begin{align}
\label{eq:evo_general_alpha}
\underbrace{\frac{\partial \Psi}{\partial \bm E^\mathrm{v}_\alpha} }_{-\bm T_{\alpha}} + \frac{\delta}{1 + \delta} \frac{\partial \Phi}{\partial \dot{\bm E}^\mathrm{v}_\alpha} = \bm O, \qquad \mbox{for} \quad \alpha = 1, \cdots, M.
\end{align}
The Lagrange multiplier $\delta$ is determined from the following identity
\begin{align*}
\Phi(\dot{\bm E}_1^{\mathrm v}, \cdots, \dot{\bm E}_{M}^{\mathrm v}) = \frac{\delta}{1+\delta} \sum_{\alpha=1}^{M} \frac{\partial \Phi}{\partial \dot{\bm E}^\mathrm{v}_\alpha}:\dot{\bm E}^\mathrm{v}_\alpha,
\end{align*}
which ensures that the condition \eqref{eq:dissipation_constraint} is satisfied. The Lagrange multiplier $\delta$ obtained from above will enter into \eqref{eq:evo_general_alpha} to complete the specification of the evolution equation. Also, the factor $1+1/\delta$ reflects the homogeneity property of $\Phi$.

\begin{remark}
Alternatively, the Biot equation is often employed to derive the evolution equation \cite{Miehe2005,Kumar2016,Govindjee2019}. That equation is identical to \eqref{eq:evo_general_alpha} with the prefactor $\delta/(1+\delta)$ absorbed into the dissipation potential $\Phi$. As a result, the dissipation potential generally differs from the internal dissipation $\mathcal D$ by a scaling factor in that approach.
\end{remark}

A common choice for the dissipation potential is a quadratic form
\begin{align*}
\Phi = \sum_{\alpha=1}^{M} \dot{\bm E}^\mathrm{v}_\alpha : \mathbb V_\alpha : \dot{\bm E}^\mathrm{v}_\alpha,
\end{align*}
where each non-equilibrium process is characterized by a rank-four symmetric and positive semi-definite viscosity tensor $\mathbb V_\alpha$. Generally speaking, $\mathbb{V}_{\alpha}$ may depend on the internal variables or thermodynamic forces due to the non-Newtonian effects \cite{Lion1997a,Bergstrom1998,Kumar2016}. Moreover, the quadratic form is a homogeneous function of degree $2$, we thus have $\delta =1$ due to Euler's homogeneous function theorem. If the viscosity tensor is Newtonian and isotropic with identical deviatoric and volumetric viscosities, we have $\mathbb V_{\alpha} = \eta_{\alpha} \mathbb I$, where $\eta_{\alpha}$ is the viscosity coefficient. The evolution equation \eqref{eq:evo_general_alpha} further simplifies to
\begin{align}
\label{eq:evolution_simple_alpha}
\frac{\partial \Psi}{\partial \bm E^\mathrm{v}_\alpha} +  \eta_\alpha \dot{\bm E}^\mathrm{v}_\alpha = \bm O.
\end{align}
From an intuitive perspective, the term $\eta_\alpha \dot{\bm E}^\mathrm{v}_\alpha$ represents the viscous stress of a Newtonian dashpot, while the thermodynamic force $\bm T_\alpha$ drives the springs toward equilibrium with this viscous response. If $\bm T_\alpha$ depends linearly on the internal variables $\bm E^\mathrm{v}_\alpha$, the evolution equation reduces to a linear ordinary differential equation and admits an analytical solution via exponential integration. We refer to this type of model as \textit{finite (deformation) linear viscoelasticity}, meaning it is a finite strain model with linear evolution equations. In contrast, a nonlinear dependence necessitates numerical iterative methods for solving $\bm E^\mathrm{v}_\alpha$, characterizing \textit{nonlinear viscoelasticity}. In the remainder of this work, we adopt the above form of the dissipation potential, and the evolution equations are instantiated based on \eqref{eq:evolution_simple_alpha}. 

\begin{remark}
The constitutive theory here is fully characterized by a pair of scalar functions: the free energy $\Psi$ and the dissipation potential $\Phi$. This modeling approach falls into the category of generalized standard materials \cite{Halphen1975}. A wide range of dissipative mechanisms can be systematically treated through the design of these two potentials, including plasticity, damage, etc. \cite{Miehe2002a,Murakami2012,Peric1993}. In the context of viscoelasticity, the dissipation potential offers a rational means to introduce non-Newtonian effects \cite{Miehe2005,Kumar2016}. More recently, it is employed in data-driven constitutive modeling to provide a thermodynamically consistent structure for machine learning approaches \cite{Flaschel2023,Huang2022}.
\end{remark}

Under the isothermal condition, the body reaches the thermodynamic equilibrium state if there is no more change of the internal variables, or no further changes take place in the dashpots. A rigorous definition for this notion is given as follows.
\begin{definition}
For a viscoelastic material characterized by $M$ internal variables $\lbrace \bm E^{\mathrm v}_{\alpha}\rbrace_{\alpha=1}^{M}$, it is in the equilibrium state if $\dot{\bm E}^{\mathrm v}_{\alpha} = \bm O$ for $\alpha=1,\cdots, M$.
\end{definition}
We use the notation $|_{\mathrm{eq}}$ to indicate quantities evaluated at the equilibrium state. A direct consequence of \eqref{eq:evolution_simple_alpha} is that 
\begin{align}
\label{eq:T_alpha_relax_cond}
\bm T_{\alpha} |_{\mathrm{eq}} = \bm O, \quad \mbox{for} \quad \alpha =1, \cdots, M.
\end{align}
This property characterizes the state under thermodynamic equilibrium. It is not only physically intuitive but also mathematically relevant for well-behaved models \cite{Liu2021b}.

\subsection{The generalized Maxwell model}
\label{sec:Gen_Maxwell_model}
The generalized Maxwell model is one widely adopted rheological framework for viscoelastic modeling. It comprises an elastic spring arranged in parallel with multiple Maxwell elements (Figure \ref{fig:gen_Maxwell_model}). Conceptually, all parallel branches in the generalized Maxwell model undergo the same instantaneous deformation, while the total stress is obtained by summing the contributions from each branch. If there is only one non-equilibrium branch (i.e., $M=1$), it is also referred to as the Zener model. In this work, we extend the discussion made above to construct a constitutive theory based on the generalized Maxwell representation. A kinematic assumption is made as follows.
\begin{assumption}
\label{as:gen_maxwell}
The free energy \eqref{eq:def_energy} for the generalized Maxwell model with $M$ parallel branches can be represented as
\begin{align}
\label{eq:free_energy_maxwell}
\Psi(\bm C, \bm E^\mathrm{v}_1, \bm E^\mathrm{v}_2, \dots, \bm E^\mathrm{v}_M)= \Psi^\infty(\bm C)+\Psi^\mathrm{neq}(\bm C, \bm E^\mathrm{v}_1,\bm E^\mathrm{v}_2,\dots,\bm E^\mathrm{v}_M) \quad \mbox{with} \quad \Psi^\mathrm{neq} =\sum_{\alpha=1}^{M}\Psi_\alpha^\mathrm{neq}(\bm C, \bm E^\mathrm{v}_\alpha).
\end{align}
The internal variables enter into the free energy $\Psi^{\mathrm{neq}}$ through the terms $\bm E^{\mathrm e}_{\alpha} := \bm E_{\alpha}(\bm C) - \bm E^{\mathrm v}_{\alpha}$.
\end{assumption}

\begin{figure}[h]
\centering
\includegraphics[trim=180 120 180 70, clip,  scale=0.5]{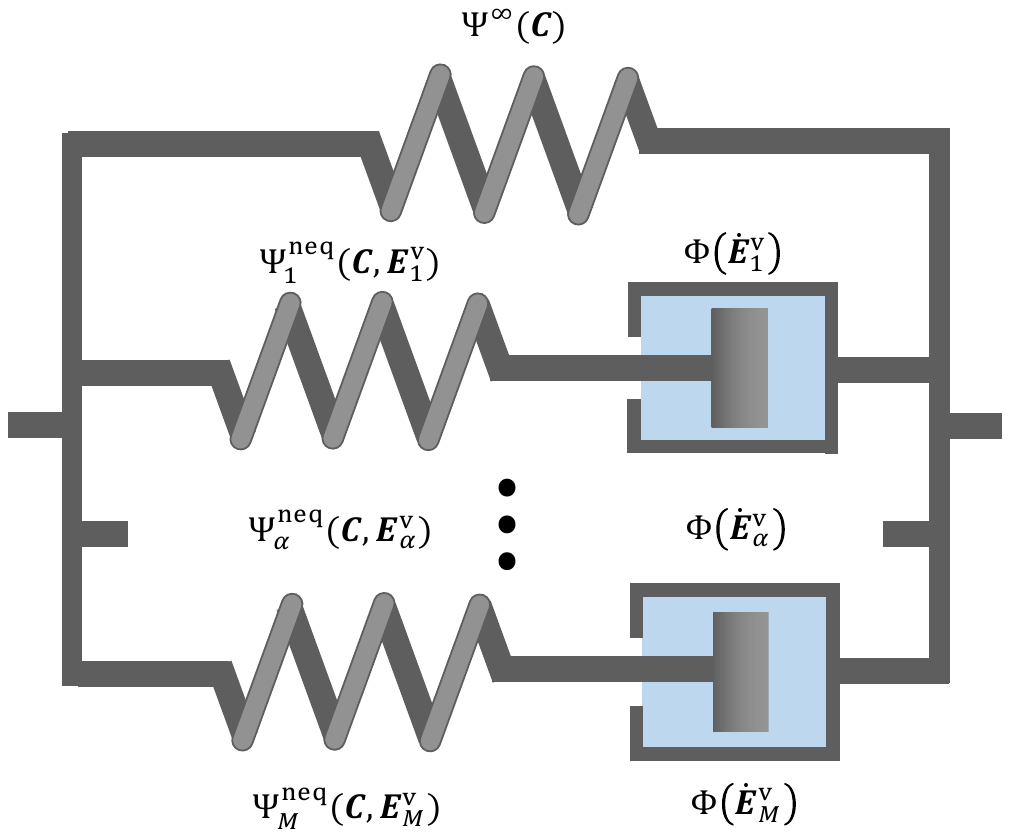}
\caption{The generalized Maxwell model.}
\label{fig:gen_Maxwell_model}
\end{figure}

In the above assumption, $\bm E^{\mathrm e}_\alpha$ is conceptually the strain of the spring within the $\alpha$-th Maxwell element. Consequently, we may express the non-equilibrium energies as
\begin{align*}
\Psi_\alpha^\mathrm{neq}(\bm C, \bm E^\mathrm{v}_\alpha) = \Psi_\alpha^\mathrm{neq}(\bm E^{\mathrm e}_{\alpha}), \quad \mbox{for} \quad \alpha = 1, \cdots, M.
\end{align*}
For notational simplicity, we do not introduce separate symbols for $\Psi_\alpha^\mathrm{neq}$, the energy stored in the $\alpha$-th Maxwell element, to explicitly indicate its dependency on $\bm E^{\mathrm e}_{\alpha}$, with the understanding that the argument of the function clarifies the relevant dependence. The non-equilibrium part of the second Piola-Kirchhoff stress can be explicitly represented as a summation of contributions from all Maxwell elements,
\begin{align*}
\bm S^\mathrm{neq} = \sum_{\alpha=1}^{M} \bm S^\mathrm{neq}_\alpha
\quad\mbox{with}\quad
\bm S^\mathrm{neq}_\alpha:= 2\frac{\partial \Psi^\mathrm{neq}_\alpha}{\partial \bm C}.
\end{align*}
Due to the form of the free energy \eqref{eq:free_energy_maxwell}, the total stress is the sum of the equilibrium and non-equilibrium contributions, i.e.,
\begin{align}
\label{eq:S_decomp_Maxwell}
\bm S = \bm S^{\infty} + \bm S^{\mathrm{neq}}.
\end{align}
The form of $\Psi^{\infty}$ follows the general design principle for hyperelasticity, and our focus is placed on the design of the non-equilibrium part of the free energy. Due to \eqref{eq:free_energy_maxwell} in the kinematic assumption, we immediately have the following identity,
\begin{align}
\label{eq:T_alpha_identity}
\bm T_{\alpha} := - \frac{\partial \Psi}{\partial \bm E^{\mathrm v}_{\alpha}} = - \frac{\partial \Psi_\alpha^\mathrm{neq}}{\partial \bm E^{\mathrm v}_{\alpha}} = \frac{\partial \Psi_\alpha^\mathrm{neq}}{\partial \bm E_{\alpha}}.
\end{align}
Consequently, the non-equilibrium part of the second Piola-Kirchhoff stress can be represented as
\begin{align}
\label{eq:Maxwell_S_neq}
\bm S^\mathrm{neq}_\alpha = \bm T_\alpha:\mathbb Q_\alpha
\quad\mbox{with}\quad
\mathbb Q_\alpha:= 2\frac{\partial \bm E_\alpha}{\partial \bm C}.
\end{align}
Recalling \eqref{eq:T_alpha_relax_cond}, we immediately have $\bm S^{\mathrm{neq}}_\alpha |_{\mathrm{eq}} = \bm S^{\mathrm{neq}} |_{\mathrm{eq}} = \bm O$.

\begin{remark}
As an alternative to Assumption \ref{as:gen_maxwell}, one may invoke the Sidoroff decomposition $\bm F = \bm F^{\mathrm e}_\alpha \bm F^{\mathrm v}_\alpha$ or the reverse decomposition $\bm F = \bm F^{\mathrm v}_\alpha \bm F^{\mathrm e}_\alpha$ for each branch. Correspondingly, the non-equilibrium free energy in the $\alpha$-th branch is represented as a function of $\hat{\bm C}^{\mathrm e}_\alpha := \bm F^{\mathrm{e}\:\mathrm T}_\alpha \bm F^\mathrm{e}_\alpha$. This strategy, used in combination with the generalized Maxwell model, is prevailing in contemporary finite viscoelastic models \cite{Reese1998,Gouhier2024}.
\end{remark}

\subsubsection{Finite linear viscoelasticity}
\label{sec:Gen_Maxwell_model_FLV}
A class of viscoelastic models is defined by prescribing the non-equilibrium free energies in quadratic forms as
\begin{align*}
\Psi^\mathrm{neq}_\alpha (\bm C, \bm E^\mathrm{v}_\alpha) = \frac{1}{2} \mu_\alpha \left \lvert \bm E^\mathrm{e}_\alpha \right \rvert^2, \quad \mbox{for} \quad \alpha = 1, \cdots, M,
\end{align*}
where $\mu_\alpha$ is the shear modulus of the $\alpha$-th Maxwell element, and $\lvert (\cdot) \rvert$ stands for the norm of a rank-two tensor $(\cdot)$. With the quadratic energy, the relation \eqref{eq:T_alpha_identity} leads to $\bm T_\alpha =\mu_\alpha\bm E^\mathrm{e}_\alpha = \mu_{\alpha} (\bm E_\alpha - \bm E^\mathrm{v}_\alpha)$. Based on the general evolution equation \eqref{eq:evolution_simple_alpha}, the evolution equation for the $\alpha$-th internal variable can be instantiated as
\begin{align}
\label{eq:FLV_Maxwell_evolution}
-\mu_\alpha (\bm E_\alpha - \bm E^\mathrm{v}_\alpha) + \eta_{\alpha} \dot{\bm E}^\mathrm{v}_\alpha = \bm O.
\end{align}
We may introduce $\tau_{\alpha} := \eta_{\alpha}/\mu_{\alpha}$ as the relaxation time for the $\alpha$-th process. The evolution equation can be rewritten as
\begin{align}
\label{eq:maxwell_flv_evo_eqn}
\dot{\bm E}^{\mathrm v}_{\alpha} + \frac{1}{\tau_{\alpha}} \bm E^{\mathrm v}_{\alpha} = \frac{1}{\tau_{\alpha}} \bm E_{\alpha}.
\end{align}
Noticing that \eqref{eq:FLV_Maxwell_evolution} suggests $\dot{\bm E}^{\mathrm v}_{\alpha} = \bm T_{\alpha} / \eta_{\alpha}$, we may alternatively obtain an evolution equation for $\bm T_\alpha$ by taking a time derivative on \eqref{eq:maxwell_flv_evo_eqn}:
\begin{align}
\label{eq:maxwell_flv_T_evo_eqn}
\dot{\bm T}_\alpha + \frac{1}{\tau_\alpha} \bm T_\alpha = \mu_\alpha \dot{\bm E}_\alpha,
\end{align}
where we have assumed that the shear modulus is time independent. Given the initial condition $\bm E^{\mathrm v}_{\alpha}|_{t=0} = \bm O$ and $\bm T_{\alpha}|_{t=0} = \bm O$, their solutions can be expressed in terms of the hereditary integrals
\begin{align}
\label{eq:maxwell_flv_evo_eqn_T}
\bm E^{\mathrm v}_\alpha = \frac{1}{\tau_{\alpha}} \int_{0}^{t} \exp\left(-\frac{t - s}{\tau_\alpha}\right) \bm E_{\alpha}(s) ds \quad \mbox{and} \quad
\bm T_\alpha = \mu_\alpha \int_{0}^{t} \exp\left(-\frac{t - s}{\tau_\alpha}\right) \dot{\bm E}_{\alpha}(s) ds.
\end{align}
Sometimes $\bm T_{\alpha}$ is referred to as the stress-type internal variable, and integrating $\bm T_{\alpha}$ can be convenient, as it directly leads to the second Piola-Kirchhoff stress. The convolution representations inspire the constitutive integration scheme for this type of model, to be detailed in Section \ref{sec:numerical_formulation_FLV_Maxwell}. Following \eqref{eq:S_decomp_Maxwell}, the total stress is represented as
\begin{align}
\label{eq:maxwell_flv_convolution}
\bm S = \bm S^{\infty} + \sum_{\alpha=1}^{M} \left( \bm T_{\alpha} : \mathbb Q_{\alpha} \right) = \bm S^{\infty} + \sum_{\alpha=1}^{M} \left\lbrace \mu_\alpha \int_{0}^{t} \exp\left(-\frac{t - s}{\tau_\alpha}\right) \dot{\bm E}_{\alpha}(s) ds : \mathbb Q_{\alpha} \right\rbrace,
\end{align}
which bears a resemblance to the convolution representation of the stress based on the relaxation function \cite[Chapter~10]{Simo2006}. We may examine the convolution representation \eqref{eq:maxwell_flv_convolution} by considering the strain history as $\bm E_{\alpha} = \bm E_{\alpha}(\bm C_0) H(t)$, where $\bm C_0$ is a constant deformation tensor and $H(t)$ is the Heaviside function. Physically, this setup corresponds to a relaxation test, and the stress can be obtained as
\begin{align*}
\bm S = \bm S^{\infty}(\bm C_0) + \sum_{\alpha=1}^{M} \Big( \mu_\alpha \exp\left(-t/\tau_\alpha\right) \bm E_{\alpha}(\bm C_0) : \mathbb Q_{\alpha} \Big).
\end{align*}
The non-equilibrium part of the stress is represented in terms of the Prony series, which has discrete spectra \cite{Lakes2009}. The instantaneous response of the material is obtained by letting $t \rightarrow 0_+$, and we have 
\begin{align*}
\bm S = \bm S^{\infty}(\bm C_0) + \sum_{\alpha=1}^{M} \Big( \mu_\alpha \bm E_{\alpha}(\bm C_0) : \mathbb Q_{\alpha} \Big).
\end{align*}
In contrast, the long-time response as $t \rightarrow \infty$ reduces to $\bm S = \bm S^{\infty}(\bm C_0)$, which can be interpreted as the equilibrium state. The same response is also directly recovered from \eqref{eq:maxwell_flv_convolution} under quasi-static loading conditions, where the strain rates $\dot{\bm E}_{\alpha}$ are negligible.

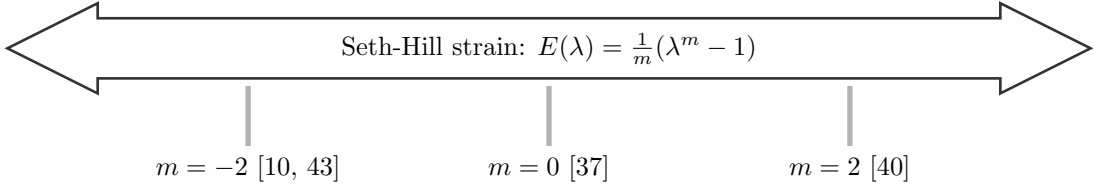
\begin{figure}[htbp]
\centering
\begin{tikzpicture}

  \def\bodylength{12}     
  \def\bodyheight{0.8}     
  \def\headlength{1.2}     
  \def\totalheight{1.2}      

  \def\xstart{0}
  \def\xend{\xstart + \bodylength}

  \draw[black!80, line width=1pt]
    (\xstart - \headlength, 0) -- 
    (\xstart, -\totalheight/2) --
    (\xstart, -\bodyheight/2) --
    (\xend, -\bodyheight/2) --
    (\xend, -\totalheight/2) --
    (\xend + \headlength, 0) --  
    (\xend, \totalheight/2) --
    (\xend, \bodyheight/2) --
    (\xstart, \bodyheight/2) --
    (\xstart, \totalheight/2) -- cycle;

  \node at (\xstart + \bodylength/2, 0.0) {Seth-Hill strain: $E(\lambda) = \frac{1}{m}(\lambda^m - 1)$};
  
  \foreach \x/\text in {2/{$m=-2$ \cite{Green1946,Lubliner1985}}, 6/{$m=0$ \cite{Miehe2000}}, 10/{$m=2$ \cite{Simo1987}}} {
    \draw[black!30, line width=2pt] (\x, -1.3) -- (\x, -0.5);
    \node[below] at (\x, -1.3) {\text};
  }

\end{tikzpicture}
\caption{Finite linear viscoelasticity based on the Seth-Hill strain, with classical models recovered as special cases.}
  \label{fig:FLV-maxwell-connection}
\end{figure}

\begin{remark}
In the finite linear viscoelastic model, the non-equilibrium energies are expressed in the quadratic form of the strain, and the strain does not need to be coercive. In particular, we may adopt the Seth-Hill strain family. The model of Green and Tobolsky \cite{Green1946,Lubliner1985} is recovered by adopting the Euler-Almansi strain, while the Green-Lagrange strain leads to the model of Simo \cite{Simo1987}. A model similar to that of Miehe and Keck \cite{Miehe2000} can be obtained by using the Hencky strain (Figure \ref{fig:FLV-maxwell-connection}). This demonstrates that, even within the framework of finite linear viscoelasticity, a broad spectrum of models can be constructed by varying the choice of strain, including the strain type and associated parameters. In particular, the strain can be selected and calibrated based on experimental data. In our experience, this endows the models with strong capabilities for material characterization, even when using linear evolution equations.
\end{remark}

\begin{remark}
In \cite{Liu2024}, the generalized Maxwell model is constructed by assuming that $\bm E_{\alpha}(\bm C) - \bm E^{\mathrm v}_{\alpha}(\bm \Gamma_\alpha)$ enters into the free energy, in which $\bm \Gamma_\alpha \in$ Sym(3)$_+$ are the primitive internal variables. The resulting viscous flow rule is nonlinear in general, even for quadratic free energies. A linear evolution equation is recovered only in the special case where the strain measure is of the Green–Lagrange type and the non-equilibrium free energy adopts a quadratic form.
\end{remark}

\subsubsection{Nonlinear viscoelasticity}
\label{sec:Gen_Maxwell_model_FV}
Given the elastic strain $\bm E^\mathrm{e}_\alpha$ associated with the $\alpha$-th Maxwell element, we introduce the corresponding elastic deformation tensor $\bm C^\mathrm{e}_\alpha$ by invoking Lemma \ref{lemma:representation}. To simplify our discussion, we choose the scale function of $\bm E^{\mathrm e}_\alpha$ to be identical to that of $\bm E_\alpha$ in this study. This choice not only simplifies the theoretical development but also reduces the number of parameters involved in model calibration. In specific, the elastic deformation tensor is constructed through the spectral representation of $\bm E^{\mathrm e}_{\alpha}$ as
\begin{align*}
\bm C^\mathrm{e}_\alpha := \sum_{a=1}^{3} \lambda^\mathrm{e\:2}_{\alpha\:a} \hat{\bm M}_{\alpha\:a}
\quad\mbox{with}\quad
\lambda^\mathrm{e}_{\alpha\:a} = E^{-1}_{\alpha}(E^\mathrm{e}_{\alpha\:a}),
\end{align*}
where $E^\mathrm{e}_{\alpha\:a}$ and $\hat{\bm M}_{\alpha\:a}$ are the principal values of the corresponding self-dyads of principal directions, and $E^{-1}_\alpha$ denotes the inverse of the scale function of the strain $\bm E_\alpha$. The non-equilibrium part of the free energy is then formulated as $\Psi^\mathrm{neq}_\alpha (\bm C, \bm E^\mathrm{v}_\alpha) = \Psi_\alpha^\mathrm{neq}(\bm C^\mathrm{e}_\alpha)$. Recalling the identity \eqref{eq:T_alpha_identity}, the thermodynamic force is given by
\begin{align}
\label{eq:FV_Maxwell_T_alpha}
\bm T_\alpha = - \frac{\partial \Psi_\alpha^\mathrm{neq}}{\partial \bm E^{\mathrm v}_{\alpha}} = \bm S^\mathrm{e}_\alpha:\mathbb Q^\mathrm{e\:-1}_\alpha
\quad\mbox{with}\quad
\bm S^\mathrm{e}_\alpha := 2\frac{\partial \Psi^\mathrm{neq}_\alpha}{\partial \bm C^\mathrm{e}_\alpha}
\quad \mbox{and} \quad
\mathbb Q^\mathrm{e\:-1}_\alpha := \frac12 \frac{\partial \bm C^\mathrm{e}_\alpha}{\partial \bm E^\mathrm{e}_\alpha}.
\end{align}
Substituting \eqref{eq:FV_Maxwell_T_alpha} into the general evolution equation \eqref{eq:evolution_simple_alpha} yields
\begin{align}
\label{eq:FV_Maxwell_evolution}
-\bm S^\mathrm{e}_\alpha : \mathbb Q^{\mathrm e\:-1}_\alpha + \eta_\alpha \dot{\bm E}^\mathrm{v}_\alpha = \bm O, \quad \mbox{for} \quad \alpha = 1, \cdots, M.
\end{align}
Notice that $\bm T_\alpha$ depends on $\bm E^{\mathrm v}_{\alpha}$ in a generally nonlinear fashion. It therefore demands a dedicated integration strategy, which is discussed in Section \ref{sec:numerical_formulation_FV_Maxwell}.

\paragraph{\textbf{Example}}
The introduction of elastic deformation tensors enables us to go beyond the quadratic form in the design of $\Psi_\alpha^\mathrm{neq}$. Here, we provide the free energy using the eight-chain model \cite{Bischoff2001,Arruda1993} for both equilibrium and non-equilibrium parts. This demonstrates the flexibility of the proposed framework in incorporating micromechanically motivated strain energy models. The free energies are given by
\begin{align}
\label{eq:example1_energy_eq}
\Psi^\infty(\bm C) =& \mu^\infty N^\infty \left( \lambda^\infty \mathfrak L^{-1}(\lambda^\infty) + \ln \frac{\mathfrak L^{-1}(\lambda^\infty)}{\sinh(\mathfrak{L}^{-1}(\lambda^\infty))}\right) 
- \frac{\mu^\infty \sqrt{N^\infty}}{3}\mathfrak{L}^{-1}(1/\sqrt{N^\infty}) \ln (J), \displaybreak[2] \\
\label{eq:example1_energy_neq}
\Psi^\mathrm{neq}_\alpha(\bm C^\mathrm{e}_\alpha) =& \mu_\alpha N_\alpha \left( \lambda^\mathrm{neq}_\alpha \mathfrak L^{-1}(\lambda^\mathrm{neq}_\alpha) + \ln \frac{\mathfrak L^{-1}(\lambda^\mathrm{neq}_\alpha)}{\sinh(\mathfrak{L}^{-1}(\lambda^\mathrm{neq}_\alpha))}\right) 
- \frac{\mu_\alpha \sqrt{N_\alpha}}{3}\mathfrak{L}^{-1}(1/\sqrt{N_\alpha}) \ln (J^\mathrm{e}_\alpha).
\end{align}
In the above, $\mu^\infty$ ($\mu_\alpha$) is the shear modulus for the (non-)equilibrium part; $N^\infty$ ($N_\alpha$) represents the number of chain segments for the (non-)equilibrium part; $\mathfrak{L}^{-1}$ stands for the inverse of the Langevin function $\mathfrak{L}(x) := \coth x - 1/x$, which accounts for the finite chain extensibility; the volume ratio of the $\alpha$-th non-equilibrium component is given by $J^\mathrm{e}_\alpha := \sqrt{\det \bm C^\mathrm{e}_\alpha}$; the relative average chain stretches for the equilibrium part $\lambda^\infty$ and the non-equilibrium part $\lambda^\mathrm{neq}_\alpha$ are defined as
\begin{align*}
\lambda^\infty := \sqrt{\frac{\bm C : \bm I}{3 N^\infty}}
\quad \text{and} \quad
\lambda^\mathrm{neq}_\alpha := \sqrt{\frac{\bm C^\mathrm{e}_\alpha : \bm I}{3 N_\alpha}}.
\end{align*}
The corresponding stresses are expressed as
\begin{align*}
\bm S^\infty =& \frac{\mu^\infty}{3\lambda^\infty} \mathfrak{L}^{-1}(\lambda^\infty)\bm I - \frac{\mu^\infty \sqrt{N^\infty}}{3} \mathfrak{L}^{-1}(1/\sqrt{N^\infty}) \bm C^{-1}, \displaybreak[2] \\
\bm S^\mathrm{e}_\alpha =& \frac{\mu_\alpha}{3\lambda^\mathrm{neq}_\alpha} \mathfrak{L}^{-1}(\lambda^\mathrm{neq}_\alpha)\bm I - \frac{\mu_\alpha \sqrt{N_\alpha}}{3}\mathfrak{L}^{-1}(1/\sqrt{N_\alpha}) \bm C^{\mathrm{e}\:-1}_\alpha.
\end{align*}
The thermodynamic force $\bm T_{\alpha}$ follows from $\bm S^\mathrm{e}_\alpha$ via \eqref{eq:FV_Maxwell_T_alpha} and is invoked to calculate the non-equilibrium stress $\bm S^{\mathrm{neq}}_\alpha$ according to \eqref{eq:Maxwell_S_neq}.

\subsection{The generalized Kelvin-Voigt model}
\label{sec:Gen_KV_model}
The classical Kelvin-Voigt model \cite{Thomson1865,Voigt1892}, also known as the Poynting–Thomson model, consists of a spring connected with a Voigt element in series. Together with the Zener model, the two are canonical three-parameter models commonly used to characterize viscoelastic material behaviors. In a generalized form, multiple Voigt elements are connected in series with a single equilibrium spring (Figure \ref{fig:gen_KV_model}), introducing multiple retardation mechanisms. In this study, we propose a Green-Naghdi-type kinematic assumption to formulate the generalized Kelvin-Voigt model at finite deformation. It is important to note that the separation of inelastic deformation in the generalized Kelvin-Voigt model differs from that in the generalized Maxwell model. In specific, the system consists of $M$ Voigt elements connected in series, with the deformation of the $\alpha$-th Voigt element characterized by an independent internal variable $\bm E^\mathrm{v}_\alpha$. To accommodate this configuration, we introduce the following kinematic assumption as the outset of our derivation.

\begin{figure}[h]
\centering
\includegraphics[trim=20 260 150 100, clip,  scale=0.5]{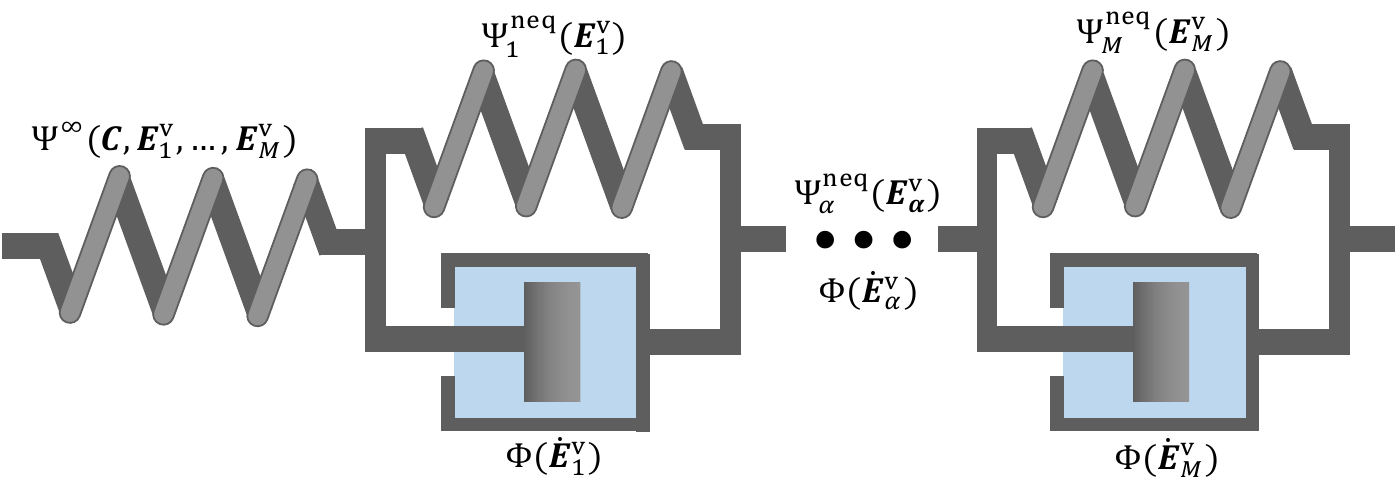}
\caption{The generalized Kelvin-Voigt model.}
\label{fig:gen_KV_model}
\end{figure}

\begin{assumption}
\label{as:gen_kelvin_voigt}
The free energy \eqref{eq:def_energy} for the generalized Kelvin-Voigt model with $M$ Voigt elements can be represented as
\begin{align}
\label{eq:strain_energy_kelvin_voigt}
\Psi(\bm C, \bm E^\mathrm{v}_1, \dots, \bm E^\mathrm{v}_M) = \Psi^\infty(\bm C,\bm E^\mathrm{v}_1, \dots, \bm E^\mathrm{v}_M ) + \Psi^\mathrm{neq}(\bm E^\mathrm{v}_1, \dots, \bm E^\mathrm{v}_M ) 
\end{align}
with
\begin{align*}
\Psi^\mathrm{neq}(\bm E^\mathrm{v}_1, \dots, \bm E^\mathrm{v}_M ) = \sum_{\alpha=1}^{M} \Psi^\mathrm{neq}_\alpha(\bm E^\mathrm{v}_\alpha).
\end{align*}
The internal variables enter into $\Psi^{\infty}$ through 
\begin{align}
\label{eq:kelvin_voigt_strain_decomp}
\bm E^\mathrm{e} := \bm E(\bm C) - \sum_{\alpha=1}^{M} \bm E^\mathrm{v}_\alpha.
\end{align}
\end{assumption}

In the above, $\bm E^{\mathrm e}$ stands for the elastic strain associated with the equilibrium spring. With this interpretation, the equilibrium energy $\Psi^{\infty}$ is expressed as
\begin{align*}
\Psi^\infty(\bm C,\bm E^\mathrm{v}_1, \dots, \bm E^\mathrm{v}_M ) = \Psi^{\infty}(\bm E^{\mathrm e}).
\end{align*}
Since $\Psi^{\mathrm{neq}}$ does not depend on $\bm C$ in Assumption \ref{as:gen_kelvin_voigt}, the second Piola-Kirchhoff stress is completely determined from the equilibrium spring, i.e., $\bm S = \bm S^\infty$. This marks a striking distinction between the two rheological representations, which arises fundamentally from their different configurations. Based on the relation \eqref{eq:def_T_alpha}, the thermodynamic force $\bm T_{\alpha}$ associated with the $\alpha$-th Voigt element becomes
\begin{align}
\label{eq:KV_T}
\bm T_\alpha := -\frac{\partial \Psi}{\partial \bm E^\mathrm{v}_\alpha} = \bm T^\infty - \bm T^\mathrm{neq}_{\alpha},
\quad \mbox{where} \quad
\bm T^\infty := -\frac{\partial \Psi^\infty}{\partial \bm E^\mathrm{v}_\alpha},
\quad
\bm T^\mathrm{neq}_{\alpha}:= \frac{\partial \Psi^\mathrm{neq}_\alpha}{\partial \bm E^\mathrm{v}_\alpha}.
\end{align}
This indicates that the thermodynamic force $\bm T_\alpha$ incorporates contributions from both the equilibrium spring and the $\alpha$-th Voigt element. Moreover, the strain decomposition \eqref{eq:kelvin_voigt_strain_decomp} implies that $\bm T^{\infty} = \partial \Psi^{\infty}/\partial \bm E$. Consequently, the second Piola-Kirchhoff stress can be represented as 
\begin{align*}
\bm S = \bm S^{\infty} = \bm T^\infty:\mathbb Q.
\end{align*}
The evolution equation \eqref{eq:evolution_simple_alpha} is instantiated as
\begin{align}
\label{eq:KV_evolution}
- \bm T^\infty + \bm T^\mathrm{neq}_\alpha + \eta_\alpha \dot{\bm E}^\mathrm{v}_\alpha = \bm O, \quad \mbox{for} \quad \alpha = 1, \cdots, M.
\end{align}
Here, the term $\eta_\alpha \dot{\bm E}^\mathrm{v}_\alpha$ and $\bm T^\mathrm{neq}_\alpha$ represent the viscous and elastic responses of the dashpot and spring in the $\alpha$-th Voigt element, respectively, while $\bm T^\infty$ corresponds to the response of the equilibrium spring. The evolution equation thus reflects the dynamic balance between the equilibrium spring and each Voigt element. In particular, at the equilibrium state, we have $\bm T^{\infty}|_{\mathrm{eq}} = \bm T^{\mathrm{neq}}_\alpha |_{\mathrm{eq}}$ for $\alpha=1,\cdots, M$ due to \eqref{eq:T_alpha_relax_cond}. This condition signifies that, in the long-time limit, all thermodynamic forces are balanced among the springs, representing an equilibrium state in the rheological model.

\begin{remark}
In the setting of the multiplicative decomposition, this rheological representation has been discussed only for the Poynting–Thomson model (i.e., $M=1$) \cite{Cai2024,Huber2000a,Laiarinandrasana2003}. In this study, one focus is the rheological model with multiple devices connected in series, such as the generalized Kelvin-Voigt model. If one invokes the Sidoroff decomposition for such models, the decomposition will be $\bm F = \bm F^{\mathrm e} \bm F^{\mathrm v}_1 \cdots \bm F^{\mathrm v}_{M}$, assuming there are $M$ non-equilibrium processes. Correspondingly, there will be $M$ intermediate configurations (Figure \ref{fig:multiplicative_GKV}) whose kinematics remain to be specified. The situation becomes even more intricate for anisotropic materials, as the treatment of the symmetry groups across multiple intermediate configurations can be challenging in theory and computation. Despite these complexities, the computational strategies developed in the following section remain applicable to this class of models.
\end{remark}

\begin{figure}[h]
\centering
\includegraphics[trim=0 280 210 0, clip,  scale=0.5]{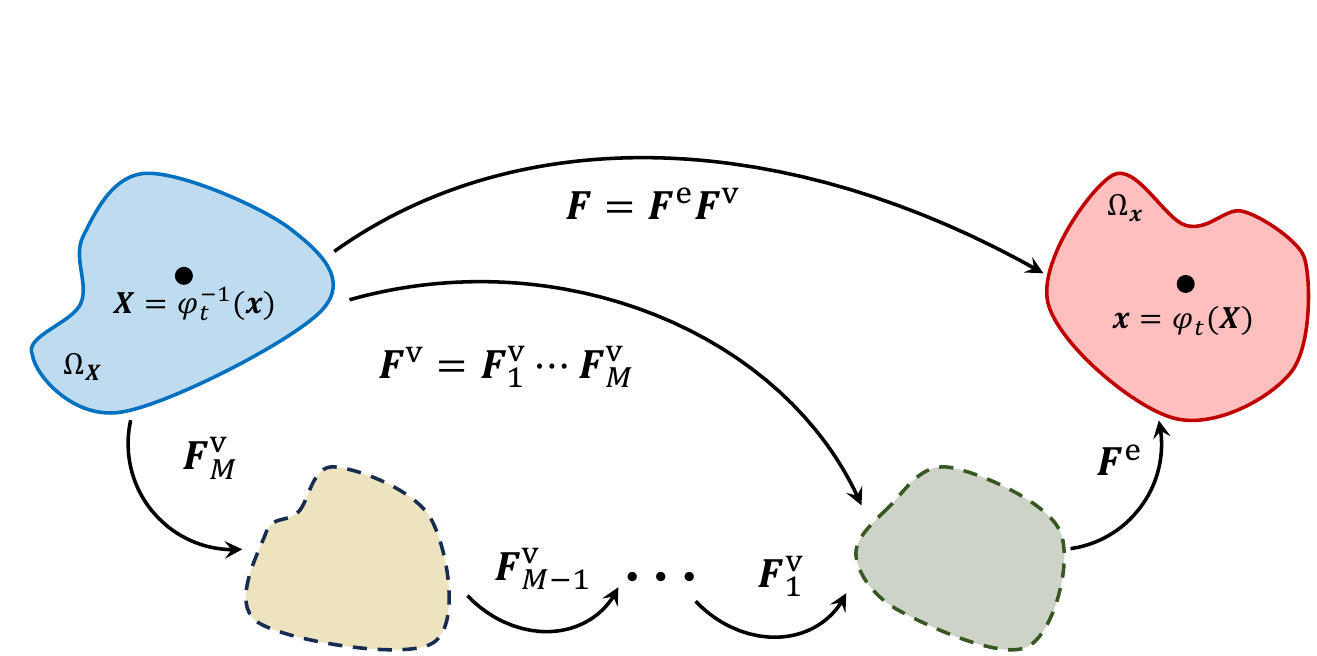}
\caption{Representation of the configurations under multiplicative decomposition for the generalized Kelvin-Voigt model with $M$ non-equilibrium processes.}
\label{fig:multiplicative_GKV}
\end{figure}

\subsubsection{Finite linear viscoelasticity}
\label{sec:Gen_KV_model_FLV}
A particular class of models is defined by the following strain energy functions,
\begin{gather*}
\Psi^\infty= \frac12\mu^\infty \lvert \bm E^\mathrm{e} \rvert^2
\quad \mbox{and} \quad
\Psi^\mathrm{neq}_\alpha (\bm E^\mathrm{v}_\alpha) = \frac12 \mu_\alpha\lvert \bm E^\mathrm{v}_\alpha\rvert^2,
\end{gather*}
where $\mu^\infty$ is the shear modulus of the equilibrium spring, and $\mu_\alpha$ is the shear modulus of the spring within the $\alpha$-th Voigt element. The thermodynamic force $\bm T_{\alpha}= \bm T^\infty - \bm T^\mathrm{neq}_{\alpha}$ is given by
\begin{align*}
\bm T^\infty= \mu^\infty \bm E^\mathrm{e}
\quad \mbox{and} \quad
\bm T^\mathrm{neq}_\alpha  = \mu_\alpha \bm E^\mathrm{v}_\alpha.
\end{align*}
Consequently, the evolution equation \eqref{eq:KV_evolution} for the $\alpha$-th Voigt element is instantiated as
\begin{align}
\label{eq:FLV_KV_evolution_1}
-\mu^\infty\bm E^\mathrm{e} +  \mu_\alpha \bm E^\mathrm{v}_\alpha + \eta_\alpha \dot{\bm E}^\mathrm{v}_\alpha  = \bm O.
\end{align}
The above evolution equation is linear with respect to the internal variables. In contrast to the generalized Maxwell representation, maintaining this linearity requires both the equilibrium and non-equilibrium strain energies to be quadratic. This requirement arises due ot the explicit presence of $\bm T^{\infty}$ in \eqref{eq:KV_evolution}. Moreover, we mention that the $M$ evolution equations are coupled, because the elastic strain $\bm E^{\mathrm e}$ involves contributions from all internal variables. The numerical integration scheme is detailed in Section~\ref{sec:numerical_formulation_FLV_KV}.

Recall that $\tau_{\alpha} := \eta_{\alpha}/\mu_{\alpha}$ reflects an intrinsic time scale for the material, and it is commonly referred to as the retardation time for the generalized Kelvin-Voigt model \cite{Ferry1980}. Henceforth, we do not distinguish between relaxation and retardation times, as both share the same formal definition. Dividing both sides of \eqref{eq:FLV_KV_evolution_1} by $\eta_{\alpha}$ yields
\begin{align*}
\dot{\bm E}^{\mathrm v}_{\alpha} + \frac{1}{\tau_{\alpha}} \bm E^{\mathrm v}_{\alpha} = \frac{1}{\eta_{\alpha}} \bm T^{\infty}.
\end{align*}
Given the initial condition $\bm E^{\mathrm v}_{\alpha}|_{t=0} = \bm O$ and $\bm T^{\infty}|_{t=0} =\bm O$, the solution can be written as
\begin{align*}
\bm E^{\mathrm v}_{\alpha} = \frac{1}{\eta_{\alpha}} \int_{0}^{t} \exp\left(-\frac{t - s}{\tau_\alpha}\right) \bm T^{\infty}(s)ds = \frac{1}{\mu_{\alpha}} \bm T^{\infty} - \frac{1}{\mu_{\alpha}} \int_{0}^{t} \exp\left(-\frac{t - s}{\tau_\alpha}\right) \dot{\bm T}^{\infty}(s)ds.
\end{align*}
Following \eqref{eq:kelvin_voigt_strain_decomp}, one has
\begin{align}
\label{eq:kelvin_voigt_flv_convolution}
\bm E = \int_{0}^{t} \mathcal J(t-s) \dot{\bm T}^{\infty}(s)ds, \quad \mbox{with} \quad \mathcal J(t) := \left(\frac{1}{\mu^{\infty}} + \sum_{\alpha=1}^{M}\frac{1}{\mu_{\alpha}} \right) - \sum_{\alpha=1}^{M} \frac{1}{\mu_{\alpha}} \exp\left(-\frac{t}{\tau_\alpha}\right).
\end{align}
The above essentially gives the convolution representation of the strain based on the creep function $\mathcal J$. This formulation is the counterpart to the solution representation \eqref{eq:maxwell_flv_convolution} of the generalized Maxwell model, after recognizing that $\bm T^{\infty}$ alone leads to the total stress $\bm S$ for the generalized Kelvin-Voigt model here. The relaxation and creep functions are complementary, each naturally capturing the material response under different loading conditions.

Further insights into \eqref{eq:kelvin_voigt_flv_convolution} can be gained by considering a creep test. The stress history is set as $\bm T^{\infty} = \bm T^{\infty}_0 H(t)$, where $\bm T^{\infty}_0$ stands for a constant stress. Substituting this into \eqref{eq:kelvin_voigt_flv_convolution} yields the following strain response
\begin{align*}
\bm E = \left(\frac{1}{\mu^{\infty}} + \sum_{\alpha=1}^{M}\frac{1}{\mu_{\alpha}} \big( 1 -  \exp\left(-t/\tau_\alpha\right) \big)  \right) \bm T^{\infty}_0.
\end{align*}
In the limit $t\rightarrow 0_+$, we obtain the instantaneous strain response given by $\bm E = \bm T^{\infty}_0 / \mu^{\infty}$, suggesting that the initial deformation only takes place in the equilibrium spring. As $t\rightarrow \infty$, we have the equilibrium strain as
\begin{align*}
\bm E = \left(\frac{1}{\mu^{\infty}} + \sum_{\alpha=1}^{M}\frac{1}{\mu_{\alpha}} \right) \bm T^{\infty}_0,
\end{align*}
representing the total long-time deformation after complete retardation. The results indicate that the applied stress is distributed across all springs in the generalized Kelvin-Voigt device. 

\paragraph{\textbf{On the relation of the two rheological representations}}
The generalized Maxwell and Kelvin–Voigt models differ in their structure and thermodynamic interpretations. It is intriguing to consider the conditions under which they yield identical constitutive responses. For the finite linear viscoelastic models discussed in Sections \ref{sec:Gen_Maxwell_model_FLV} and \ref{sec:Gen_KV_model_FLV}, one cannot establish a direct equivalence between them, as the generalized Maxwell model permits a fully nonlinear form for its equilibrium energy. To enable a meaningful comparison, we impose a restriction that the equilibrium part of the energy in the generalized Maxwell model is quadratic, mirroring the corresponding energy in the generalized Kelvin-Voigt model. If we further assume the same generalized strain is employed, the stress response \eqref{eq:maxwell_flv_convolution} reduces to the following convolution form based on a relaxation function:
\begin{align}
\label{eq:gen_maxwell_def_relaxation_fun_G}
\bm S = \int_{0}^{t} \mathcal G(t-s) \dot{\bm E}(s) ds : \mathbb Q, \quad \mbox{with} \quad \mathcal G(t) = \bar{\mu}^{\infty} + \sum_{\alpha=1}^{M} \bar{\mu}_\alpha\exp\left(-\frac{t}{\bar{\tau}_\alpha}\right)
\end{align}
where the overbar denotes moduli associated with the generalized Maxwell model. Achieving full equivalence over time requires matching the relaxation and creep functions of the two models. The result is summarized in the following proposition.

\begin{proposition}
\label{proposition1}
Consider a generalized Maxwell model and a generalized Kelvin-Voigt model with the following conditions:
\begin{itemize}
\item The generalized Maxwell model is characterized by the following quadratic free energies
\begin{align*}
\Psi^{\infty} = \frac12 \bar{\mu}^{\infty} \lvert \bm E \rvert^2 \quad \mbox{and} \quad \Psi^{\mathrm{neq}}_{\alpha} = \frac12 \bar{\mu}_{\alpha} \lvert \bm E - \bm E^{\mathrm v}_{\alpha} \rvert^2,
\end{align*}
where $\bar{\mu}^\infty$ and $\bar{\mu}_\alpha$ denote the corresponding shear moduli and we use $\bar{\eta}_\alpha$ to represent its viscosity coefficients;
\item The same generalized strain $\bm E$ is utilized to characterize the total deformation in the two models;
\item The relaxation function $\mathcal G$ of the generalized Maxwell model and the creep function $\mathcal J$ of the generalized Kelvin-Voigt model satisfy the following relation
\begin{align}
\label{eq:relax_creep_convolution}
\int_{0}^{t} \mathcal G(s)\mathcal J(t-s) ds = \int_{0}^{t} \mathcal G(t-s)\mathcal J(s) ds = t.
\end{align}
\end{itemize}
Then the two models give identical stress responses over time for the same strain history.
\end{proposition}

\begin{proof}
Considering the generalized Maxwell model in this setting, the second Piola–Kirchhoff stress is expressed as $\bm S = \bm T : \mathbb{Q}$. Given a strain history $\tilde{\bm E}(t)$, the relation \eqref{eq:gen_maxwell_def_relaxation_fun_G} gives the stress response as
\begin{align}
\label{eq:prop_1_T_E}
\bm T(t) = \int_{0}^{t} \mathcal G(t-s) \dot{\tilde{\bm E}}(s) ds.
\end{align}
For the generalized Kelvin-Voigt model, the strain-stress relation is expressed via \eqref{eq:kelvin_voigt_flv_convolution} as
\begin{align}
\label{eq:prop_1_E_T_infty}
\tilde{\bm E}(s) = \int_{0}^{s} \mathcal J(s-\tau) \dot{\bm T}^{\infty}(\tau) d\tau.
\end{align}
We may substitute \eqref{eq:prop_1_E_T_infty} into \eqref{eq:prop_1_T_E} and have
\begin{align}
\bm T(t) =& \int_0^t \mathcal G(t-s) \mathcal J(0) \dot{\bm T}^{\infty}(s) ds + \int_0^t \int_0^s \mathcal G(t-s) \dot{\mathcal J}(s-\tau) \dot{\bm T}^{\infty}(\tau)d\tau ds \nonumber \displaybreak[2] \\
\label{eq:prop_1_intermediate}
=& \int_0^t \mathcal G(t-\tau) \mathcal J(0) \dot{\bm T}^{\infty}(\tau) d\tau + \int_0^t \dot{\bm T}^{\infty}(\tau) \int_\tau^t \mathcal G(t-s) \dot{\mathcal J}(s-\tau) ds d\tau.
\end{align}
Taking time derivative on relation \eqref{eq:relax_creep_convolution}$_1$ yields
\begin{align*}
\mathcal G(t) \mathcal J(0) + \int_0^t \mathcal G(s) \dot{\mathcal J}(t-s) ds = 1,
\end{align*}
which can be rewritten as
\begin{align*}
\mathcal G(t-\tau) \mathcal J(0) + \int_{\tau}^{t} \mathcal G(t-s) \dot{\mathcal J}(s-\tau) ds = 1.
\end{align*}
Substituting the above into \eqref{eq:prop_1_intermediate} results in $\bm T(t) = \bm T^{\infty}(t)$. Since the same generalized strain is used, we may conclude that the second Piola-Kirchhoff stresses are identical over time for the two models. 
\end{proof}

We may observe that, under the assumption of quadratic energy for the generalized Maxwell model and a consistent choice of generalized strain, the proof of Proposition \ref{proposition1} essentially establishes the interrelation between the relaxation and creep functions. This can be more conveniently analyzed via their Laplace transforms, a well-established technique in small-strain linear viscoelasticity \cite{Lakes2009}. For models with more than two non-equilibrium processes (i.e. $M>2$), obtaining an analytic expression for the interconversion of the Prony series is non-trivial \cite{Loy2015,SerraAguila2019}. Nevertheless, when $M=1$, the relation \eqref{eq:relax_creep_convolution} can be replaced by the following conditions,
\begin{align}
\label{eq:equivalence_moduli_relation}
\mu^{\infty} = \bar{\mu}^{\infty} + \bar{\mu}_{1}, \quad
\frac{1}{\bar{\mu}^{\infty}} = \frac{1}{\mu^{\infty}} + \frac{1}{\mu_{1}}, \quad
\frac{\bar{\tau}_1}{\tau_1} = \frac{\bar{\mu}^{\infty}}{\mu^{\infty}}.
\end{align}
The first two expressions of \eqref{eq:equivalence_moduli_relation} ensure that the instantaneous and equilibrium responses match, based on our prior discussions. In the last, we note that Proposition \ref{proposition1} only ensures that the total stress response of the two models coincide. The internal variables generally evolve differently in the two models.

\subsubsection{Nonlinear viscoelasticity}
\label{sec:Gen_KV_model_FV}
Given the definition of the elastic strain $\bm E^\mathrm{e}$, the corresponding elastic deformation tensor $\bm C^\mathrm{e}$ can be constructed via equation \eqref{eq:def_Ce}. A set of viscous deformation tensors can also be introduced in an analogous manner to $\bm C^{\mathrm v}$ \eqref{eq:def_Cv},
\begin{align*}
\bm C^\mathrm{v}_{\alpha} := \sum_{a=1}^{3}\lambda^\mathrm{v\:2}_{\alpha\:a} \bar{\bm M}_{\alpha\:a}
\quad\mbox{with}\quad
\lambda^\mathrm{v}_{\alpha\:a} = E^{\mathrm{v}\: -1}_{\alpha}(E^\mathrm{v}_{\alpha\:a}),
\end{align*}
where $E^\mathrm{v}_{\alpha\:a}$ and $\bar{\bm M}_a$ are obtained from the spectral decomposition of $\bm E^\mathrm{v}_\alpha$, and $E^{\mathrm{v}\: -1}_{\alpha}$ is the inverse of the scale $E^{\mathrm{v}}_\alpha$. With the deformation tensors properly introduced, the strain energy \eqref{eq:strain_energy_kelvin_voigt} takes the particular form
\begin{align*}
\Psi^\infty(\bm C,\bm E^\mathrm{v}_1, \dots, \bm E^\mathrm{v}_M ) = \Psi^\infty(\bm C^\mathrm{e})
\quad\text{and}\quad
\Psi^\mathrm{neq}_\alpha(\bm E^\mathrm{v}_\alpha) = \Psi^\mathrm{neq}_{\alpha}(\bm C^\mathrm{v}_\alpha).
\end{align*}
Following \eqref{eq:KV_T}, the thermodynamic forces can be represented as
\begin{gather*}
\bm T^\infty = \bm S^\mathrm{e} : \mathbb Q^{\mathrm{e}\:-1}
\quad\text{with}\quad
\bm S^\mathrm{e}:= 2\frac{\partial \Psi^\infty}{\partial \bm C^\mathrm{e}},
\quad
\mathbb Q^{\mathrm{e}\:-1} := \frac12 \frac{\partial \bm C^\mathrm{e}}{\partial \bm E^\mathrm{e}},
\end{gather*}
and
\begin{gather*}
\bm T^\mathrm{neq}_\alpha = \bm S^\mathrm{v}_\alpha:\mathbb Q^\mathrm{v\:-1}_\alpha
\quad\text{with}\quad
\bm S^\mathrm{v}_\alpha:= 2\frac{\partial \Psi^\mathrm{neq}_\alpha}{\partial \bm C^\mathrm{v}_\alpha},
\quad
\mathbb Q^\mathrm{v\:-1}_\alpha := \frac12 \frac{\partial \bm C^\mathrm{v}_\alpha}{\partial \bm E^\mathrm{v}_\alpha}.
\end{gather*}
With the above expressions, the evolution equation for the $\alpha$-th Voigt element, as presented in the general form \eqref{eq:KV_evolution}, becomes
\begin{align}
\label{eq:FV_KV_evolution}
- \bm S^\mathrm{e}:\mathbb Q^\mathrm{e\:-1} + \bm S^\mathrm{v}_\alpha:\mathbb Q^{\mathrm{v}\:-1}_\alpha+ \eta_\alpha \dot{\bm E}^\mathrm{v}_\alpha = \bm O.
\end{align}
Note that the nonlinear evolution equations remain fully coupled for all non-equilibrium processes, as $\bm S^\mathrm{e}$ depends on all internal variables. The integration strategy for this coupled nonlinear system is discussed in Section \ref{sec:numerical_formulation_FV_KV}. A summary of the constitutive relations based on the two rheological architectures is presented in Table \ref{tab:comparison}, which highlights their distinct stress–strain responses and evolution equations.

\begin{sidewaystable}[hbtp]
\centering
\renewcommand{\arraystretch}{1.6}
\setlength{\tabcolsep}{10pt}
\begin{tabular}{>{\bfseries}l c c}
\toprule
& \textbf{The generalized Maxwell Model} & \textbf{The generalized Kelvin-Voigt Model} \\
\midrule
& \includegraphics[trim=180 120 180 70, clip, scale=0.3]{./gen_Maxwell_model.pdf} &
\raisebox{1.0cm}{\includegraphics[trim=20 260 150 100, clip, scale=0.3]{gen_KV_model.pdf}} \\
Green–Naghdi assumption & 
$\displaystyle \bm E^\mathrm{e}_\alpha = \bm E_\alpha - \bm E^\mathrm{v}_\alpha$ & 
$\displaystyle \bm E^\mathrm{e} = \bm E - \sum_{\alpha=1}^M \bm E^\mathrm{v}_\alpha$ \\
\addlinespace[1em]
Strain energy & 
$\displaystyle \Psi = \Psi^\infty(\bm C) + \sum_{\alpha=1}^M \Psi^\mathrm{neq}_\alpha(\bm C, \bm E^\mathrm{v}_\alpha)$ & 
$\displaystyle \Psi = \Psi^\infty(\bm C, \bm E^\mathrm{v}_1, \dots,\bm E^\mathrm{v}_M ) + \sum_{\alpha=1}^M \Psi^\mathrm{neq}_\alpha(\bm E^\mathrm{v}_\alpha)$ \\
\addlinespace[1em]
The second Piola-Kirchhoff stress & 
$\displaystyle \bm S = 2\frac{\partial \Psi^\infty}{\partial \bm C} + \sum_{\alpha=1}^M 2\frac{\partial \Psi^\mathrm{neq}_\alpha}{\partial \bm C}$ & 
$\displaystyle \bm S = 2\frac{\partial \Psi^\infty}{\partial \bm C}$ \\
\addlinespace[1em]
General evolution equations & 
$\displaystyle \frac{\partial \Psi^\mathrm{neq}_\alpha}{\partial \bm E^\mathrm{v}_\alpha} + \eta_\alpha \dot{\bm E}^\mathrm{v}_\alpha = \bm{O}$ & 
$\displaystyle \frac{\partial \Psi^\infty}{\partial \bm E^\mathrm{v}_\alpha} + \frac{\partial \Psi^\mathrm{neq}_\alpha}{\partial \bm E^\mathrm{v}_\alpha} + \eta_\alpha \dot{\bm E}^\mathrm{v}_\alpha = \bm{O}$ \\
\addlinespace[1em]
Linear evolution equations & 
$\displaystyle -\mu_\alpha \bm E^\mathrm{e}_\alpha + \eta_\alpha \dot{\bm E}^\mathrm{v}_\alpha = \bm{O}$ & 
$\displaystyle -\mu^\infty \bm E^\mathrm{e} + \mu_\alpha \bm E^\mathrm{v}_\alpha + \eta_\alpha \dot{\bm E}^\mathrm{v}_\alpha = \bm{O}$ \\
\addlinespace[1em]
Nonlinear evolution equations & 
$\displaystyle -\bm S^\mathrm{e}_\alpha : \mathbb{Q}_\alpha^{\mathrm{e}\,-1} + \eta_\alpha \dot{\bm E}^\mathrm{v}_\alpha = \bm{O}$ & 
$\displaystyle -\bm S^\mathrm{e} : \mathbb{Q}^{\mathrm{e}\,-1} + \bm S^\mathrm{v}_\alpha : \mathbb{Q}_\alpha^{\mathrm{v}\,-1} + \eta_\alpha \dot{\bm E}^\mathrm{v}_\alpha = \bm{O}$ \\
\bottomrule
\end{tabular}
\caption{Comparison of constitutive relations between the generalized Maxwell and generalized Kelvin-Voigt models.}
\label{tab:comparison}
\end{sidewaystable}

\paragraph{\textbf{Example}}
Here we present the expressions for the free energy in the context of the generalized Kelvin-Voigt model. The equilibrium part of the energy, as given in \eqref{eq:example1_energy_eq}, is modified here by replacing $\bm C$ with $\bm C^\mathrm{e}$, and the $\alpha$-th non-equilibrium energy component in \eqref{eq:example1_energy_neq} is modified by replacing $\bm C^\mathrm{e}_\alpha$ with the viscous deformation tensor $\bm C^\mathrm{v}_\alpha$. Thus, they take the following form,
\begin{align*}
\Psi^\infty(\bm C^\mathrm{e}) =& \mu^\infty N^\infty \left( \lambda^\infty \mathfrak{L}^{-1}(\lambda^\infty) + \ln\frac{\mathfrak L^{-1}(\lambda^\infty)}{\sinh(\mathfrak L^{-1}(\lambda^\infty))} \right) - \frac{\mu^\infty \sqrt{N^\infty}}{3} \mathfrak L^{-1}(1/\sqrt{N^\infty})\ln(J^\mathrm{e}),\\
\Psi^\mathrm{neq}_\alpha(\bm C^\mathrm{v}_\alpha) =& \mu_\alpha N_\alpha \left( \lambda^\mathrm{neq}_\alpha \mathfrak L^{-1}(\lambda^\mathrm{neq}_\alpha) + \ln \frac{\mathfrak L^{-1}(\lambda^\mathrm{neq}_\alpha)}{\sinh(\mathfrak{L}^{-1}(\lambda^\mathrm{neq}_\alpha))}\right) - \frac{\mu_\alpha \sqrt{N_\alpha}}{3}\mathfrak{L}^{-1}(1/\sqrt{N_\alpha}) \ln (J^\mathrm{v}_\alpha),
\end{align*}
where the relative average chain stretches and volume ratios are defined as
\begin{align*}
\lambda^\infty := \sqrt{\frac{\bm C^\mathrm{e} :\bm I}{3 N^\infty}},
\quad
\lambda^\mathrm{neq}_\alpha := \sqrt{\frac{\bm C^\mathrm{v}_\alpha : \bm I}{3 N_\alpha}},
\quad
J^\mathrm{e} := \sqrt{\det \bm C^\mathrm{e}},
\quad
J^\mathrm{v}_\alpha := \sqrt{\det \bm C^\mathrm{v}_\alpha}.
\end{align*}
The corresponding fictitious stress-like tensors derived from the energy functions are given by
\begin{align*}
\bm S^\mathrm{e} &= 2\frac{\partial \Psi^\infty}{\partial \bm C^\mathrm{e}}
= \frac{\mu^\infty}{3\lambda^\infty} \mathfrak{L}^{-1}(\lambda^\infty) \bm I 
- \frac{\mu^\infty \sqrt{N^\infty}}{3} \mathfrak{L}^{-1}\left(\frac{1}{\sqrt{N^\infty}}\right) \bm C^{\mathrm{e}\:-1},\\
\bm S^\mathrm{v}_\alpha &= 2\frac{\partial \Psi^\mathrm{neq}_\alpha}{\partial \bm C^\mathrm{v}_\alpha} 
= \frac{\mu_\alpha}{3\lambda^\mathrm{neq}_\alpha} \mathfrak{L}^{-1}(\lambda^\mathrm{neq}_\alpha)\bm I-\frac{\mu_\alpha\sqrt{N_\alpha}}{3}\mathfrak{L}^{-1}\left(\frac{1}{\sqrt{N_\alpha}}\right) \bm C^{\mathrm{v}\:-1}_\alpha.
\end{align*}
The above two stresses enter into the evolution equations \eqref{eq:FV_KV_evolution} to complete their specification. Furthermore, the stress response is obtained by $\bm T^{\infty} = \bm S^{\mathrm e} : \mathbb Q^{\mathrm{e}\:-1}$ and $\bm S = \bm T^{\infty} : \mathbb Q$.

\section{Numerical formulation}
\label{sec:numerical_formula}
In this section, we present the integration algorithms for the numerical modeling of the constitutive models developed in Section \ref{sec:constitutive_theory}. We begin by defining discrete quantities at specific time steps. Let the total time interval of interest be $(0, T)$, which is partitioned into $n_\mathrm{ts}$ subintervals with discrete time points $\{ t_n \}_{n=0}^{n_\mathrm{ts}}$. The time increment is denoted by $\Delta t_n = t_{n+1} - t_n$, where $n = 0, 1, \dots, n_\mathrm{ts} - 1$. For a time-dependent quantity $(\cdot)$, its exact value evaluated at time $t_n$ is denoted by $(\cdot)|_{t_n}$, while its discrete approximation is represented as $(\cdot)_n$. Unless otherwise specified, the midpoint value is defined as
\begin{align*}
(\cdot)_{n+\frac{1}{2}} := \frac{1}{2} \Big[ (\cdot)_n + (\cdot)_{n+1} \Big].
\end{align*}
The consistent linearization of the momentum balance equation is completed by constructing the elasticity tensor,
\begin{align*}
\mathbb C_{n+1} := 2\frac{\partial \bm S_{n+1}}{\partial \bm C_{n+1}},
\end{align*}
which is closely related to the strain energies, the evolution of internal variables, and the underlying rheological representation.

\subsection{The generalized Maxwell model}
We consider the constitutive integration strategy as well as the consistent linearization of the generalized Maxwell model presented in Section \ref{sec:Gen_Maxwell_model}. Recall that the stress in the generalized Maxwell model is additively decomposed into the equilibrium and non-equilibrium parts, and the elasticity tensor is correspondingly represented as
\begin{align*}
\mathbb C_{n+1} = \mathbb C^\infty_{n+1} + \sum_{\alpha=1}^{M} \mathbb C^\mathrm{neq}_{\alpha\:n+1}, \quad \mbox{with} \quad \mathbb C^\infty_{n+1} := 2\frac{\partial \bm S^\infty_{n+1}}{\partial \bm C_{n+1}},
\quad
\mathbb C^\mathrm{neq}_{\alpha\:n+1} := 2\frac{\partial \bm S^\mathrm{neq}_{\alpha\:n+1}}{\partial \bm C_{n+1}}.
\end{align*}
Note that the equilibrium contribution $\mathbb C^\infty_{n+1}$ is obtained following the standard formulation of hyperelasticity. Our attention is directed to the non-equilibrium component. As expressed in \eqref{eq:Maxwell_S_neq}, the non-equilibrium stress $\bm S^\mathrm{neq}_\alpha$ is governed by the tensor $\bm T_\alpha$, whose specific formulation determines whether the resulting theory corresponds to linear or nonlinear viscoelasticity. We will discuss these two cases separately in the following subsections.

\subsubsection{Finite linear viscoelasticity}
\label{sec:numerical_formulation_FLV_Maxwell}
In the framework of finite linear viscoelastic theory, the evolution equation \eqref{eq:FLV_Maxwell_evolution} prescribes a linear relation between the thermodynamic force and the elastic strain, i.e., $\bm T_\alpha = \mu_\alpha \bm E^\mathrm{e}_\alpha$. Based on this linear relation, the evolution equations for the internal variables $\bm E^\mathrm{v}_\alpha$ and $\bm T_\alpha$ are explicitly formulated in \eqref{eq:maxwell_flv_evo_eqn} and \eqref{eq:maxwell_flv_T_evo_eqn}, respectively. The resulting algorithm is inspired by their hereditary integral representations.

\paragraph{\textbf{Constitutive integration}}
Based on the hereditary integral representation of $\bm E^\mathrm{v}_\alpha$ given in \eqref{eq:maxwell_flv_evo_eqn_T}$_1$, a recursive relation can be established between its values at $t_n$ and $t_{n+1}$, that is,
\begin{align*}
\left. \bm E^\mathrm{v}_\alpha\right|_{t_{n+1}} = \exp\left(-\frac{\Delta t_n}{\tau_\alpha}\right) \left.\bm E^\mathrm{v}_\alpha\right|_{t_n} + \int_{t_n}^{t_{n+1}} \exp\left(-\frac{t_{n+1}-s}{\tau_\alpha}\right) \frac{\bm E_\alpha(s)}{\tau_\alpha} ds.
\end{align*}
We mention that one may devise different integration formulas based on the hereditary integral \cite[pp.~353--355]{Simo2006}. Here we approximate the above through evaluating $\bm E_\alpha$ in the interval by its midpoint value $\bm E_{\alpha\:n+\frac12}$, yielding
\begin{align}
\label{eq:FLV_Maxwell_evo_Ev}
\bm E^\mathrm{v}_{\alpha\:n+1} = \xi_{\alpha \: n} \bm E^\mathrm{v}_{\alpha\:n} + \left(1 - \xi_{\alpha \: n} \right)\bm E_{\alpha\:n+\frac12} \quad \mbox{with} \quad \xi_{\alpha \: n} := \exp\left( - \frac{\Delta t_n}{\tau_\alpha}\right).
\end{align}
The above formula provides an explicit, one-step update formula for $\bm E^\mathrm{v}_{\alpha\:n+1}$ with second-order accuracy. Alternatively, a recursive expression for $\bm T_\alpha$ between $t_{n+1}$ and $t_n$ can be analogously derived from the corresponding hereditary integral \eqref{eq:maxwell_flv_evo_eqn_T}$_2$. An explicit update of $\bm T_{\alpha\:n+1}$ may facilitate a direct evaluation of the non-equilibrium stress component via the relation $\bm S^\mathrm{neq}_{\alpha\:n+1} = \bm T_{\alpha\:n+1} : \mathbb Q_{\alpha\:n+1}$.

\paragraph{\textbf{Elasticity tensor}}
Recalling the expression for the non-equilibrium stress in equation \eqref{eq:Maxwell_S_neq}, it can be further written as $\bm S^\mathrm{neq}_\alpha = \mu_\alpha\, \bm E^\mathrm{e}_\alpha : \mathbb Q_\alpha$. Accordingly, the non-equilibrium contribution to the elasticity tensor takes the form
\begin{align*}
\mathbb C^\mathrm{neq}_{\alpha\:n+1}= \mu_\alpha\mathbb Q^\mathrm{T}_{\alpha\:n+1} :\left(\mathbb Q_{\alpha\:n+1}-\mathbb H_{\alpha\:n+1}\right)+\mu_\alpha \bm E^\mathrm{e}_{\alpha\:n+1} :\bm{\mathcal L}_{\alpha\:n+1}
\end{align*}
with
\begin{align*}
\mathbb H_{\alpha\:n+1}:=2\frac{\partial \bm E^\mathrm{v}_{\alpha\:n+1}}{\partial \bm C_{n+1}},
\quad
\mathbb Q_{\alpha\:n+1}:=2\frac{\partial \bm E_{\alpha\:n+1}}{\partial \bm C_{n+1}},
\quad 
\bm E^{\mathrm e}_{\alpha\:n+1}:= \bm E_{\alpha\:n+1} - \bm E^\mathrm{v}_{\alpha\:n+1}, \quad \mbox{and} \quad
\bm{\mathcal L}_{\alpha\:n+1}:=2\frac{\partial \mathbb Q_{\alpha\:n+1}}{\partial \bm C_{n+1}}.
\end{align*}
Note that the rank-four algorithmic tensor $\mathbb H_{\alpha\:n+1}$ arises from the numerical integration of the evolution equation \eqref{eq:FLV_Maxwell_evo_Ev} and is given by
\begin{gather*}
\mathbb H_{\alpha\:n+1} = \frac{1}{2}\left( 1-\xi_{\alpha \: n} \right)\mathbb Q_{\alpha\:n+1}.
\end{gather*}
Consequently, the non-equilibrium part of the elasticity tensor can be explicitly represented as
\begin{align}
\label{eq:elasticity_neq_alpha_n+1}
\mathbb C^\mathrm{neq}_{\alpha\:n+1} =\frac{\mu_{\alpha}}{2} \left( 1+\xi_{\alpha \: n} \right)\mathbb Q^\mathrm{T}_{\alpha\:n+1}:\mathbb Q_{\alpha\:n+1} + \mu_\alpha \bm E^\mathrm{e}_{\alpha\:n+1} :\bm{\mathcal L}_{\alpha\:n+1}.
\end{align}
This representation is particularly amenable for finite element implementation, as it enables modular and systematic integration of rheological components into the overall material response.

\subsubsection{Nonlinear viscoelasticity}
\label{sec:numerical_formulation_FV_Maxwell}
In the nonlinear theory presented in Section \ref{sec:Gen_Maxwell_model_FV}, the thermodynamic force $\bm T_\alpha = \bm S^\mathrm{e}_\alpha : \mathbb Q^\mathrm{e\:-1}_\alpha$ exhibits a nonlinear dependence on the internal variable $\bm E^\mathrm{v}_\alpha$. Consequently, the evolution equation needs to be solved by an iterative procedure.

\paragraph{\textbf{Constitutive integration}}
We discretize the evolution equation \eqref{eq:FV_Maxwell_evolution} by the midpoint scheme, yielding the discrete evolution equations as
\begin{align}
\label{eq:FV_Maxwell_evo_Ev}
\bm E^\mathrm{v}_{\alpha\:n+1}  = \bm E^\mathrm{v}_{\alpha\:n} + \frac{\Delta t_n}{\eta_\alpha} \bm S^\mathrm{e}_{\alpha\:n+\frac12} : \mathbb Q^\mathrm{e\:-1}_{\alpha\:n+\frac12}, \quad \mbox{for} \quad \alpha = 1, \cdots, M,
\end{align}
in which the quantities evaluated at the midpoint are specifically defined as
\begin{align*}
\bm C^\mathrm{e}_{\alpha\: n+\frac12} := \frac{1}{2} \left( \bm C^\mathrm{e}_{\alpha\: n} + \bm C^\mathrm{e}_{\alpha\: n+1} \right),
\quad
\bm S^\mathrm{e}_{\alpha\:n+\frac12}:=2\frac{\partial \Psi^\mathrm{neq}_\alpha}{\partial \bm C^\mathrm{e}_{\alpha\:n+\frac12}},
\quad
\mathbb Q^\mathrm{e\:-1}_{\alpha\: n+\frac12} := \frac12\frac{\partial \bm C^\mathrm{e}_{\alpha\: n+\frac12}}{\partial \bm E^\mathrm{e}_{\alpha\: n+\frac12}}.
\end{align*}
Given $\bm E_{\alpha\:n+1}$, the evolution equation for the $\alpha$-th process is solved using Newton-Raphson iteration by formulating the residual of \eqref{eq:FV_Maxwell_evo_Ev} as $\bm R_{\alpha\: n+1} = \bm O$, where
\begin{align*}
\bm R_{\alpha\: n+1} := \bm E^\mathrm{v}_{\alpha\: n+1} - \bm E^\mathrm{v}_{\alpha\: n} - \frac{\Delta t_n}{\eta_\alpha} \bm S^\mathrm{e}_{\alpha\: n+\frac12} : \mathbb Q^\mathrm{e\:-1}_{\alpha\: n+\frac12}.
\end{align*}
The consistent tangent for the above residual is given by
\begin{align*}
\mathbb K_{\alpha\: n+1} := \frac{\partial \bm R_{\alpha\: n+1}}{\partial \bm E^\mathrm{v}_{\alpha\: n+1}} = \mathbb I + \frac{\Delta t_n}{\eta_\alpha} \left( \mathbb Q^\mathrm{e\:-T}_{\alpha\: n+\frac12} : \frac{1}{2} \mathbb C^\mathrm{e}_{\alpha\: n+\frac12} : \mathbb Q^\mathrm{e\:-1}_{\alpha\: n+1} + \bm S^\mathrm{e}_{\alpha\: n+\frac12} : \bm{\mathcal{K}}^\mathrm{e}_{\alpha\: n+\frac12} : \mathbb Q^\mathrm{e}_{\alpha\: n+\frac12} : \mathbb Q^\mathrm{e\:-1}_{\alpha\: n+1} \right),
\end{align*}
in which the auxiliary tensors are defined as
\begin{gather*}
\mathbb Q^\mathrm{e\:-1}_{\alpha\: n+1} := \frac12\frac{\partial \bm C^\mathrm{e}_{\alpha\: n+1}}{\partial \bm E^\mathrm{e}_{\alpha\: n+1}}, \quad
\bm{\mathcal{K}}^\mathrm{e}_{\alpha\: n+\frac12} := \frac12\frac{\partial \mathbb Q^\mathrm{e\:-1}_{\alpha\: n+\frac12}}{\partial \bm E^\mathrm{e}_{\alpha\: n+\frac12}}, \quad
\mathbb C^\mathrm{e}_{\alpha\:n+\frac12} := 2\frac{\partial \bm S^\mathrm{e}_{\alpha\:n+\frac12}}{\partial \bm C^\mathrm{e}_{\alpha\:n+\frac12}}.
\end{gather*}
The rank-six tensor $\bm{\mathcal{K}}^\mathrm{e}_{\alpha\:n+1/2}$ is constructed based on the eigendecomposition of $\bm C^{\mathrm e}_{\alpha \: n+1/2}$, as was discussed in Section \ref{sec:kinematics}; the rank-four tensor $\mathbb C^\mathrm{e}_{\alpha\:n+1/2}$ is a fictitious elasticity tensor associated with the non-equilibrium energy contributions. If we use a subscript $(i)$ to denote quantities associated with the $i$-th Newton-Raphson iteration, the update scheme is formulated as
\begin{align*}
\mathbb K_{\alpha\: n+1\:(i)} : \Delta \bm E^\mathrm{v}_{\alpha\: n+1\:(i)} = - \bm R_{\alpha\: n+1\:(i)}, \qquad
\bm E^\mathrm{v}_{\alpha\: n+1\:(i+1)} = \bm E^\mathrm{v}_{\alpha\: n+1\:(i)} + \Delta \bm E^\mathrm{v}_{\alpha\: n+1\:(i)}.
\end{align*}
The tangent $\mathbb K_{\alpha\: n+1\:(i)}$ possesses minor symmetries and can be conveniently represented in matrix form using Voigt notation. Accordingly, the symmetric residual tensor $\bm R_{\alpha\: n+1\:(i)}$ is represented as a vector. This representation facilitates the solution of the above equations through a standard linear algebra solver.

\begin{remark}
If we discretize the constitutive relation by the first-order accurate backward Euler scheme, the consistent tangent takes the form
\begin{align*}
\mathbb I + \frac{\Delta t_n}{\eta_\alpha} \left( \mathbb Q^\mathrm{e\:-T}_{\alpha\: n+1} : \frac{1}{2} \mathbb C^\mathrm{e}_{\alpha\: n+1} : \mathbb Q^\mathrm{e\:-1}_{\alpha\: n+1} + \bm S^\mathrm{e}_{\alpha\: n+1} : \bm{\mathcal{K}}^\mathrm{e}_{\alpha\: n+1} \right).
\end{align*}
One may show that the above rank-four tensor possesses both major and minor symmetries. For sufficiently small $\Delta t_n$, the tensor is close to $\mathbb I$. One may thus expect it to be positive definite, at least for small time step sizes.
\end{remark}

\paragraph{\textbf{Elasticity tensor}}
We now turn to the elasticity tensor due to the non-equilibrium contributions. Recalling the definition of $\bm T_\alpha$ in the nonlinear viscoelastic model given in \eqref{eq:FV_Maxwell_T_alpha}, the $\alpha$-th non-equilibrium stress can be further expressed as $\bm S^\mathrm{neq}_\alpha = \bm S^\mathrm{e}_\alpha:\mathbb Q^\mathrm{e\:-1}_\alpha : \mathbb Q_\alpha$. The corresponding elasticity tensor is obtained by taking the derivative of $\bm S^\mathrm{neq}_{\alpha\:n+1}$ with $\bm C_{n+1}$, resulting in
\begin{align*}
\mathbb C^\mathrm{neq}_{\alpha\:n+1} &= \mathbb Q_{\alpha\:n+1} : ( \mathbb Q^\mathrm{e\:-T}_{\alpha\:n+1} : \mathbb C^\mathrm{e}_{\alpha\:n+1} : \mathbb Q^\mathrm{e\:-1}_{\alpha\:n+1} + 2\bm S^\mathrm{e}_{\alpha\:n+1}:\bm{\mathcal K}^\mathrm{e}_{\alpha\:n+1}):(\mathbb Q_{\alpha\:n+1} - \mathbb H_{\alpha\:n+1}),
\end{align*}
with
\begin{align*}
\mathbb C^\mathrm{e}_{\alpha\:n+1} := 2\frac{\partial \bm S^\mathrm{e}_{\alpha\:n+1}}{\partial \bm C^\mathrm{e}_{\alpha\:n+1}},
\quad
\bm S^\mathrm{e}_{\alpha\:n+1} := 2\frac{\partial \Psi^\mathrm{neq}_{\alpha}}{\partial \bm C^\mathrm{e}_{\alpha\:n+1}},
\quad
\bm{\mathcal{K}}^\mathrm{e}_{\alpha\: n+1} := \frac12\frac{\partial \mathbb Q^\mathrm{e\:-1}_{\alpha\: n+1}}{ \partial \bm E^\mathrm{e}_{\alpha\: n+1}}.
\end{align*}
Here, the rank-four algorithmic tensor $\mathbb H_{\alpha\:n+1}$ emerges from the implicit integration of the internal variable $\bm E^\mathrm{v}_\alpha$ given by \eqref{eq:FV_Maxwell_evo_Ev}. It takes the form
\begin{align*}
\mathbb H_{\alpha\:n+1} =& \frac{\Delta t_n}{\eta_\alpha} \left( \mathbb Q^\mathrm{e\:-T}_{\alpha\:n+\frac12}:\frac12 \mathbb C^\mathrm{e}_{\alpha\:n+\frac12}:\mathbb Q^\mathrm{e\:-1}_{\alpha\:n+1} + \bm S^\mathrm{e}_{\alpha\:n+\frac12}:\bm{\mathcal K}^\mathrm{e}_{\alpha\:n+\frac12} : \mathbb Q^\mathrm{e}_{\alpha\:n+\frac12}:\mathbb Q^\mathrm{e\:-1}_{\alpha\:n+1}\right):\left(\mathbb Q_{\alpha\:n+1}-\mathbb H_{\alpha\:n+1}\right) \displaybreak[2] \\
=& \left( \mathbb K_{\alpha \: n+1} - \mathbb I \right) : \left(\mathbb Q_{\alpha\:n+1}-\mathbb H_{\alpha\:n+1}\right)
\end{align*}
which can be reorganized as the following tensorial equation
\begin{gather*}
\mathbb K_{\alpha\:n+1} : \mathbb H_{\alpha\:n+1} = (\mathbb K_{\alpha\:n+1} - \mathbb I) : \mathbb Q_{\alpha\:n+1}.
\end{gather*}
The tensor $\mathbb H_{\alpha\:n+1}$, for $\alpha=1,\cdots, M$, is determined by solving the above tensorial equation. Once obtained, the formulation of the elasticity tensor $\mathbb C_{n+1}$ is fully established, allowing its adoption in the integration of the momentum balance equation.

\subsection{The generalized Kelvin-Voigt model}
This section outlines the numerical schemes for the generalized Kelvin-Voigt model. Unlike the generalized Maxwell model, the non-equilibrium energies in the generalized Kelvin-Voigt model are independent of $\bm C$. As a result, the elasticity tensor originates solely from the equilibrium contribution,
\begin{align*}
\mathbb C_{n+1} = \mathbb C^\infty_{n+1}
\quad\text{with}\quad
\mathbb C^\infty_{n+1} := 2\frac{\partial \bm S^\infty_{n+1}}{\partial \bm C_{n+1}}.
\end{align*}
Recall that the internal variables enter the equilibrium energy through the elastic strain $\bm E^\mathrm{e} = \bm E - \sum_{\alpha=1}^{M} \bm E^\mathrm{v}_\alpha$, which results in a fully coupled structure among all non-equilibrium processes. The coupling necessitates additional considerations in the numerical design.

\subsubsection{Finite linear viscoelasticity}
\label{sec:numerical_formulation_FLV_KV}
As discussed in Section \ref{sec:Gen_KV_model_FLV}, the finite linear viscoelastic theory is formulated based on the energies with quadratic forms. Substituting the definition of $\bm E^\mathrm{e}$ into the evolution equation \eqref{eq:FLV_KV_evolution_1} gives
\begin{align*}
\eta_\alpha\dot{\bm E}^\mathrm{v}_\alpha + \mu_\alpha \bm E^\mathrm{v}_\alpha = \mu^\infty \left(\bm E - \sum_{\alpha=1}^{M} \bm E^\mathrm{v}_\alpha\right).
\end{align*}
The equations remain linear with respect to each $\bm E^\mathrm{v}_\alpha$. Yet, the evolution equations for different non-equilibrium processes are intrinsically coupled, as the right-hand side involves all internal variables. Collecting the equations for $\alpha = 1, \dots, M$, the system can be expressed in a matrix-vector form
\begin{align}
\label{eq:FLV_KV_evolution_2}
\dot{\bm {\mathcal E}}^\mathrm{v} + \bm{\mathsf A} \bm{\mathcal E}^\mathrm{v} = \bm{\mathsf B} \bm{\mathcal E},
\end{align}
in which a suit of matrices and vectors with components of tensors is introduced as
\begin{gather*}
\dot{\bm{\mathcal E}}^\mathrm{v}=\begin{bmatrix}
\dot{\bm E}^\mathrm{v}_{1} &
\dot{\bm E}^\mathrm{v}_{2} &
\cdots &
\dot{\bm E}^\mathrm{v}_{M}
\end{bmatrix}^T,
\quad
\bm{\mathcal E}=
\begin{bmatrix}
\bm E &
\bm E &
\cdots &
\bm E
\end{bmatrix}^T, \displaybreak[2] \\
\bm{\mathsf A}=
\begin{bmatrix}
(\mu^\infty+\mu_1)/\eta_1 & \mu^\infty/\eta_1 & \cdots & \mu^\infty/\eta_1\\
\mu^\infty/\eta_2 & (\mu^\infty+\mu_2)/\eta_2 & \cdots & \mu^\infty/\eta_2\\
\vdots & \ddots & \ddots & \vdots\\
\mu^\infty/\eta_M & \cdots & \cdots & (\mu^\infty+\mu_M)/\eta_M
\end{bmatrix},
\quad
\bm{\mathsf B}=
\begin{bmatrix}
\mu^\infty/\eta_1 & 0 & \cdots & 0\\
0 & \mu^\infty/\eta_2 & \cdots & 0\\
\vdots & \ddots & \ddots & \vdots\\
0 & \cdots & \cdots & \mu^\infty/\eta_M
\end{bmatrix}.
\end{gather*}
Here, we organized the internal variables and strains into column vectors, whose entries are rank-two tensors. The multiplication between the scalar-valued matrices and the tensor-valued vectors is understood componentwise, that is, each scalar entry of the matrix multiplies the corresponding tensor entry in the vector. The above equation is formally analogous to the evolution equation \eqref{eq:maxwell_flv_evo_eqn} by recognizing that the unknowns are vectors of rank-two tensors. To obtain a solution for this problem, one will need to invert the matrix $\bm{\mathsf A}$, and the following result is useful.

\begin{lemma}
\label{lemma:sherman-morrison}
Let $\bm{\mathsf K} \in \mathbb{R}^{M \times M}$ be a matrix that admits the decomposition
\begin{align*}
\bm{\mathsf K} = \bm{\mathsf{ \Lambda}} + \bm{\mathsf u} \bm{\mathsf v}^\mathrm{T},
\end{align*}
where $\bm{\mathsf{ \Lambda}} \in \mathbb{R}^{M \times M}$, and $\bm{\mathsf u}, \bm{\mathsf v} \in \mathbb{R}^{M}$ are column vectors. The matrix $\bm{\mathsf K} \in \mathbb{R}^{M \times M}$ is invertible if and only if $1 + \bm{\mathsf v}^\mathrm{T} \bm{\mathsf{ \Lambda}}^{-1} \bm{\mathsf u} \neq 0$, and the inverse of $\bm{\mathsf K}$ can be explicitly expressed as
\begin{align*}
\bm{\mathsf K}^{-1} = \bm{\mathsf{ \Lambda}}^{-1} - \frac{\bm{\mathsf{ \Lambda}}^{-1} \bm{\mathsf u} \bm{\mathsf v}^\mathrm{T}\bm{\mathsf{ \Lambda}}^{-1}}{1+ \bm{\mathsf v}^\mathrm{T}\bm{\mathsf{ \Lambda}}^{-1}\bm{\mathsf u}}.
\end{align*}
\end{lemma}

The above lemma is a special case of the celebrated Sherman-Morrison-Woodbury formula \cite{Sherman1950,Hager1989}. We observe that the constant matrix $\bm{\mathsf A}$ in \eqref{eq:FLV_KV_evolution_2} admits the decomposition
\begin{align*}
\bm{\mathsf A} = \bm{\mathsf D} + \bm{\mathsf p} \bm{\mathsf q}^\mathrm{T},
\end{align*}
where
\begin{align*}
\bm{\mathsf D} = \mathrm{diag} \left( \frac{\mu_1}{\eta_1}, \dots, \frac{\mu_M}{\eta_M} \right), \quad
\bm{\mathsf p}^\mathrm{T} =
\begin{bmatrix}
\frac{\mu^\infty}{\eta_1} & \cdots & \frac{\mu^\infty}{\eta_M}
\end{bmatrix},
\quad
\bm{\mathsf q}^\mathrm{T} =
\begin{bmatrix}
1 & \cdots & 1
\end{bmatrix}.
\end{align*}
A direct consequence of Lemma \ref{lemma:sherman-morrison} is an explicit expression for the inverse of $\bm{\mathsf A}$, whose entries read
\begin{align}
\label{eq:A-inverse}
\left( \bm{\mathsf A}^{-1} \right)_{ij} = \delta_{ij} \, \frac{\eta_i}{\mu_i}- \frac{1}{ 1 + \sum\limits_{\alpha=1}^{M} \frac{ \mu^\infty }{ \mu_\alpha } } \frac{ \mu^\infty \, \eta_i }{ \mu_i \mu_j }
\quad\mbox{for}\quad
i,j=1,\dots,M.
\end{align}
This explicit formula for $\bm{\mathsf A}^{-1}$ is particularly valuable for the constitutive integration to be discussed.

\paragraph{\textbf{Constitutive integration}}
Since the evolution equations have been reorganized into a system of ordinary differential equations, one may verify that the following relation holds for the solution $\bm{\mathcal E}^{\mathrm v}(t)$ of \eqref{eq:FLV_KV_evolution_2},
\begin{align*}
\left. \left( \exp(\bm{\mathsf A} t)\bm{\mathcal E}^\mathrm{v}(t) \right) \right|_{t_{n+1}} =  \left. \left( \exp\left(\bm{\mathsf A} t \right)\bm{\mathcal E}^\mathrm{v}(t) \right) \right|_{t_n} + \int_{t_n}^{t_{n+1}}\exp\left(\bm{\mathsf A}s \right)\bm{\mathsf B}\bm{\mathcal E}(s) ds.
\end{align*}
In the above integral, the multiplication of $\exp\left(\bm{\mathsf A}s \right)\bm{\mathsf B}$ with $\bm{\mathcal E}(s)$ is defined componentwise. The integral is performed for a time-dependent vector of rank-two tensors and is also defined entrywise. If the integration over the interval $(t_n, t_{n+1})$ is approximated using the midpoint rule, one has
\begin{gather*}
\int_{t_n}^{t_{n+1}}\exp(\bm{\mathsf A}s) \bm{\mathsf B}\bm{\mathcal E}\, ds \approx\left( \int_{t_n}^{t_{n+1}}\exp(\bm{\mathsf A}s)\, ds\right)  \bm{\mathsf B} \bm{\mathcal E}_{n+\frac12} = \Delta t_n \bm{\mathsf A}^{-1} \left(\exp(\bm{\mathsf A}t_{n+1})-\exp(\bm{\mathsf A} t_n) \right)\bm{\mathsf B}\bm{\mathcal E}_{n+\frac12}.
\end{gather*}
Reorganizing the recursive relation, the update formula for $\bm{\mathcal E}^\mathrm{v}_{n+1}$ is given by
\begin{align}
\label{eq:FLV_KV_evo}
\bm{\mathcal E}^\mathrm{v}_{n+1} = \exp(-\bm{\mathsf A}\Delta t_n)\bm{\mathcal E}^\mathrm{v}_n + \bm{\mathsf A}^{-1}\left(\bm{\mathsf I} - \exp(-\bm{\mathsf A}\Delta t_n) \right)\bm{\mathsf B} \bm{\mathcal E}_{n+\frac12}.
\end{align}
Here, $\bm{\mathsf I}$ stands for the $M \times M$ identity matrix, and the matrix $\bm{\mathsf A}^{-1}$ has been given explicitly by \eqref{eq:A-inverse}. Despite the coupling nature of the generalized Kelvin-Voigt model, the explicit formula of $\bm{\mathsf A}^{-1}$ enables an explicit and simultaneous update of all internal variables $\bm E^\mathrm{v}_\alpha$ for a given input $\bm E_{n+1}$ in a single step. This makes the implementation particularly efficient, with the computational cost comparable to that of the generalized Maxwell model of Section \ref{sec:numerical_formulation_FLV_Maxwell}.

\paragraph{\textbf{Elasticity tensor}}
As previously discussed, under this rheological representation, the elasticity tensor originates solely from the equilibrium stress $\bm S^\infty = \mu^\infty \bm E^\mathrm{e}:\mathbb Q$. A subtle aspect lies in the term $\partial \bm E^\mathrm{e}/\partial \bm C$, which involves the summation of $\partial \bm E^\mathrm{v}_\alpha/\partial \bm C$. The elasticity tensor takes the form
\begin{gather*}
\mathbb C_{n+1} = \mu^\infty \mathbb Q^\mathrm{T}_{n+1} : \left(\mathbb Q_{n+1} - \mathbb G_{n+1}\right) + \mu^\infty \bm E^\mathrm{e}_{n+1}:\bm{\mathcal L}_{n+1},
\end{gather*}
in which
\begin{gather}
\label{eq:FLV_GKV_rank_four_tensors_def}
\mathbb Q_{n+1} := 2\frac{\partial \bm E_{n+1}}{\partial \bm C_{n+1}},
\quad
\mathbb G_{n+1} := \sum_{\alpha=1}^{M}2\frac{\partial \bm E^\mathrm{v}_{\alpha\:n+1}}{\partial \bm C_{n+1}},
\quad \mbox{and} \quad
\bm{\mathcal L}_{n+1} := 2\frac{\partial \mathbb Q_{n+1}}{\partial \bm C_{n+1}}.
\end{gather}
The rank-four tensor $\mathbb G_{n+1}$ depends on the specific integration scheme employed for updating $\bm{\mathcal E}^\mathrm{v}_{n+1}$. According to the constitutive integration scheme \eqref{eq:FLV_KV_evo}, its explicit formula is
\begin{gather*}
\mathbb G_{n+1} = \omega \mathbb Q_{n+1}, \quad \mbox{with} \quad \omega := \sum_{i=1}^M\sum_{j=1}^M \left(\frac12 \bm{\mathsf A}^{-1}\left( \bm{\mathsf I} - \exp(-\bm{\mathsf A}\Delta t_n)\right)\bm{\mathsf B}\right)_{ij}.
\end{gather*}
We may now rewrite the elasticity tensor as
\begin{gather*}
\mathbb C_{n+1} = \mu^\infty (1-\omega) \mathbb Q^\mathrm{T}_{n+1} : \mathbb Q_{n+1} + \mu^\infty \bm E^\mathrm{e}_{n+1}:\bm{\mathcal L}_{n+1}.
\end{gather*}
Although the non-equilibrium energy does not contribute explicitly to the elasticity tensor, the evolution of internal variables still affects its construction through the elastic strain. Comparing with the elasticity tensor \eqref{eq:elasticity_neq_alpha_n+1} derived for the generalized Maxwell model, we notice that both exhibit a similar pattern, reflecting a common underlying quadratic energy form. We also note that the calculation of the elasticity tensor here for the generalized Kelvin-Voigt model does not involve the summation over multiple branches.

\subsubsection{Nonlinear viscoelasticity}
\label{sec:numerical_formulation_FV_KV}
In the nonlinear theory, the thermodynamic force is given by $\bm T_\alpha = \bm S^\mathrm{e} : \mathbb Q^{\mathrm{e}\:-1} - \bm S^\mathrm{v} : \mathbb Q^{\mathrm{v}\:-1}_\alpha$. Both components of $\bm T_\alpha$ exhibit nonlinear dependence on the internal variable $\bm E^\mathrm{v}_\alpha$, necessitating an iterative procedure to solve the evolution equation \eqref{eq:FV_KV_evolution}.
\paragraph{\textbf{Constitutive integration}}
We discretize the evolution equation \eqref{eq:FV_KV_evolution} by the midpoint rule, yielding
\begin{align}
\label{eq:FV_KV_evo}
\bm E^\mathrm{v}_{\alpha\:n+1} = \bm E^\mathrm{v}_{\alpha\:n} + \frac{\Delta t_n}{\eta_\alpha} \left( \bm S^\mathrm{e}_{n+\frac12} : \mathbb Q^\mathrm{e\:-1}_{n+\frac12} - \bm S^\mathrm{v}_{\alpha\:n+\frac12} : \mathbb Q^\mathrm{v\:-1}_{\alpha\:n+\frac12} \right),
\end{align}
where the quantities at the midpoint are given by
\begin{gather*}
\bm C^\mathrm{e}_{n+\frac12} = \frac{1}{2} \left( \bm C^\mathrm{e}_n + \bm C^\mathrm{e}_{n+1} \right),
\quad
\bm C^\mathrm{v}_{\alpha\:n+\frac12} = \frac{1}{2} \left( \bm C^\mathrm{v}_{\alpha\:n} + \bm C^\mathrm{v}_{\alpha\:n+1} \right), \displaybreak[2] \\
\bm S^\mathrm{e}_{n+\frac12} := 2\frac{\partial\Psi^\infty}{\partial \bm C^\mathrm{e}_{n+\frac12}},
\quad
\bm S^\mathrm{v}_{\alpha\:n+\frac12} :=  2\frac{\partial\Psi^\mathrm{neq}_\alpha}{\partial\bm C^\mathrm{v}_{\alpha\:n+\frac12}},
\quad
\mathbb Q^\mathrm{e\:-1}_{n+\frac12} = \frac12 \frac{\partial \bm C^\mathrm{e}_{n+\frac12}}{\partial \bm E^\mathrm{e}_{n+\frac12}},
\quad
\mathbb Q^\mathrm{v\:-1}_{\alpha\:n+\frac12} = \frac12 \frac{\partial \bm C^\mathrm{v}_{\alpha\:n+\frac12}}{\partial \bm E^\mathrm{v}_{\alpha\:n+\frac12}}.
\end{gather*}
Given $\bm E_{n+1}$, the internal variables $\bm E^\mathrm{v}_{\alpha\:n+1}$ are determined by the Newton-Raphson procedure. The residual associated with the $\alpha$-th Voigt element is denoted by $\bm R_{\alpha\:n+1}$, and it reads
\begin{align*}
\bm R_{\alpha\:n+1} := \bm E^\mathrm{v}_{\alpha\:n+1} - \bm E^\mathrm{v}_{\alpha\:n} - \frac{\Delta t_n}{\eta_\alpha} \left( \bm S^\mathrm{e}_{n+\frac12} : \mathbb Q^\mathrm{e\:-1}_{n+\frac12} \right) + \frac{\Delta t_n}{\eta_\alpha}\left( \bm S^\mathrm{v}_{\alpha\:n+\frac12} : \mathbb Q^\mathrm{v\:-1}_{\alpha\:n+\frac12} \right).
\end{align*}
Correspondingly, the consistent tangent $\mathbb K_{\alpha\beta\:n+1}$, defined as the derivative of $\bm R_{\alpha\:n+1}$ with respect to $\bm E^\mathrm{v}_{\beta\:n+1}$, takes the form
\begin{align*}
\mathbb K_{\alpha\beta\:n+1} := \frac{\partial \bm R_{\alpha\:n+1}}{\partial \bm E^\mathrm{v}_{\beta\:n+1}} = \delta_{\alpha\beta} \mathbb I + \frac{\Delta t_n}{\eta_\alpha} \mathbb K^\mathrm{e}_{n+1} + \frac{\Delta t_n}{\eta_\alpha} \delta_{\alpha\beta}\mathbb K^\mathrm{v}_{\alpha\:n+1}
\end{align*}
with
\begin{align}
\label{eq:Ke_n+1_expression}
\mathbb K^\mathrm{e}_{n+1}:=& -\frac{\partial \left( \bm S^\mathrm{e}_{n+\frac12} : \mathbb Q^\mathrm{e\:-1}_{n+\frac12} \right) }{\partial \bm E^\mathrm{v}_{\beta\:n+1}}=\frac{\partial \left( \bm S^\mathrm{e}_{n+\frac12} : \mathbb Q^\mathrm{e\:-1}_{n+\frac12} \right) }{\partial \bm E^\mathrm{e}_{n+1}} \nonumber \displaybreak[2] \\
=& \mathbb Q^\mathrm{e\:-T}_{n+\frac12} : \frac{1}{2} \mathbb C^\mathrm{e}_{n+\frac12} : \mathbb Q^\mathrm{e\:-1}_{n+1} + \bm S^\mathrm{e}_{n+\frac12} : \bm{\mathcal{K}}^\mathrm{e}_{n+\frac12} : \mathbb Q^\mathrm{e}_{n+\frac12} : \mathbb Q^\mathrm{e\:-1}_{n+1}
\end{align}
and
\begin{align}
\label{eq:Kv_alpha_n+1_expression}
\mathbb K^\mathrm{v}_{\alpha\:n+1}:=& \frac{\partial \left( \bm S^\mathrm{v}_{\alpha\:n+\frac12} : \mathbb Q^\mathrm{v\:-1}_{\alpha\:n+\frac12} \right)}{\partial \bm E^\mathrm{v}_{\alpha\:n+1}} \nonumber \displaybreak[2] \\
=&\mathbb Q^\mathrm{v\:-T}_{\alpha\:n+\frac12} : \frac{1}{2} \mathbb C^\mathrm{v}_{\alpha\:n+\frac12} : \mathbb Q^\mathrm{v\:-1}_{\alpha\:n+1} + \bm S^\mathrm{v}_{\alpha\:n+\frac12} : \bm{\mathcal{K}}^\mathrm{v}_{\alpha\:n+\frac12} : \mathbb Q^\mathrm{v}_{\alpha\:n+\frac12} : \mathbb Q^\mathrm{v\:-1}_{\alpha\:n+1}.
\end{align}
In the above, the auxiliary tensors at the midpoint are defined as
\begin{align*}
\mathbb C^\mathrm{e}_{n+\frac12} := 2\frac{\partial \bm S^\mathrm{e}_{n+\frac12}}{\partial \bm C^\mathrm{e}_{n+\frac12}},
\quad
\bm{\mathcal{K}}^\mathrm{e}_{n+\frac12} := \frac12\frac{\partial \mathbb Q^\mathrm{e\:-1}_{n+\frac12}}{\partial \bm E^\mathrm{e}_{n+\frac12}},
\quad
\mathbb C^\mathrm{v}_{\alpha\:n+\frac12} := 2\frac{\partial \bm S^\mathrm{v}_{\alpha\:n+\frac12}}{\partial \bm C^\mathrm{v}_{\alpha\:n+\frac12}},
\quad
\bm{\mathcal{K}}^\mathrm{v}_{\alpha\:n+\frac12} := \frac12\frac{\partial \mathbb Q^\mathrm{v}_{\alpha\:n+\frac12}}{\partial \bm E^\mathrm{v}_{\alpha\:n+\frac12}}.
\end{align*}
Note that the tangent $\mathbb K_{\alpha\beta\:n+1}$ consists of an elastic part due to $\mathbb K^\mathrm{e}_{n+1}$ and a viscous part due to $\mathbb K^\mathrm{v}_{\alpha\:n+1}$. The elastic part originates from the equilibrium spring and is identical for all indices $\beta$; the viscous part is nonzero only when $\alpha = \beta$, indicating that the viscous evolution of the $\alpha$-th element is unaffected by the state of other Voigt elements. The tensors $\mathbb K^\mathrm{e}_{n+1}$ and $\mathbb K^\mathrm{v}_{\alpha\:n+1}$ for $\alpha =1,\cdots, M$ possess minor symmetries. If we apply the backward Euler scheme for the evolution equation, these tensors will enjoy both major and minor symmetries.

Using the subscript $(i)$ to denote quantities associated with the $i$-th Newton-Raphson iteration, the update formula for the internal variables $\bm E^\mathrm{v}_{\alpha\:n+1}$ is
\begin{align*}
\sum_{\beta=1}^{M}\mathbb K_{\alpha\beta\:n+1\:(i)} : \Delta \bm E^\mathrm{v}_{\beta\:n+1\:(i)} = - \bm R_{\alpha\:n+1\:(i)}, \quad
\bm E^\mathrm{v}_{\beta\:n+1\:(i+1)} = \bm E^\mathrm{v}_{\beta\:n+1\:(i)} + \Delta \bm E^\mathrm{v}_{\beta\:n+1\:(i)}.
\end{align*}
The summation over $\beta$ in the linearized equation naturally motivates a matrix representation, which can be formulated as
\begin{align}
&
\begin{bmatrix}
\mathbb{I} + \frac{\Delta t_n}{\eta_1}(\mathbb{K}^\mathrm{e}_{n+1} + \mathbb{K}^\mathrm{v}_{1\:n+1}) & 
\frac{\Delta t_n}{\eta_1} \mathbb{K}^\mathrm{e}_{n+1} & 
\cdots & 
\frac{\Delta t_n}{\eta_1} \mathbb{K}^\mathrm{e}_{n+1} \\
\frac{\Delta t_n}{\eta_2} \mathbb{K}^\mathrm{e}_{n+1} & 
\mathbb{I} + \frac{\Delta t_n}{\eta_2}(\mathbb{K}^\mathrm{e}_{n+1} + \mathbb{K}^\mathrm{v}_{2\:n+1}) & 
\cdots & 
\frac{\Delta t_n}{\eta_2} \mathbb{K}^\mathrm{e}_{n+1} \\
\vdots & 
\vdots & 
\ddots & 
\vdots \\
\frac{\Delta t_n}{\eta_M} \mathbb{K}^\mathrm{e}_{n+1} & 
\frac{\Delta t_n}{\eta_M} \mathbb{K}^\mathrm{e}_{n+1} & 
\cdots & 
\mathbb{I} + \frac{\Delta t_n}{\eta_M}(\mathbb{K}^\mathrm{e}_{n+1} + \mathbb{K}^\mathrm{v}_{M\:n+1})
\end{bmatrix}
\begin{bmatrix}
\Delta \bm E^\mathrm{v}_{1\:n+1} \\ 
\Delta \bm E^\mathrm{v}_{2\:n+1} \\
\vdots \\
\Delta \bm E^\mathrm{v}_{M\:n+1}
\end{bmatrix} \displaybreak[2] \nonumber \\
\label{eq:tensor_matrix}
& =
\begin{bmatrix} 
-\bm R_{1\:n+1} \\
-\bm R_{2\:n+1} \\
\vdots \\
-\bm R_{M\:n+1}
\end{bmatrix},
\end{align}
in which we ignored the subscript $(i)$ for notational simplicity. Here, the matrix on the left-hand side is of size $M\times M$, with each component being a rank-four tensor. Its multiplication with the vector is understood as a componentwise operation, while the multiplication between a rank-four tensor and a rank-two tensor is a standard tensor contraction. Although this system can be solved by expanding all components explicitly, doing so will lead to a linear system of considerable size, considering that this equation needs to be solved at each quadrature point. We observe that the non-diagonal components in the $\alpha$-th row consist solely of $(\Delta t_n / \eta_\alpha) \mathbb K^\mathrm{e}_{n+1}$, while the diagonal component additionally contains the term $\mathbb I + (\Delta t_n/\eta_\alpha) \mathbb K^\mathrm{v}_{\alpha\:n+1}$. Inspired by the formula for $\bm{\mathsf A}^{-1}$ in Section \ref{sec:numerical_formulation_FLV_KV}, we generalize that approach to \eqref{eq:tensor_matrix}. We first present a closed-form solution to \eqref{eq:tensor_matrix} that depends only on rank-four tensors.

\begin{proposition}
\label{prop:tensor_matrix}
If there exist fourth-order tensors $\mathbb M$ and $\mathbb N_{\alpha}$, for $\alpha=1,\cdots, M$, such that
\begin{align*}
\left(\mathbb I + \mathbb K^\mathrm{e}_{n+1\: (i)} : \sum_{\alpha=1}^{M} \frac{\Delta t_n}{\eta_\alpha} \mathbb N_{\alpha} \right) : \mathbb M = \mathbb I, \qquad
\left(\mathbb I + \frac{\Delta t_n}{\eta_\alpha} \mathbb K^\mathrm{v}_{\alpha\:n+1\: (i)}\right) : \mathbb N_{\alpha} = \mathbb I,
\end{align*}
then the solution of the system \eqref{eq:tensor_matrix} is given by
\begin{align}
\label{eq:tensor_matrix_sol}
\Delta \bm E^\mathrm{v}_{\alpha\:n+1\: (i)} =  -\mathbb N_{\alpha} : \bm R_{\alpha\:n+1\: (i)} + \frac{\Delta t_n}{\eta_\alpha} \mathbb N_{\alpha} : \mathbb M : \mathbb K^\mathrm{e}_{n+1\: (i)} : \sum_{\beta=1}^{M} \mathbb N_{\beta} : \bm R_{\beta\:n+1\: (i)}
\quad\mbox{for}\quad\alpha=1,\cdots,M.
\end{align}

\end{proposition}
\begin{proof}
The $\gamma$-th row of the system \eqref{eq:tensor_matrix}, for $1 \leq \gamma \leq M$, is given by the following tensorial equation,
\begin{align} 
\label{eq:row_gamma}
\frac{\Delta t_n}{\eta_\gamma} \mathbb K^\mathrm{e}_{n+1\: (i)} : \left( \sum_{\alpha = 1}^{M} \Delta \bm E^\mathrm{v}_{\alpha\:n+1\: (i)} \right) + (\mathbb I + \frac{\Delta t_n}{\eta_\gamma} \mathbb K^\mathrm{v}_{\gamma\:n+1\: (i)}) : \Delta \bm E^\mathrm{v}_{\gamma\:n+1\: (i)}
= -\bm R_{\gamma\:n+1\: (i)}.
\end{align}
We now verify that $\Delta \bm E^\mathrm{v}_{\alpha\:n+1\: (i)}$ given by the expression \eqref{eq:tensor_matrix_sol} satisfies the above equation. First, according to \eqref{eq:tensor_matrix_sol}, the summation of $\Delta \bm E^\mathrm{v}_{\alpha\:n+1\: (i)}$ can be written as
\begin{align*}
\sum_{\alpha=1}^{M} \Delta \bm E^\mathrm{v}_{\alpha\:n+1\: (i)} 
&= \sum_{\alpha=1}^{M} \left( -\mathbb N_{\alpha} : \bm R_{\alpha\:n+1\: (i)} 
+ \frac{\Delta t_n}{\eta_\alpha} \mathbb N_{\alpha} : \mathbb M : \mathbb K^\mathrm{e}_{n+1\: (i)} : \sum_{\beta=1}^{M} \mathbb N_{\beta} : \bm R_{\beta\:n+1\: (i)} \right) \\
&= - \bm Y
+ \left( \sum_{\alpha=1}^{M} \frac{\Delta t_n}{\eta_\alpha} \mathbb N_{\alpha} \right) : \mathbb M : \mathbb K^\mathrm{e}_{n+1\:(i)} : \bm Y
\end{align*}
with an auxiliary tensor $\bm Y$ defined as
\begin{align*}
\bm Y := \sum_{\beta=1}^{M} \mathbb N_{\beta} : \bm R_{\beta\:n+1\: (i)}.
\end{align*}	
Consequently, the first term in \eqref{eq:row_gamma} can be evaluated as
\begin{align}
\label{eq:term-1-expand}
\frac{\Delta t_n}{\eta_\gamma} \mathbb K^\mathrm{e}_{n+1\: (i)} : \sum_{\alpha=1}^{M} \Delta \bm E^\mathrm{v}_{\alpha\:n+1\: (i)} 
= \frac{\Delta t_n}{\eta_\gamma} \left( -\mathbb I + \mathbb K^\mathrm{e}_{n+1\: (i)}:\left( \sum_{\alpha=1}^{M} \frac{\Delta t_n}{\eta_\alpha} \mathbb N_{\alpha} \right) : \mathbb M   \right) : \left(  \mathbb K^\mathrm{e}_{n+1\: (i)} : \bm Y \right).
\end{align}
Substituting the expression \eqref{eq:tensor_matrix_sol} for $\Delta \bm E^\mathrm{v}_{\gamma\:n+1\: (i)}$ into the second term of equation \eqref{eq:row_gamma}, we obtain
\begin{align}
\left(\mathbb I + \frac{\Delta t_n}{\eta_\gamma} \mathbb K^\mathrm{v}_{\gamma\:n+1\: (i)} \right) : \Delta \bm E^\mathrm{v}_{\gamma\:n+1\: (i)}
&= \left(\mathbb I + \frac{\Delta t_n}{\eta_\gamma} \mathbb K^\mathrm{v}_{\gamma\:n+1\: (i)}\right) : \left( -\mathbb N_{\gamma} : \bm R_{\gamma\:n+1\: (i)} + \frac{\Delta t_n}{\eta_\gamma} \mathbb N_{\gamma} : \mathbb M : \mathbb K^\mathrm{e}_{n+1\: (i)} : \bm Y \right)\displaybreak[2] \nonumber \\
\label{eq:term-2-expand}
&= -\bm R_{\gamma\:n+1\: (i)} + \frac{\Delta t_n}{\eta_\gamma} \mathbb M : \mathbb K^\mathrm{e}_{n+1\: (i)} : \bm Y,
\end{align}
where we have invoked the condition $(\mathbb I + (\Delta t_n / \eta_\gamma) \mathbb K^\mathrm{v}_{\gamma\:n+1\:(i)}) : \mathbb N_{\gamma} = \mathbb I$. Combining \eqref{eq:term-1-expand} and \eqref{eq:term-2-expand} leads to
\begin{align*}
&-\bm R_{\gamma\:n+1\: (i)} +\frac{\Delta t_n}{\eta_\gamma} \left( -\mathbb I + \mathbb K^\mathrm{e}_{n+1\: (i)}:\left( \sum_{\alpha=1}^{M} \frac{\Delta t_n}{\eta_\alpha} \mathbb N_{\alpha} \right) : \mathbb M +\mathbb M  \right) : \left(  \mathbb K^\mathrm{e}_{n+1\: (i)} : \bm Y \right) = -\bm R_{\gamma\:n+1\: (i)},
\end{align*}
where the definition of $\mathbb M$ ensures the identity $\mathbb M + \mathbb K^\mathrm{e}_{n+1\: (i)} : \left(\sum_{\alpha=1}^{M} (\Delta t_n/\eta_\alpha) \mathbb N_{\alpha}\right) : \mathbb M = \mathbb I$. This verifies that the equation \eqref{eq:row_gamma} is satisfied by the expression \eqref{eq:tensor_matrix_sol}, which completes the proof.
\end{proof}

\begin{remark}
The explicit solution of the tensor-matrix equation system \eqref{eq:tensor_matrix}, as presented in Proposition \ref{prop:tensor_matrix}, is inspired by the Sherman–Morrison-Woodbury formula given in Lemma \ref{lemma:sherman-morrison}. Specifically, we observe that the tensor-valued matrix in \eqref{eq:tensor_matrix} can be formally decomposed into a diagonal matrix plus a rank-one modification formed by the outer product of two vectors. The tensor $\mathbb N_{\alpha}$ essentially stands for the inverse of the diagonal matrix, while the tensor $\mathbb M$ plays a role analogous to the denominator in the Sherman–Morrison-Woodbury formula.
\end{remark}

Proposition \ref{prop:tensor_matrix} provides an explicit solution to the system \eqref{eq:tensor_matrix}, which clearly reveals that the evolution of each individual internal variable is influenced by all Voigt elements. The incremental formula \eqref{eq:tensor_matrix} offers a structured and convenient procedure, which can be stated as follows. Given the residuals $\lbrace \bm R_{\alpha\: n+1 \: (i)}\rbrace_{\alpha=1}^{M}$, their incrementals are determined through the following algorithm.
\begin{enumerate}
\item For $\alpha=1,\cdots, M$, assemble $\mathbb K^\mathrm{v}_{\alpha\:n+1\: (i)}$ according to \eqref{eq:Kv_alpha_n+1_expression} and determine $\mathbb N_{\alpha}$ by solving 
\begin{align}
\label{eq:solve_NN_alpha}
\left(\mathbb I + \frac{\Delta t_n}{\eta_\alpha} \mathbb K^\mathrm{v}_{\alpha\:n+1\: (i)}\right) : \mathbb N_{\alpha} = \mathbb I.
\end{align}

\item Calculate $\bm Y = \sum_{\alpha=1}^{M} \mathbb N_{\alpha} : \bm R_{\alpha\:n+1\: (i)}$.

\item Assemble $\mathbb K^\mathrm{e}_{n+1\: (i)}$ according to \eqref{eq:Ke_n+1_expression}.

\item Evaluate the tensor contractions $\mathbb K^\mathrm{e}_{n+1\: (i)} : \bm Y$ and $\mathbb I + \mathbb K^\mathrm{e}_{n+1\: (i)} : \sum_{\alpha=1}^{M} (\Delta t_n / \eta_\alpha) \mathbb N_{\alpha}$. Then solve for $\bm X$ from the equation
\begin{align}
\label{eq:solve_XX}
\left(\mathbb I + \mathbb K^\mathrm{e}_{n+1\:(i)} : \sum_{\alpha=1}^{M} \frac{\Delta t_n}{\eta_\alpha} \mathbb N_{\alpha} \right) : \bm X = \mathbb K^\mathrm{e}_{n+1\:(i)} : \bm Y.
\end{align}

\item For $\alpha=1,\cdots, M$, the incremental of the internal variables are given by 
\begin{align*}
\Delta \bm E^\mathrm{v}_{\alpha\:n+1\: (i)} =  \mathbb N_{\alpha} : \left( -\bm R_{\alpha\:n+1\: (i)} + \frac{\Delta t_n}{\eta_\alpha} \bm X \right).
\end{align*}
\end{enumerate}
In the above algorithm, the rank-four tensor $\mathbb M$ does not arise explicitly. Instead, its action is implicitly achieved through solving the equation in Step 4. In the above algorithm, the coupling among the $M$ non-equilibrium processes is handled in a modular approach, enabled by the specific structure of the matrix in \eqref{eq:tensor_matrix}.

In the local Newton-Raphson iteration, we consider two strategies. Approach 1 directly solves \eqref{eq:tensor_matrix}, and the above algorithm is referred to as Approach 2. We assess their scalability with respect to the number of non-equilibrium processes $M$. In terms of memory cost, the storage of the matrix in \eqref{eq:tensor_matrix} grows as $\mathcal O(M^2)$. Approach 2 stores only the tensors $\mathbb N_\alpha$. Thus, its storage scales linearly with $M$. In terms of computational cost, Approach 1 solves the tensorial equation \eqref{eq:tensor_matrix} by converting it into a $6M\times6M$ matrix problem, in which each tensor is expressed as a $6\times6$ matrix in the Voigt form. The cost of solving this matrix problem using Gaussian elimination requires $\mathcal{O}(M^3)$ operations. In Approach 2, the dominant computational cost comes from determining $\mathbb{N}_\alpha$ for $\alpha=1,\dots,M$ in \eqref{eq:solve_NN_alpha} and solving a single tensorial equation for $\bm{X}$ in \eqref{eq:solve_XX}. Since the cost of determining $\mathbb{N}_\alpha$ is constant, the cost of Approach 2 scales linearly with the number of non-equilibrium processes. Therefore, Approach 2 achieves linear complexity in both memory and computation costs and is thus preferable for large values of $M$.

\paragraph{\textbf{Elasticity tensor}}
For the nonlinear model, the equilibrium stress is given by $\bm S^\infty = \bm S^\mathrm{e}:\mathbb Q^{\mathrm{e}\:-1}:\mathbb Q$. The corresponding elasticity tensor is obtained by taking the derivative of $\bm S^\infty_{n+1}$ with respect to $\bm C_{n+1}$. The resulting simplified form of the elasticity tensor reads
\begin{align*}
\mathbb C_{n+1} = \mathbb Q^\mathrm{T}_{n+1}:\left(\mathbb Q^\mathrm{e\:-T}_{n+1} : \mathbb C^\mathrm{e}_{n+1}:\mathbb Q^\mathrm{e\:-1}_{n+1}+2\bm S^\mathrm{e}_{n+1}:\bm{\mathcal K}^\mathrm{e}_{n+1}\right):(\mathbb Q_{n+1}-\mathbb G_{n+1})+\bm S^\mathrm{e}_{n+1}:\mathbb Q^\mathrm{e\:-1}_{n+1}:\bm{\mathcal L}_{n+1}
\end{align*}
with
\begin{gather*}
\bm S^\mathrm{e}_{n+1}:=2\frac{\partial \Psi^\infty}{\partial \bm C^\mathrm{e}_{n+1}},
\quad
\mathbb C^\mathrm{e}_{n+1}:=2\frac{\partial \bm S^\mathrm{e}_{n+1}}{\partial \bm C^\mathrm{e}_{n+1}},
\quad
\bm{\mathcal K}^\mathrm{e}_{n+1}:= \frac12\frac{\partial \mathbb Q^\mathrm{e\:-1}_{n+1}}{\partial\bm E^\mathrm{e}_{n+1}}.
\end{gather*}
Recall that the tensor $\mathbb G_{n+1}$ was defined in \eqref{eq:FLV_GKV_rank_four_tensors_def} as
\begin{align*}
\mathbb G_{n+1} := \sum_{\alpha=1}^{M}2\frac{\partial \bm E^\mathrm{v}_{\alpha\:n+1}}{\partial \bm C_{n+1}}.
\end{align*}
For the discrete evolution equation \eqref{eq:FV_KV_evo}, taking the derivative with respect to $\bm C_{n+1}$ leads to
\begin{gather}
\label{eq:solve_GG}
\left(\mathbb I+\frac{\Delta t_n}{\eta_\alpha}\mathbb K^\mathrm{v}_{\alpha\:n+1}\right):2\frac{\partial \bm E^\mathrm{v}_{\alpha\:n+1}}{\partial \bm C_{n+1}}=\frac{\Delta t_n}{\eta_\alpha}\mathbb K^\mathrm{e}_{n+1}:(\mathbb Q_{n+1}-\mathbb G_{n+1}).
\end{gather}
The rank-four tensors $\mathbb{N}_{\alpha\:n+1}$, for $\alpha=1,\cdots, M$, are defined analogously to \eqref{eq:solve_NN_alpha} by
\begin{align*}
\left(\mathbb I + \frac{\Delta t_n}{\eta_\alpha} \mathbb K^\mathrm{v}_{\alpha\:n+1} \right):\mathbb N_{\alpha\:n+1} = \mathbb I.
\end{align*}
Given $\mathbb N_{\alpha\:n+1}$, \eqref{eq:solve_GG} can be rewritten as
\begin{align*}
2\frac{\partial \bm E^\mathrm{v}_{\alpha\:n+1}}{\partial \bm C_{n+1}}=\mathbb N_{\alpha\:n+1} :\frac{\Delta t_n}{\eta_\alpha}\mathbb K^\mathrm{e}_{n+1}:(\mathbb Q_{n+1}-\mathbb G_{n+1}).
\end{align*}
Summing the above over $\alpha$ leads to
\begin{gather*}
\mathbb G_{n+1} = \left(\sum_{\alpha=1}^{M} \frac{\Delta t_n}{\eta_\alpha} \mathbb N_{\alpha\:n+1}\right):\mathbb K^\mathrm{e}_{n+1}:(\mathbb Q_{n+1} - \mathbb G_{n+1}),
\end{gather*}
which can be reorganized as
\begin{align*}
\left(\mathbb I +  \left(\sum_{\alpha=1}^{M} \frac{\Delta t_n}{\eta_\alpha} \mathbb N_{\alpha\:n+1}\right):\mathbb K^\mathrm{e}_{n+1} \right) :\mathbb G_{n+1} =  \left(\sum_{\alpha=1}^{M} \frac{\Delta t_n}{\eta_\alpha} \mathbb N_{\alpha\:n+1}\right):\mathbb K^\mathrm{e}_{n+1}:\mathbb Q_{n+1}.
\end{align*}
The tensor $\mathbb G_{n+1}$ is determined by solving the above equation. Then the formulation of $\mathbb C_{n+1}$ is completed, which enables a consistent integration of the momentum balance equation. We note that, in the nonlinear generalized Maxwell model, the computation likewise involves solving $M$ tensorial equations to determine $\mathbb H_{\alpha\: n+1}$. Therefore, both nonlinear models incur a similar computational cost for constructing the elasticity tensor.

\section{Results}
\label{sec:results}
This section presents the results of model calibration and isogeometric analysis, aiming to evaluate, validate, and compare the constitutive theories proposed in this work. We refer to the finite linear viscoelastic model within the generalized Maxwell framework as FLV-GM, and its counterpart based on the generalized Kelvin-Voigt framework as FLV-GKV. Likewise, the nonlinear viscoelastic models are denoted by NV-GM and NV-GKV, corresponding to the generalized Maxwell and Kelvin-Voigt formulations, respectively. These abbreviations will be used throughout the remaining part for clarity and brevity. Here, we employ the Curnier-Rakotomanana strain to characterize the total deformation, whose scale function is defined as
\begin{align*}
E(\lambda) = \frac{1}{m+n} \left(\lambda^m - \lambda^{-n}\right),
\quad \text{with} \quad mn > 0,
\end{align*}
where $m$ and $n$ are material parameters \cite{Curnier1991,Liu2024}. In particular, this strain is coercive, and we use the same form of the scale function to define both the elastic and viscous deformation tensors.

For the numerical analysis, the boundary of the body is denoted as $\Gamma_{\bm X} := \partial \Omega_{\bm X}$, and it can be decomposed into non-overlapping regions where displacement or traction is prescribed. The problem is discretized using an inf-sup stable element pair \cite{Liu2019a}. The discrete pressure space is defined by Non-Uniform Rational B-Splines (NURBS) basis functions of degree $\mathsf p$ with the highest possible continuity, while the discrete velocity space is constructed through $p$-refinement, increasing the degree to $\mathsf p+1$ while maintaining the continuity. The time integration is conducted using the generalized-$\alpha$ scheme \cite{Jansen2000,Liu2018}, which is parameterized by a single parameter $\varrho_{\infty}$, representing the spectral radius of the amplification matrix at the highest mode. The following settings are adopted in the numerical investigations:
\begin{enumerate}
\item The degree of the NURBS basis functions $\mathsf p$ is set to be $1$.
\item The generalized-$\alpha$ scheme is applied with $\varrho_{\infty}=0.0$.
\item Gaussian quadrature with $\mathsf p+2$ points is used in each direction.
\item In the local Newton-Raphson iteration for the constitutive integration, the stopping criteria involve a relative tolerance of $\mathrm{tol}_\mathrm{r} = 10^{-12}$, an absolute tolerance of $\mathrm{tol}_\mathrm{a} = 10^{-12}$, and a maximum of $i_{\mathrm{max}} = 10$ iterations.
\item The global Newton-Raphson iteration for integrating the linear momentum equations employs a relative tolerance of $\mathrm{tol}_\mathrm{R} = 10^{-10}$, an absolute tolerance of $\mathrm{tol}_\mathrm{A} = 10^{-10}$, and a maximum iteration limit of $l_{\mathrm{max}} = 10$, which are used as stopping criteria for convergence.
\end{enumerate}

\subsection{Relaxation and cyclic shear of a viscoelastic cube}
\label{sec:FEA}
In this section, only a single non-equilibrium process (i.e., $M=1$) is considered. For the finite linear viscoelastic models, the material parameters are chosen based on \eqref{eq:equivalence_moduli_relation}, to achieve identical stress responses. For the nonlinear models, we adopt the eight-chain model with $N^\infty=N_1 = 150$, which is a reasonable choice for the number of chain segments \cite{Liu2025}. The parameters are detailed in Table \ref{tab:FEM_parameters}. Additionally, the strain parameters are taken to be $m = n = 1$. The results displayed in this section use a mesh with 9 elements in each direction, and they have been verified to be mesh-independent. During the simulation, we monitor the following averaged quantities on the top surface,
\begin{align*}
\bar{\bm F} = \frac{\int_{\Gamma_\mathrm{top}} \bm P \bm N \, d\Gamma_{\bm X} }{\int_{\Gamma_\mathrm{top}}d\Gamma_{\bm X}} \quad \mbox{and} \quad
\bar{\bm E}^\mathrm{e}_{1} = \frac{\int_{\Gamma_\mathrm{top}} \bm E_{1} - \bm E^\mathrm{v}_{1} \, d\Gamma_{\bm X} }{\int_{\Gamma_\mathrm{top}}d\Gamma_{\bm X}}.
\end{align*}
In the above, $\bm P := \bm F \bm S$ is the first Piola–Kirchhoff stress, and $\bm N$ is the unit outward normal vector to the top surface. Here, $\bar{\bm F}$ represents the total traction acting on the top surface, and $\bar{\bm E}^\mathrm{e}_{1}$ denotes the averaged elastic strain over the same surface. The notations $\bar{\bm F}_{(\cdot)}$ and $\bar{\bm E}^\mathrm{e}_{1\:(\cdot)}$ indicate specific components of them.

\begin{table}[h]
\centering
\renewcommand{\arraystretch}{1.2}
\setlength{\tabcolsep}{6pt}
\begin{tabular}{lcccccc}
\toprule
\textbf{Model} & $\bm{\mu^\infty}$(kPa) & $\bm{N^\infty}$ & $\bm{\mu_1}$(kPa) & $\bm{N_1}$ & $\bm{\eta_1}$(kPa$\cdot$s) & $\bm \tau_1$(s)\\
\midrule
FLV-GM & 10.0 & - & 10.0 & -       & 20.0  & 2.0\\
FLV-GKV & 20.0 & - & 20.0 & -           & 80.0 & 4.0\\
NV-GM & 10.0 & 150.0 & 10.0 & 150.0 & 20.0 & - \\
NV-GKV & 20.0 & 150.0 & 20.0 & 150.0 & 80.0 & -\\
\bottomrule
\end{tabular}
\caption{The material parameters adopted in different models.}
\label{tab:FEM_parameters}
\end{table}

\begin{figure}[h]
\centering
\includegraphics[angle=0, trim=0 320 620 70, clip=true, scale = 0.7]{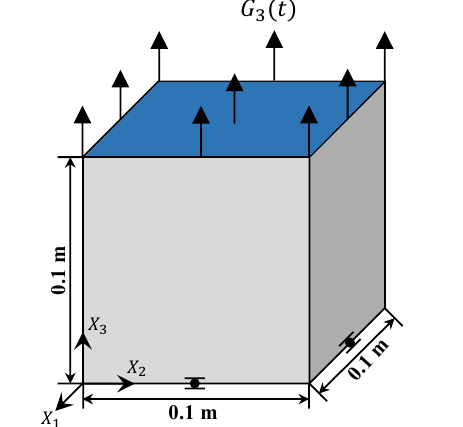}
\includegraphics[angle=0, trim=30 260 60 150, clip=true, scale = 0.4]{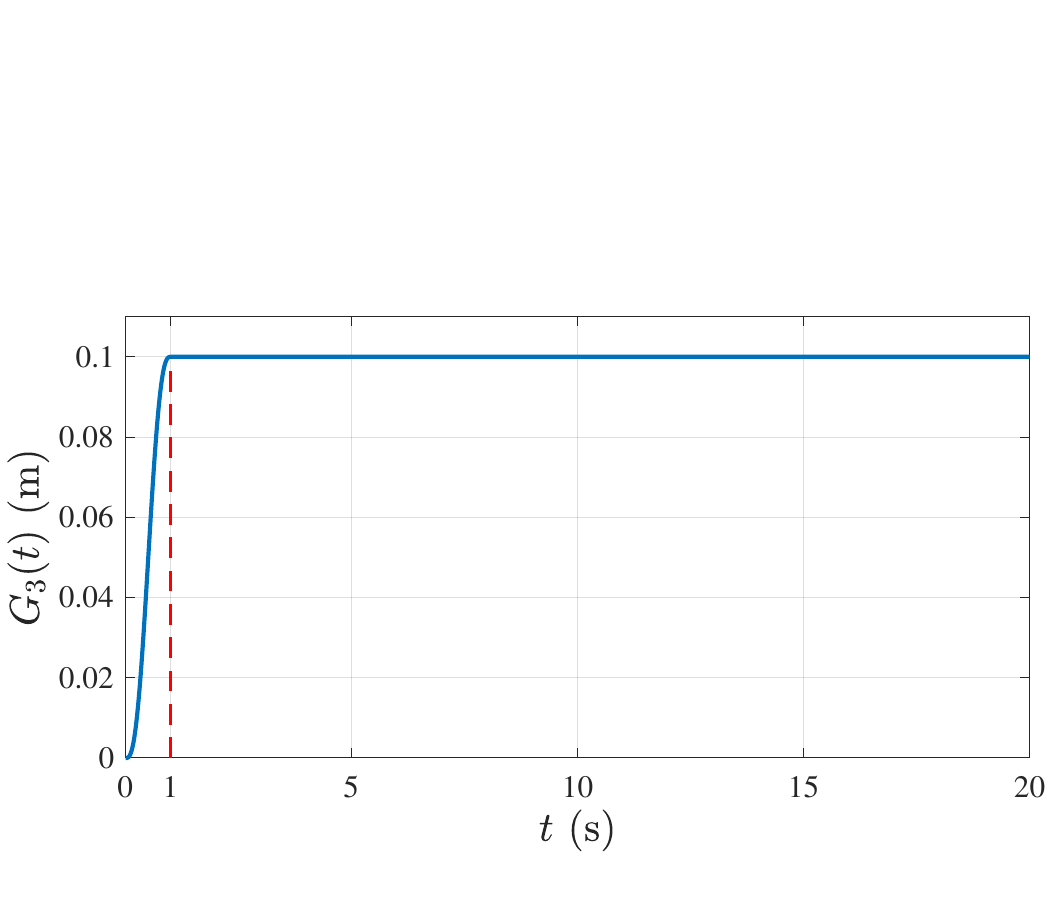}
\caption{Problem settings for the relaxation tests.}
\label{fig:relaxation_setting}
\end{figure}

\begin{figure}[h]
\centering
\includegraphics[angle=0, trim=40 110 30 100, clip=true, scale = 0.6]{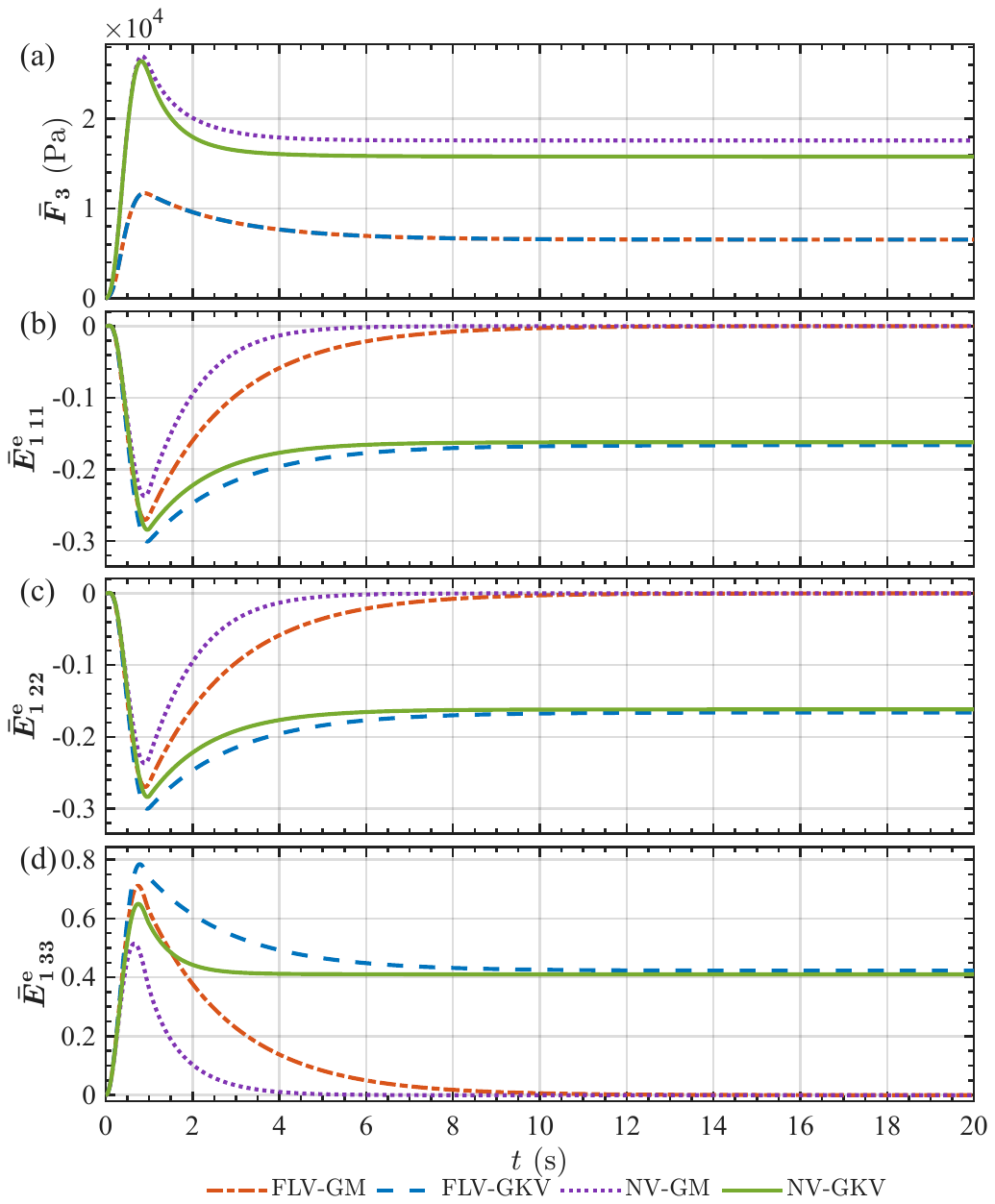}
\caption{The results of the relaxation tests.}
\label{fig:relaxation_result}
\end{figure}

\paragraph{Relaxation test}
The prescribed displacement loading condition is illustrated in Figure \ref{fig:relaxation_setting}. The problem is simulated up to $T = 20~\mathrm{s}$ with a uniform time step size of $\Delta t_n=0.01~\mathrm{s}$. On the bottom surface, the normal component of the displacement and the tangential stress are set to zero; on the top surface, a time-dependent displacement along the normal direction is prescribed as
\begin{align*}
&G_3(t) = 
\begin{cases}
C_1 (t/T)^5 - C_2 (t/T)^4 + C_3(t/T)^3& \mbox{for} \quad 0\le t/T \leq 0.05,\\
C_4& \mbox{for} \quad 0.05 < t/T \leq 1
\end{cases}\\
&\mbox{with}\quad C_1 = 1.92\times 10^6~\mathrm{m},\quad C_2=2.4\times 10^{5}~\mathrm{m},\quad C_3 = 8\times 10^3~\mathrm{m},\quad C_4=0.1~\mathrm{m},
\end{align*}
and the tangential stress is also set to zero; the remaining boundary surfaces are subjected to traction-free conditions.

Figure \ref{fig:relaxation_result} presents the variations of the averaged traction $\bar{\bm F}_3$ along the tensile direction and the components of the elastic strain, i.e. $\bar{\bm E}^\mathrm{e}_{1\:11}$, $\bar{\bm E}^\mathrm{e}_{1\:22}$ and $\bar{\bm E}^\mathrm{e}_{1\:33}$, with respect to time. The identical stress responses between the two finite linear viscoelastic models are confirmed by the overlapping blue dashed and red dash-dot lines in Figure \ref{fig:relaxation_result} (a), corroborating Proposition \ref{proposition1}. We also notice that the nonlinear models exhibit stress responses with discrepancies at the final simulation time, as expected. In the meantime, different evolution histories of the internal variables are observed in Figures \ref{fig:relaxation_result} (b), (c), and (d) for the FLV-GM and FLV-GKV models. In particular, the components of $\bar{\bm E}^\mathrm{e}_1$ in the generalized Maxwell models tend to vanish upon reaching the final simulation time. This aligns well with the property of the equilibrium state \eqref{eq:T_alpha_relax_cond} for the generalized Maxwell model. On the contrary, the components of $\bar{\bm E}^\mathrm{e}_1$ in the generalized Kelvin-Voigt models approach identical non-zero values at the final time. Recall that $\bm T^\infty|_\mathrm{eq} = \bm T^\mathrm{neq}_\alpha|_\mathrm{eq}$, for $\alpha=1,\cdots, M$, in the equilibrium state for the generalized Kelvin-Voigt models. Since the prescribed material parameters are identical for the equilibrium and non-equilibrium components, the balance of the thermodynamic forces implies that elastic strain components are close in the equilibrium state, and their values are approximately half of the maximum strain in both FLV-GKV and NV-GKV models. The simulation results further show that the nonlinear models produce substantially higher stress responses compared to the finite linear viscoelastic models. Also, the internal variables in the nonlinear models evolve at faster rates, indicating a distinct dynamic response behavior with pronounced energy dissipation.

\begin{table}[htbp]
\centering
\begin{tabular}{ccccc}
\toprule
& \multicolumn{4}{c}{$\mathfrak l_2$-norm of the residual vector} \\
\cmidrule(lr){2-5}
{Iteration} & {FLV-GM} & {FLV-GKV} & {NV-GM} & {NV-GKV} \\
& {66-th time step} & {16-th time step} & {88-th time step}  & {5-th time step}\\
\midrule
1 & $1.430 \times 10^{-1}$ & $8.570 \times 10^{-2}$ & $8.049 \times 10^{-2}$ &  $6.618 \times 10^{-2}$ \\
2 & $5.997 \times 10^{-5}$ & $1.503 \times 10^{-4}$ & $5.475 \times 10^{-5}$ & $3.538 \times 10^{-5}$\\
3 & $2.079 \times 10^{-10}$& $9.617 \times 10^{-10}$ & $5.743 \times 10^{-11}$ & $4.415 \times 10^{-11}$\\
4 & $1.842\times 10^{-14}$ & $9.258\times 10^{-15}$ & - & - \\
\bottomrule
\end{tabular}
\caption{The convergence of the Newton-Raphson iteration for relaxation tests.}
\label{tab:relaxation_convergence}
\end{table}

We also present the $\mathfrak l_2$-norms of the residual vectors in the Newton-Raphson iteration at different time steps during the loading phase in Table \ref{tab:relaxation_convergence}. Notably, convergence is achieved within only three to four iterations across all models, and the residual decay is consistent with the expected quadratic convergence rate.

\paragraph{Cyclic shear test}

\begin{figure}[h]
\centering
\includegraphics[angle=0, trim= 0 330 620 70, clip=true, scale = 0.7]{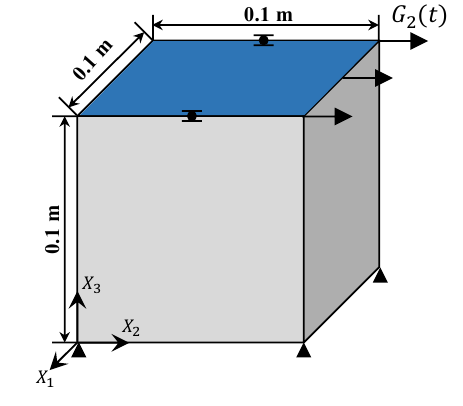}
\includegraphics[angle=0, trim=20 250 60 15, clip=true, scale = 0.4]{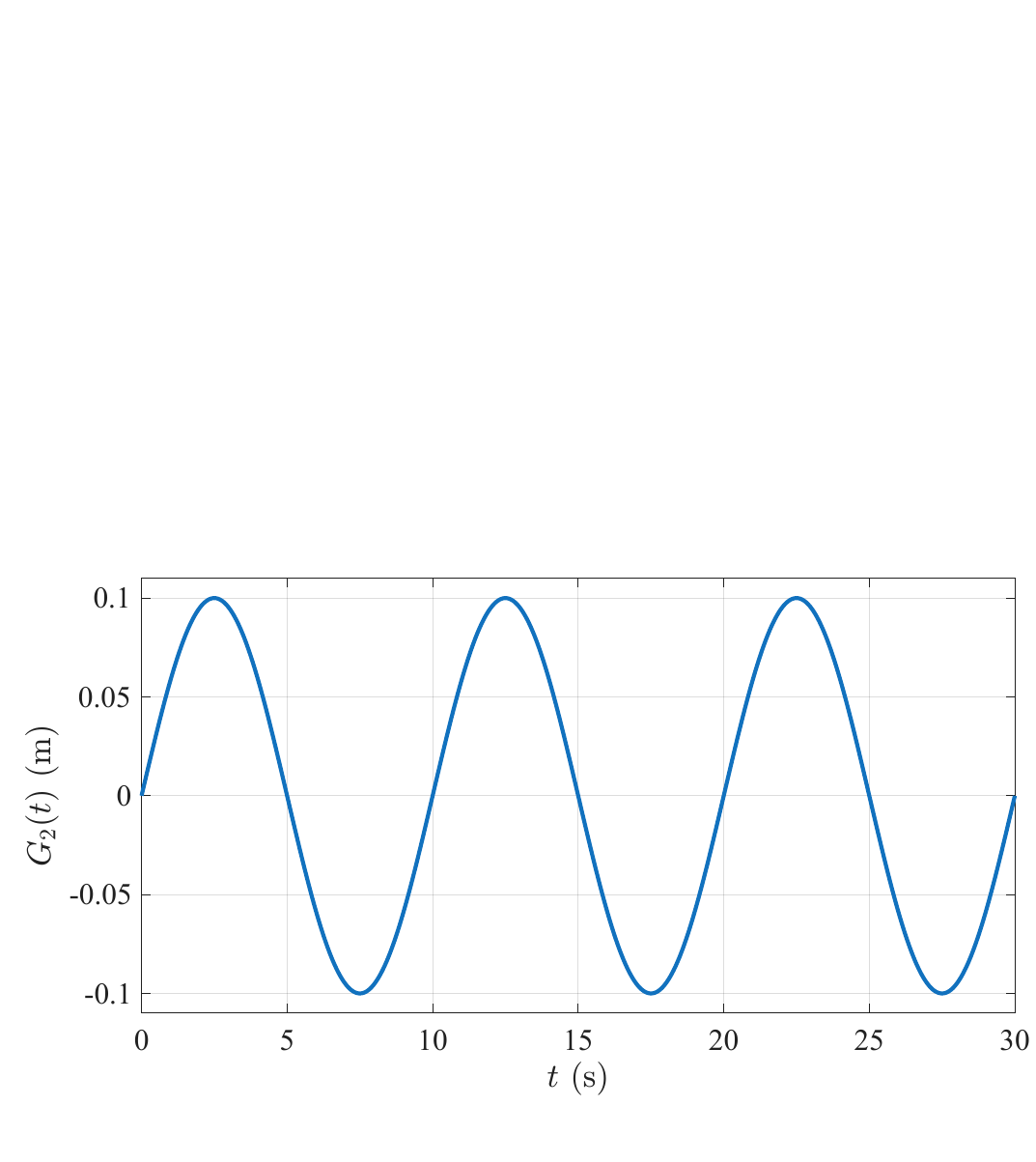}
\caption{Problem settings for the cyclic shear tests.}
\label{fig:shear_setting}
\end{figure}

The prescribed displacement follows a periodic function with a period of $10$~s, and three full cycles are considered. This results in a total simulation time of $T=30$~s. A uniform time step size of $\Delta t_n = 0.01$~s is adopted throughout the simulation. The boundary conditions for the cyclic shear test are specified as follows. The bottom surface is fully clamped with zero displacement; the top surface is subjected to a cyclic displacement loading, specified as
\begin{align*}
G_1(t) = 0, \quad
G_2(t) = C_5 \sin\left(C_6 \left(t/T\right)\right)
\quad\mbox{with}\quad
C_5 = 0.1~\mathrm{m},
\quad
C_6=6\pi
\quad\mbox{for}\quad
0\le t/T\le 1,
\end{align*}
and the traction in the remaining direction is set to zero; on the other four surfaces, we impose $G_1(t) = 0$ and traction-free in the remaining directions. The geometric configuration and the displacement boundary condition are depicted in Figure~\ref{fig:shear_setting}.

\begin{figure}[h]
\centering
\includegraphics[angle=0, trim=40 110 30 100, clip=true, scale = 0.6]{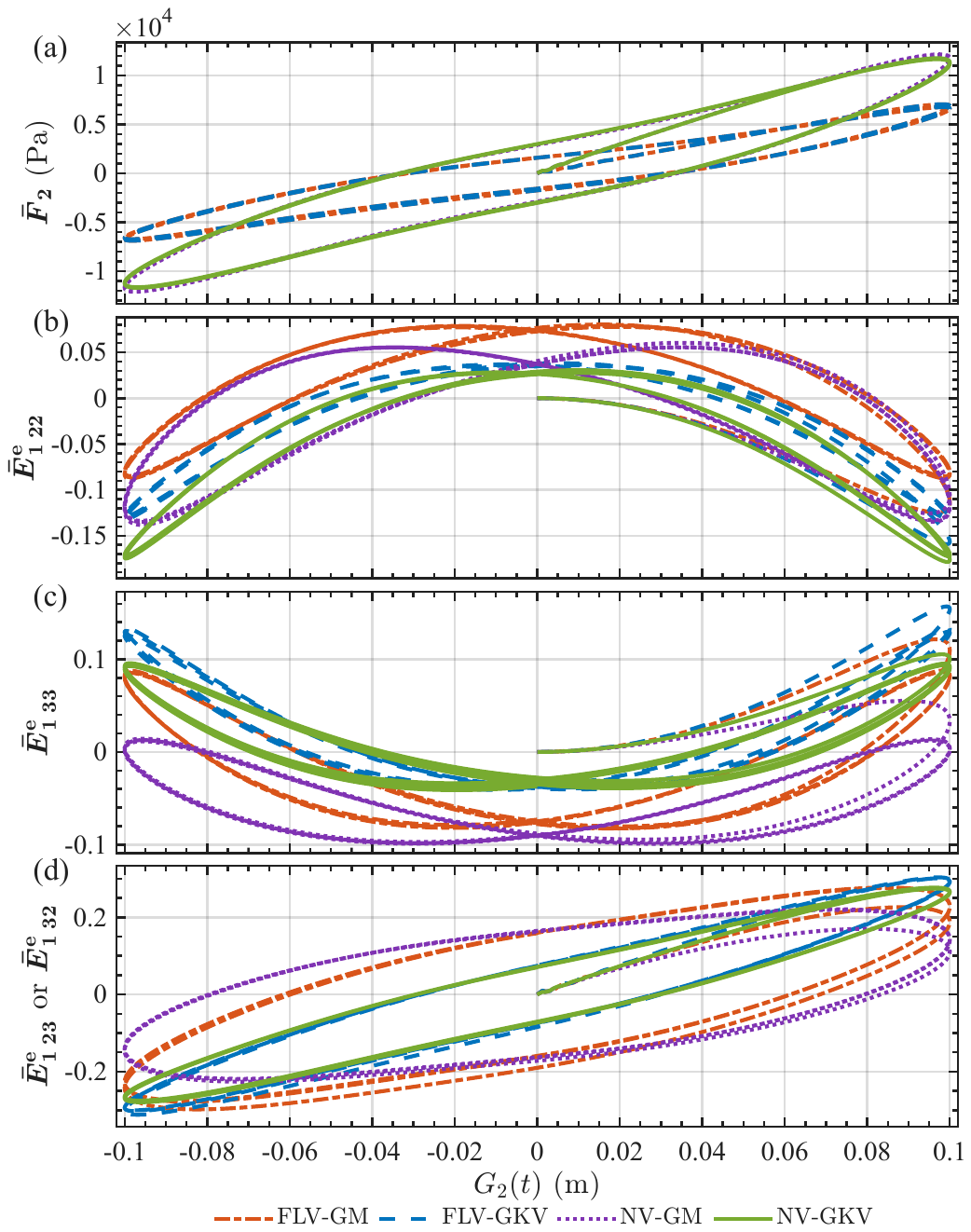}
\caption{The results of the cyclic shear tests.}
\label{fig:shear_result}
\end{figure}

\begin{table}[h]
\centering
\begin{tabular}{ccccc}
\toprule
& \multicolumn{4}{c}{$\mathfrak l_2$-norm of the residual vector} \\
\cmidrule(lr){2-5}
{Iteration} & {FLV-GM} & {FLV-GKV} & {NV-GM} & {NV-GKV} \\
& {976-th time step} & {1278-th time step} & {1995-th time step}  & {112-th time step}\\
\midrule
1 & $4.053 \times 10^{-2}$ & $3.966 \times 10^{-2}$ & $6.149 \times 10^{-2}$ &  $4.949 \times 10^{-2}$ \\
2 & $1.208 \times 10^{-6}$ & $1.568 \times 10^{-6}$ & $1.578 \times 10^{-6}$ & $1.002 \times 10^{-6}$\\
3 & $1.872 \times 10^{-13}$ & $5.936 \times 10^{-13}$ & $4.026 \times 10^{-13}$ & $3.922 \times 10^{-14}$\\
\bottomrule
\end{tabular}
\caption{The convergence of the Newton-Raphson iteration for the cyclic shear tests.}
\label{tab:shear_convergence}
\end{table}

The results of the cyclic shear tests are presented in Figure \ref{fig:shear_result}. The identical stress responses of two finite linear viscoelastic models is again confirmed by the overlapping red dash-dot and blue dashed lines in Figure \ref{fig:shear_result} (a). Additionally, the nonlinear models exhibit nearly identical hysteresis loops in the traction response, with minor differences observed in this case. All four models achieve a dynamic equilibrium state after the first quarter cycle, and the hysteresis loops of the finite linear viscoelastic models are noticeably narrower than those of the NV models, indicating lower energy dissipation in the finite linear viscoelastic models.
 
Figures \ref{fig:shear_result} (b), (c), and (d) depict the temporal evolution of the components of $\bar{\bm E}^\mathrm{e}_1$ for different models. The generalized Kelvin-Voigt models exhibit larger amplitudes and narrower hysteresis loops compared to the generalized Maxwell models, which can be attributed to the differences in the viscosity $\eta_1$. Moreover, the response curves of the generalized Kelvin-Voigt models converge and stabilize from the first half-cycle onward, whereas those of the generalized Maxwell models do not achieve overlap until the beginning of the second cycle. This behavior indicates that the generalized Kelvin–Voigt models adjust their internal variables more rapidly, reaching dynamic equilibrium sooner.

We monitor the convergence histories of the Newton-Raphson iteration. Table \ref{tab:shear_convergence} reports the $\mathfrak l_2$-norms of the residual vectors at selected time steps. It can be observed that convergence is achieved within three iterations. This observation is again consistent with the quadratic convergence rate, which is expected of a consistent linearization.

\subsection{Algorithm evaluation of the constitutive integration for the nonlinear generalized Kelvin-Voigt model}
This section examines the efficiency of the algorithms for determining the incremental updates during constitutive integration in the NV-GKV model. The evaluation is conducted with a fixed time step size of $0.01$ s. The deformation gradient of the previous step is set to the identity, and all internal variables are initialized to zero. The deformation gradient of the current step is generated randomly with a positive determinant. Simulations are performed for $M=1,2,4,8$, and $16$. Each test is repeated $5000$ times to obtain a statistically averaged runtime required for the solution of \eqref{eq:tensor_matrix}, solving \eqref{eq:solve_NN_alpha} for $\mathbb{N}_\alpha$, and solving \eqref{eq:solve_XX} for $\bm{X}$.

\begin{figure}[h]
\centering
\includegraphics[trim=40 270 70 270, clip,  scale=0.7]{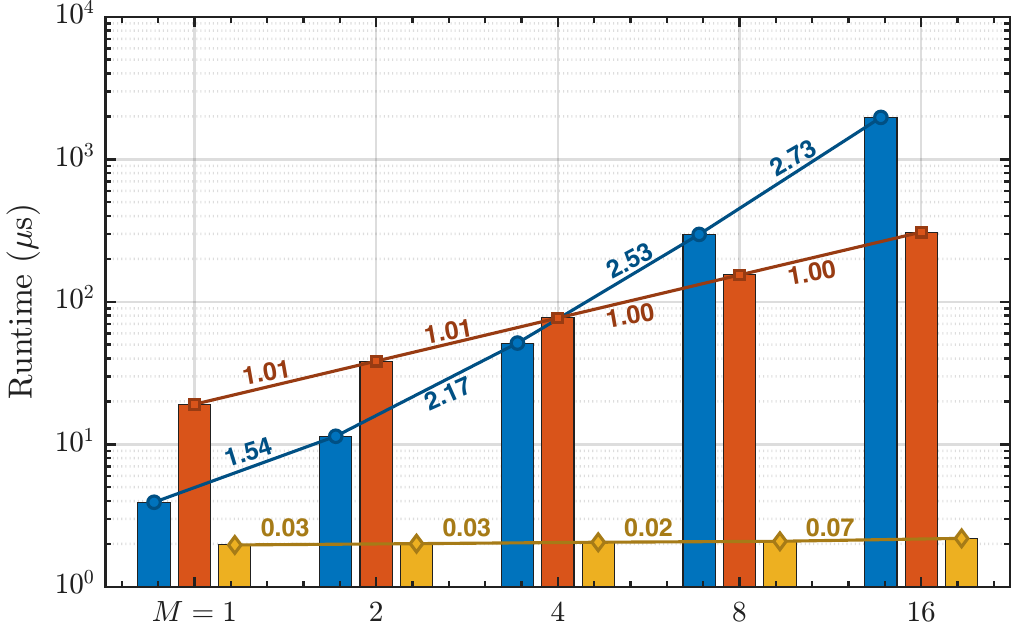}
\caption{Compute time consumed by the solution of \eqref{eq:tensor_matrix} (blue), solving \eqref{eq:solve_NN_alpha} for $\mathbb{N}_\alpha$ (red), and solving \eqref{eq:solve_XX} for $\bm{X}$ (yellow), with the growth rates obtained from the slopes of log-log fits.}
\label{fig:algorithm_compare}
\end{figure}

All reported average times are plotted on a logarithmic scale in Figure~\ref{fig:algorithm_compare}. The time consumed by solving \eqref{eq:tensor_matrix} increases approximately cubically with $M$, whereas the total cost of solving $\mathbb{N}_\alpha$ grows linearly with $M$. The computational cost of solving $\bm{X}$ remains essentially independent of $M$. When $M$ reaches $8$, the total computational times of the two approaches become nearly identical. This finding is consistent with the theoretical analysis made in Section \ref{sec:numerical_formulation_FV_KV} and indicates the algorithm based on Proposition \ref{prop:tensor_matrix} enjoys a better scaling property with the number of $M$. It is thus recommended for use in practical implementation.

\subsection{Calibration and simulation of VHB 4910}
\label{sec:calibration}
In this section, we consider VHB 4910, an acrylic polymer that has attracted attention due to its potential use as a dielectric elastomer material. Its mechanical behavior exhibits strong rate-dependency, and its viscoelastic behavior has been characterized using the model based on the multiplicative decomposition \cite{Hossain2012}. We conducted uniaxial tests on VHB 4910 specimens using an Instron 5966 universal testing machine. Rectangular specimens, with dimensions of $10~\mathrm{cm} \times 2.2~\mathrm{cm} \times 1~\mathrm{mm}$, were prepared by bonding $1$~cm-long acrylic tabs to both ends along the tensile direction, resulting in an active gauge length of $8$~cm. In each test, the universal testing machine applied a uniaxial stretch until a predetermined stretch level was reached, followed by unloading until the measured force returned to zero. The prescribed maximum displacements were $4$, $8$, $12$, and $16$~cm, corresponding to stretch ratios of $1.5$, $2.0$, $2.5$, and $3.0$, respectively. A constant stretch rate of $0.05$~s$^{-1}$ was employed in all four experiments. The proposed four models (FLV-GM, FLV-GKV, NV-GM, NV-GKV) were employed to fit four sets of experimental data, using $M=1$ or $2$. To identify the material parameters and evaluate model performance, the normalized mean absolute difference (NMAD) was adopted as the metric \cite[p.~444]{Bergstrom2015}, where a lower NMAD value indicates a better fit.

\begin{figure}[h]
\centering
\begin{minipage}{0.4\textwidth}
\centering
\textbf{FLV-GM $(M=1)$}
\end{minipage}%
\begin{minipage}{0.4\textwidth}
\centering
\textbf{FLV-GM $(M=2)$}
\end{minipage}
\\
\begin{minipage}{0.4\textwidth}
\centering
\includegraphics[width=\linewidth, trim=110 140 170 140, clip]{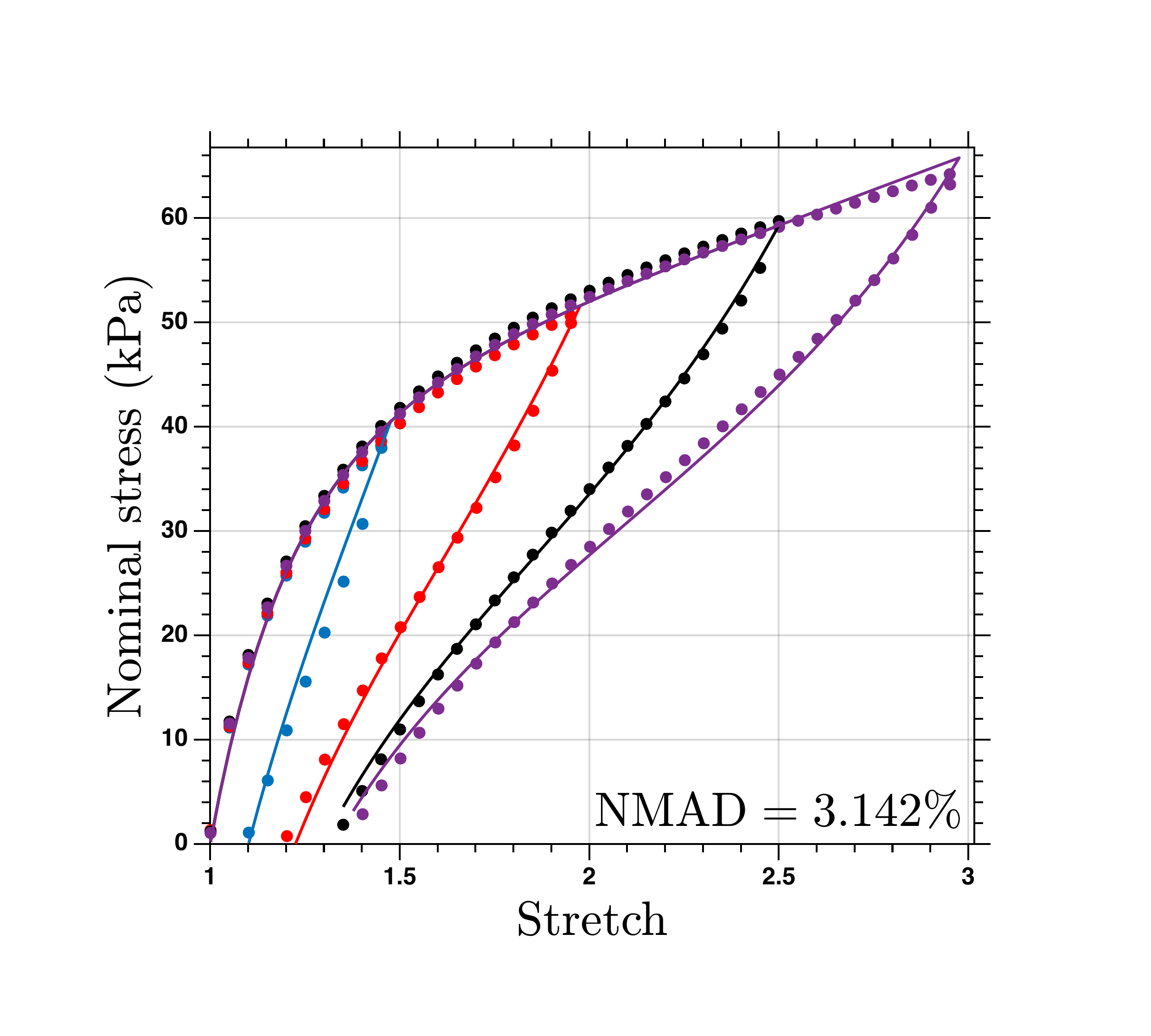}  
\end{minipage}%
\begin{minipage}{0.38\textwidth}
\centering
\includegraphics[width=\linewidth, trim=157 140 170 140, clip]{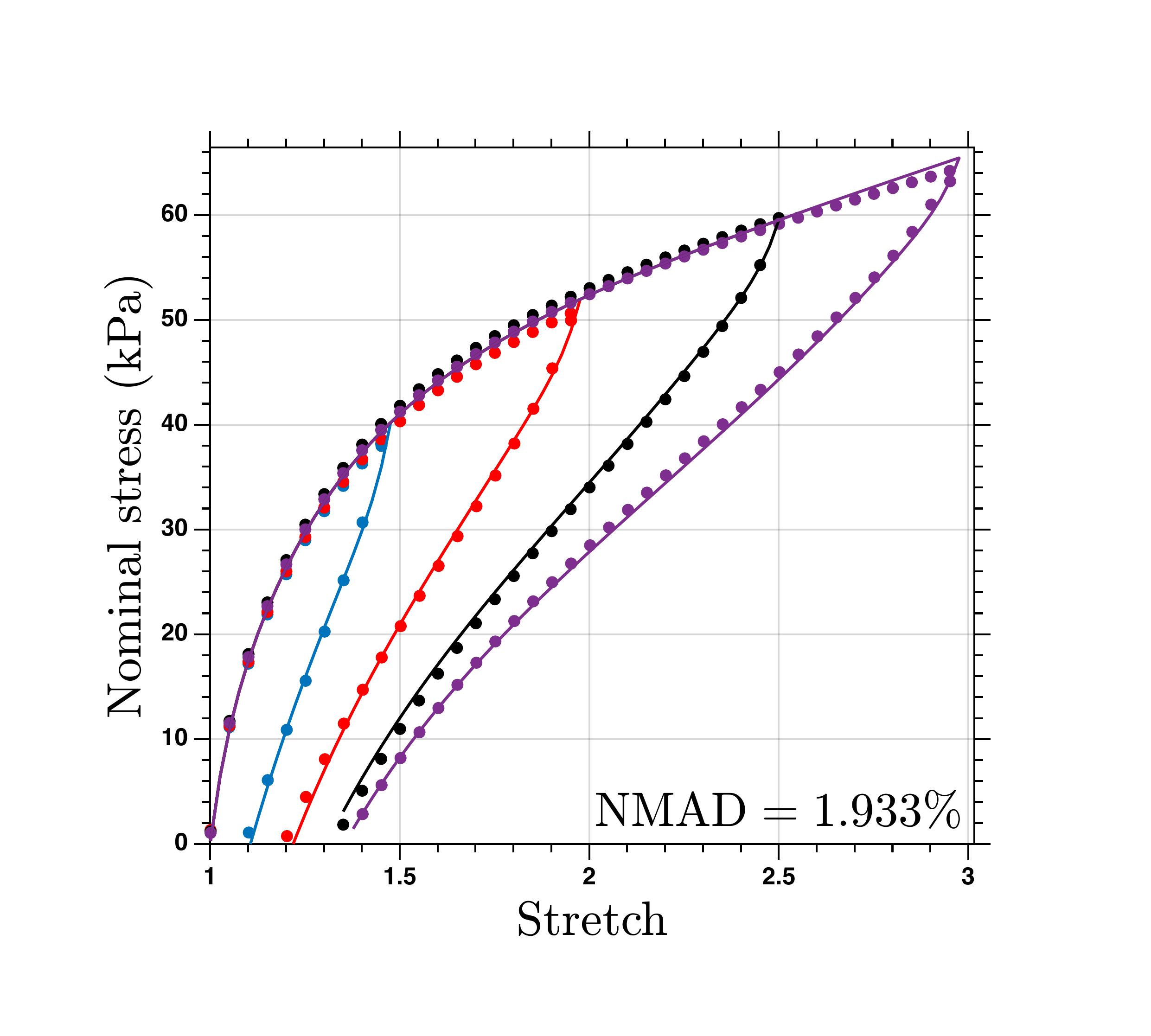}  
\end{minipage}%
\\
\begin{minipage}{0.4\textwidth}
\centering
\textbf{FLV-GKV $(M=1)$}
\end{minipage}%
\begin{minipage}{0.4\textwidth}
\centering
\textbf{FLV-GKV $(M=2)$}
\end{minipage}
\\
\begin{minipage}{0.4\textwidth}
\centering
\includegraphics[width=\linewidth, trim=110 100 170 140, clip]{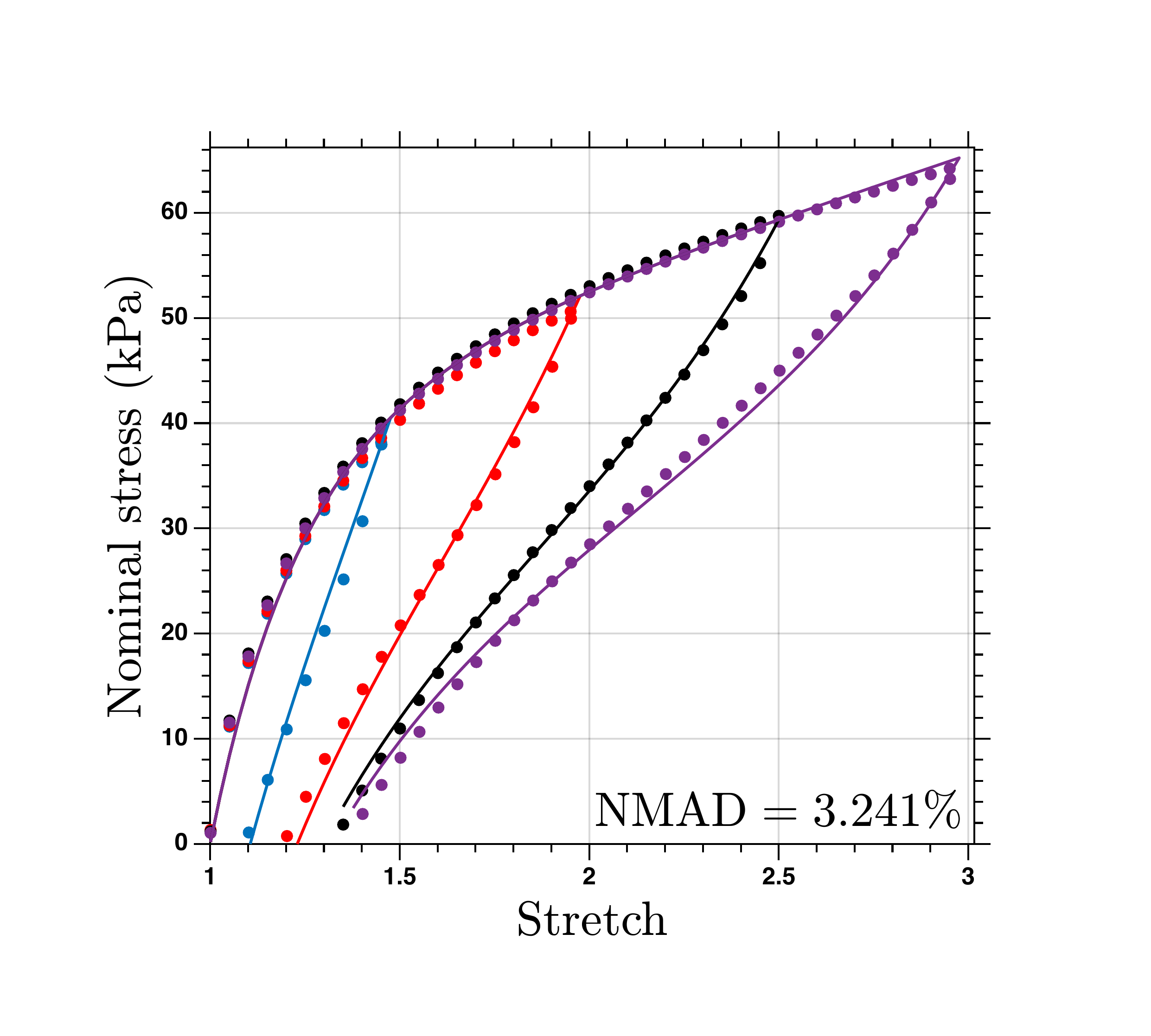}  
\end{minipage}%
\begin{minipage}{0.38\textwidth}
\centering
\includegraphics[width=\linewidth, trim=157 100 170 140, clip]{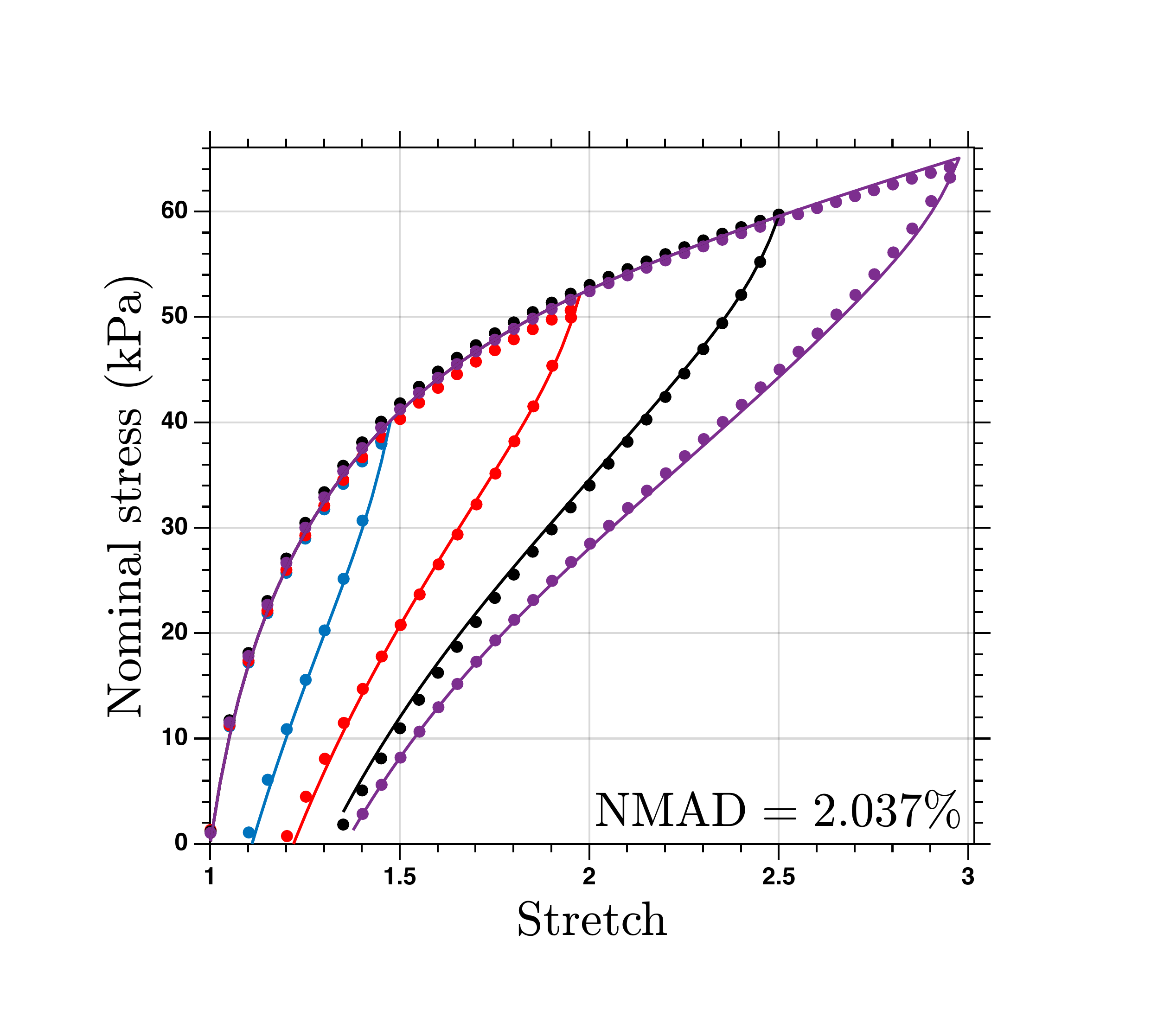}  
\end{minipage}%
\caption{Calibration results of the finite linear viscoelastic models.}
\label{fig:FLV_calibration}
\end{figure}

\begin{figure}[hp]
\centering
\begin{minipage}{0.4\textwidth}
\centering
\textbf{NV-GM $(M=1)$}
\end{minipage}%
\begin{minipage}{0.4\textwidth}
\centering
\textbf{NV-GM $(M=2)$}
\end{minipage}
\\
\begin{minipage}{0.4\textwidth}
\centering
\includegraphics[width=\linewidth, trim=110 140 170 140, clip]{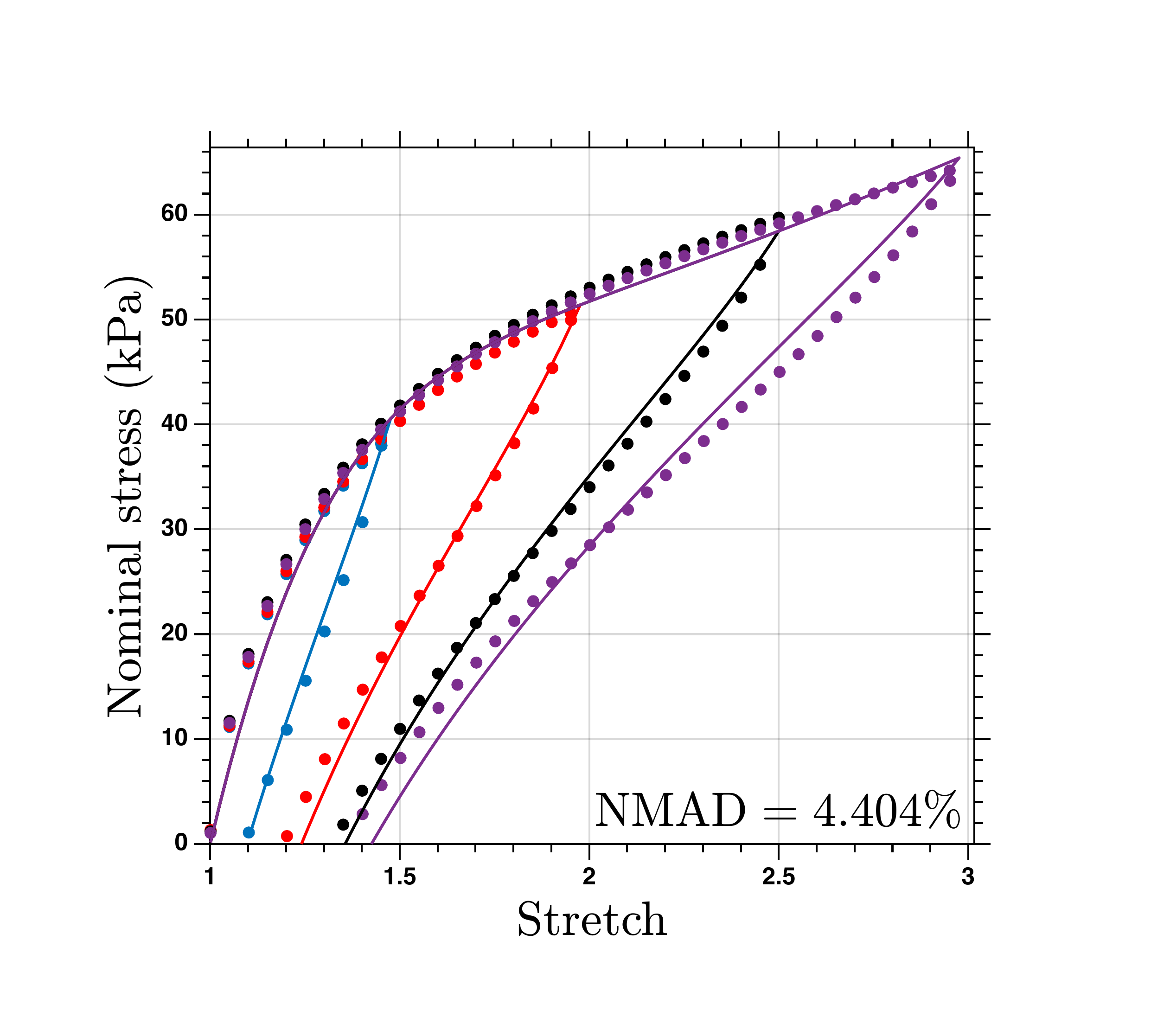}  
\end{minipage}%
\begin{minipage}{0.38\textwidth}
\centering
\includegraphics[width=\linewidth, trim=157 140 170 140, clip]{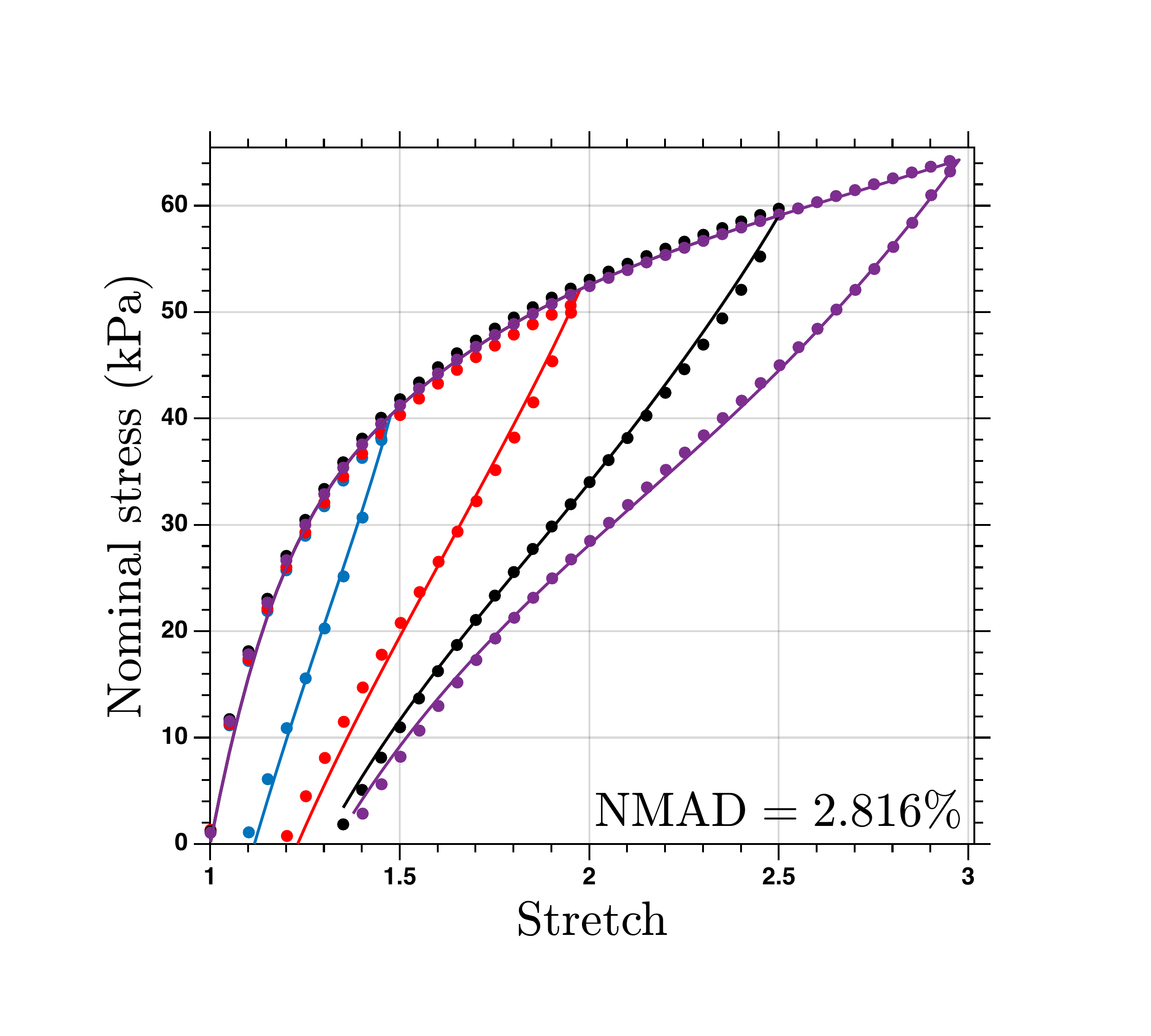}  
\end{minipage}%
\\
\begin{minipage}{0.4\textwidth}
\centering
\textbf{NV-GKV $(M=1)$}
\end{minipage}%
\begin{minipage}{0.4\textwidth}
\centering
\textbf{NV-GKV $(M=2)$}
\end{minipage}
\\
\begin{minipage}{0.4\textwidth}
\centering
\includegraphics[width=\linewidth, trim=110 100 170 140, clip]{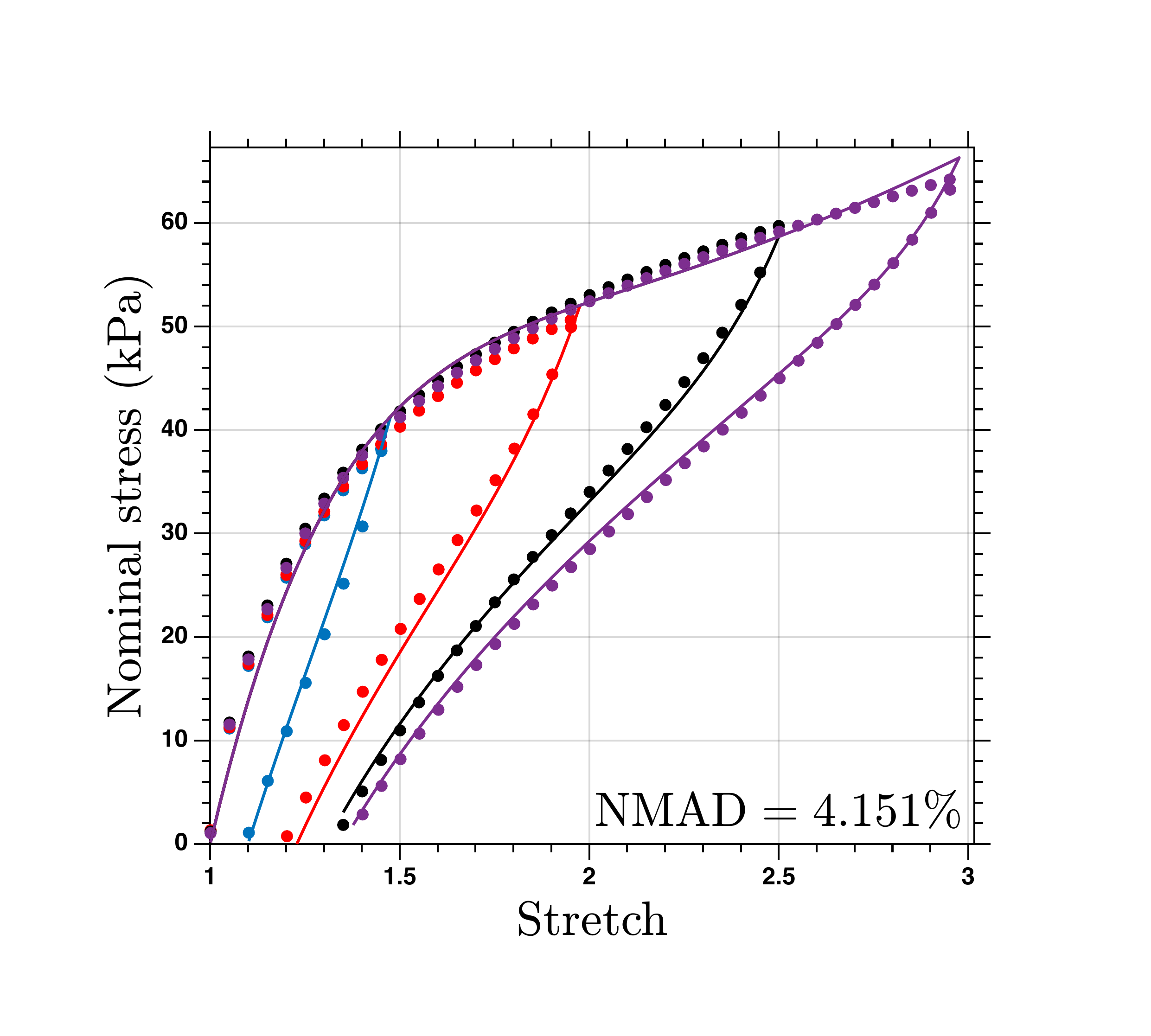}  
\end{minipage}%
\begin{minipage}{0.38\textwidth}
\centering
\includegraphics[width=\linewidth, trim=157 100 170 140, clip]{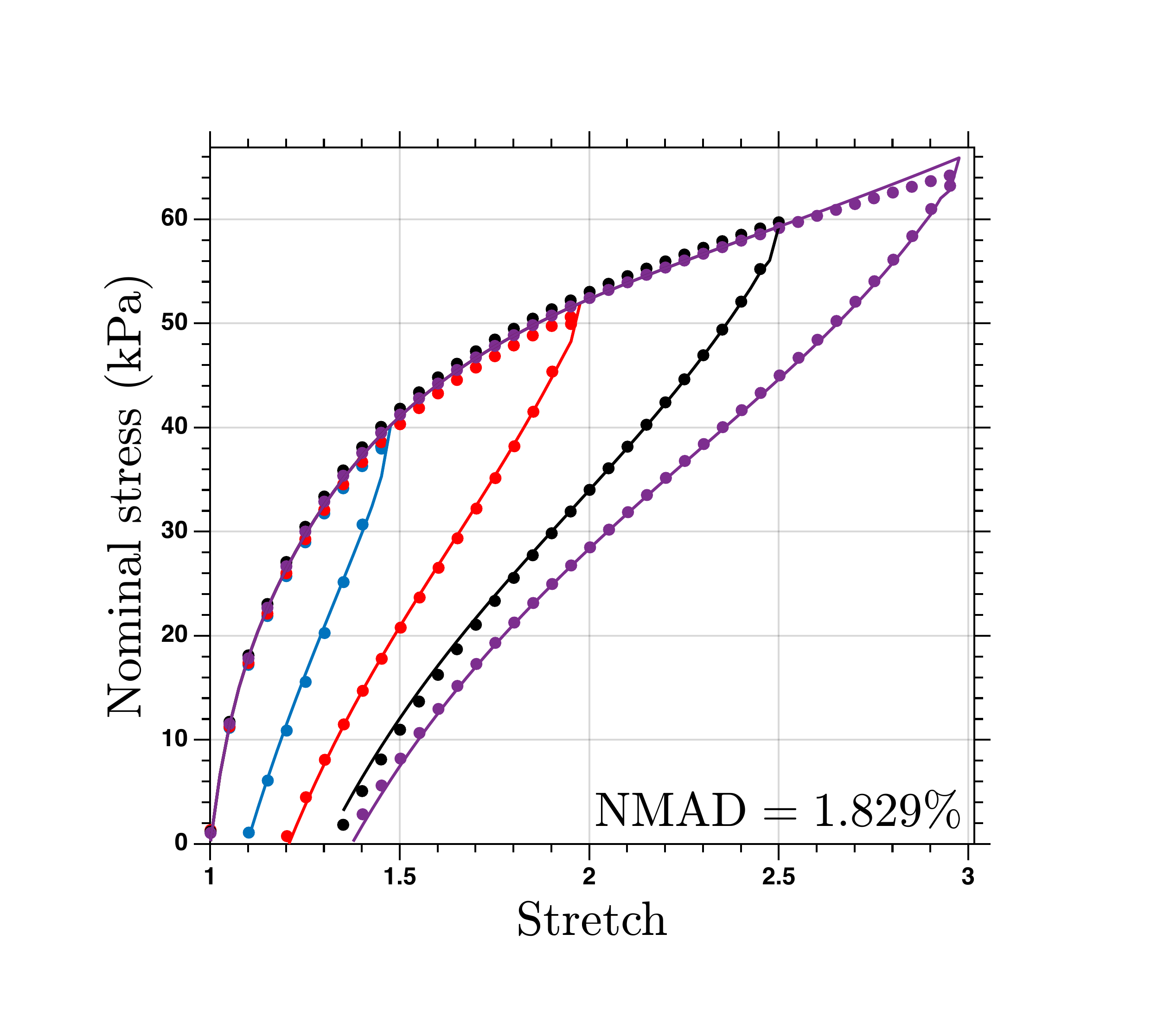}  
\end{minipage}%
\caption{Calibration results of the nonlinear viscoelastic models.}
\label{fig:NV_calibration}
\end{figure}

\begin{table}[h]
\centering
\renewcommand{\arraystretch}{1.2}
\setlength{\tabcolsep}{8pt}
\small
\begin{tabular}{c|c|ccccccc}
\toprule
\multicolumn{1}{c}{} & 
\multicolumn{1}{c}{\textbf{Model}} & 
\multicolumn{1}{c}{$\bm{\mu^\infty}$} & 
\multicolumn{1}{c}{$\bm{m}$} & 
\multicolumn{1}{c}{$\bm{n}$} & 
\multicolumn{1}{c}{$\bm{\mu_1}$} & 
\multicolumn{1}{c}{$\bm{\tau_1}$} & 
\multicolumn{1}{c}{$\bm{\mu_2}$} & 
\multicolumn{1}{c}{$\bm{\tau_2}$} \\
		
\cmidrule(lr){2-9}
\multicolumn{1}{c}{} & 
\multicolumn{1}{c}{} & 
\multicolumn{1}{c}{(kPa)} & 
\multicolumn{1}{c}{} & 
\multicolumn{1}{c}{} & 
\multicolumn{1}{c}{(kPa)} & 
\multicolumn{1}{c}{(s)} & 
\multicolumn{1}{c}{(kPa)} & 
\multicolumn{1}{c}{(s)} \\
\midrule
		
\multirow{2}{*}{\rotatebox{90}{$M=1$}} 
& \textbf{FLV-GM} & 67.20 & 0.94 & 1.67 & 67.39 & 9.44 & -- & -- \\
\addlinespace
& \textbf{FLV-GKV} & 122.56 & 1.0 & 1.13 & 124.58 & 17.48 & -- & -- \\
\midrule		
\multirow{2}{*}{\rotatebox{90}{$M=2$}} 
& \textbf{FLV-GM} & 61.07 & 1.0 & 1.28 & 46.31 & 13.42 & 86.46 & 0.78 \\
\addlinespace
& \textbf{FLV-GKV} & 166.04 & 0.96 & 0.97 & 3.45 & 460.84 & 369.69 & 1.85 \\
\bottomrule
\end{tabular}
\caption{Fitted material parameters for the finite linear viscoelastic models.}
\label{tab:FLV_calibration_parameters}
\end{table}

\begin{table}[h]
\centering
\renewcommand{\arraystretch}{1.2}
\setlength{\tabcolsep}{6pt}
\small
\begin{tabular}{c|c|cccccccccc}
\toprule
\multicolumn{1}{c}{} & 
\multicolumn{1}{c}{\textbf{Model}} & 
\multicolumn{1}{c}{$\bm{\mu^\infty}$} & 
\multicolumn{1}{c}{$\bm{N^\infty}$} & 
\multicolumn{1}{c}{$\bm{m}$} & 
\multicolumn{1}{c}{$\bm{n}$} & 
\multicolumn{1}{c}{$\bm{\mu_1}$} & 
\multicolumn{1}{c}{$\bm{N_1}$} & 
\multicolumn{1}{c}{$\bm{\eta_1}$} & 
\multicolumn{1}{c}{$\bm{\mu_2}$} & 
\multicolumn{1}{c}{$\bm{N_2}$} & 
\multicolumn{1}{c}{$\bm{\eta_2}$} \\
		
\cmidrule(lr){2-11}
\multicolumn{1}{c}{} & 
\multicolumn{1}{c}{} & 
\multicolumn{1}{c}{(kPa)} & 
\multicolumn{1}{c}{} & 
\multicolumn{1}{c}{} & 
\multicolumn{1}{c}{} & 
\multicolumn{1}{c}{(kPa)} & 
\multicolumn{1}{c}{} & 
\multicolumn{1}{c}{(kPa$\cdot$s)} & 
\multicolumn{1}{c}{(kPa)} & 
\multicolumn{1}{c}{} & 
\multicolumn{1}{c}{(kPa$\cdot$s)} \\
\midrule
	
\multirow{2}{*}{\rotatebox{90}{$M=1$}} 
& \textbf{NV-GM} & 19.35 & 1834.19 & 0.26 & 0.48 & 31.74 & 21370.59 & 1106.37 & -- & -- & -- \\
\addlinespace
& \textbf{NV-GKV} & 52.19 & 79256.09 & 0.79 & 1.0 & 38.21 & 3.02 & 2689.13 & -- & -- & -- \\
\midrule		
\multirow{2}{*}{\rotatebox{90}{$M=2$}} 
& \textbf{NV-GM} & 16.74 & 414933.62 & 0.39 & 1.17 & 12.76 & 507.40 & 4030.58 & 31.24 & 21529.04 & 441.10 \\
\addlinespace
& \textbf{NV-GKV} & 88.35 & 275836.38 & 0.96 & 1.85 & 125.14 & 11250.99 & 477.36 & 50.45 & 547.77 & 5372.63 \\
\bottomrule
\end{tabular}
\caption{Fitted material parameters for the nonlinear viscoelastic models.}
\label{tab:NV_calibration_parameters}
\end{table}

The fitting results of the finite linear viscoelastic models are presented in Figure \ref{fig:FLV_calibration}. All four models demonstrate satisfactory performance in capturing the loading–unloading behavior. Notably, the inclusion of an additional non-equilibrium process in the model (i.e., $M=2$) improves the agreement with the experimental data. The fitted parameters are listed in Table \ref{tab:FLV_calibration_parameters}, which allows for a validation of Proposition \ref{proposition1}. For $M=1$, the calibrated parameters approximately satisfy \eqref{eq:equivalence_moduli_relation}. This confirms that the two rheological representations yield equivalent stress responses, and it also demonstrates the robustness of the models as well as the calibration method. The fitting results of the nonlinear models are presented in Table \ref{tab:NV_calibration_parameters} and Figure \ref{fig:NV_calibration}. Compared to the finite linear viscoelastic models, the nonlinear models offer only limited improvement in fitting accuracy. An exception is the NV-GKV model with $M=2$, which attains the best accuracy among all models.

\begin{figure}[h]
\centering
\includegraphics[angle=0, trim=0 170 460 70, clip=true, scale = 0.55]{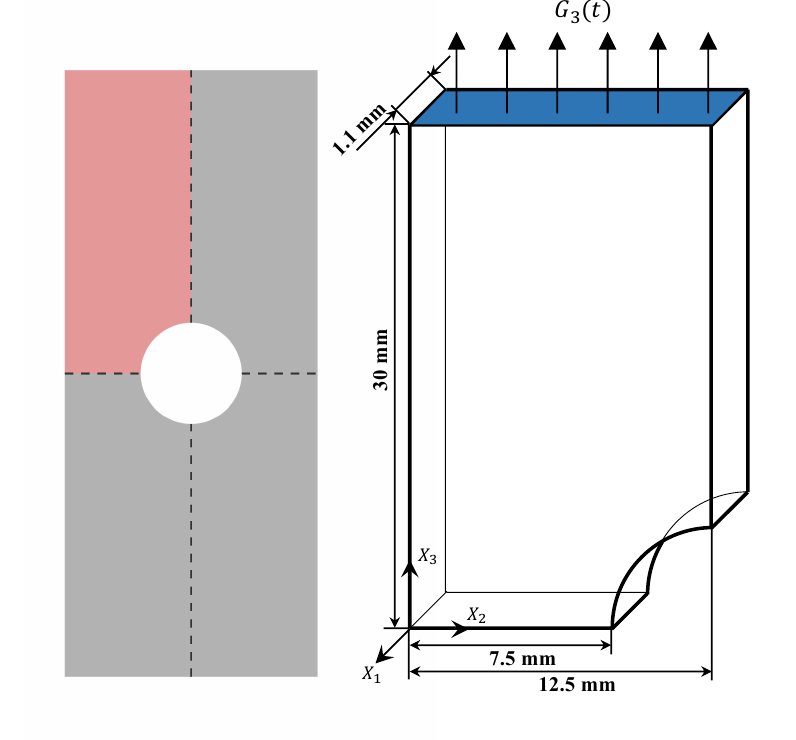}
\includegraphics[angle=0, trim=50 200 60 300, clip=true, scale = 0.4]{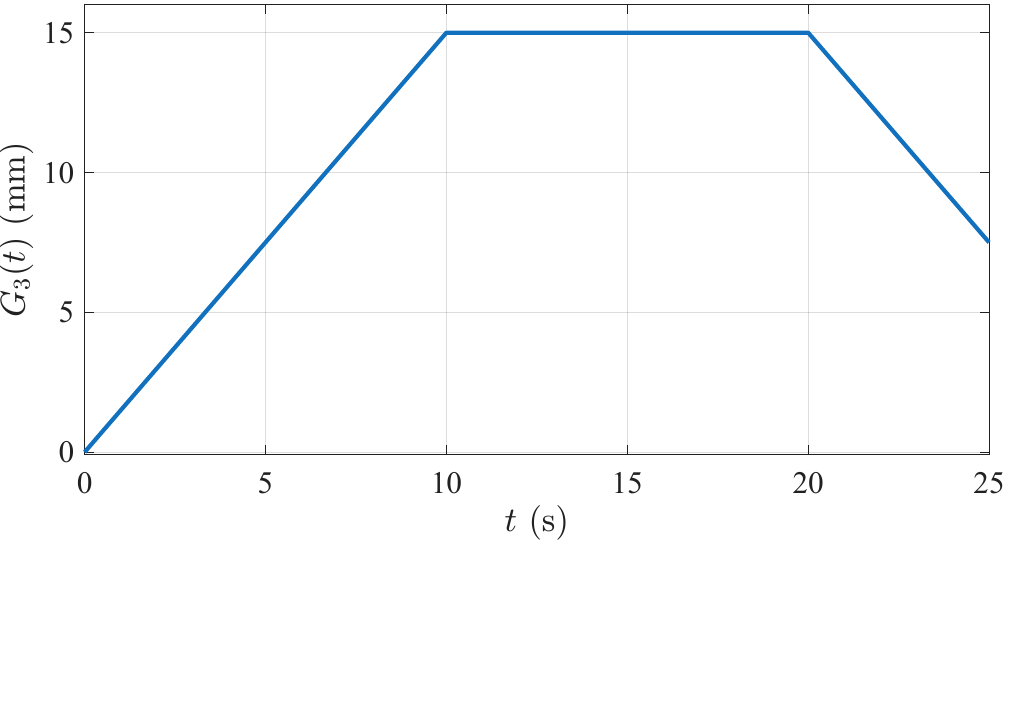}
\caption{Plate with a central hole: geometry and loading conditions.}
\label{fig:plate_hole}
\end{figure}

Next, we perform a numerical study for a plate with a central hole (Figure \ref{fig:plate_hole}), using the fitted parameters of the four models with $M=2$. The size of the plate is $60$~mm~$\times$~$25$~mm~$\times$~$2.2$~mm, with a circular hole of radius $5$~mm located at its center. The simulation is performed up to $T=25$~s. During the first $10$~s, the two ends of the plate along its 60 mm length are stretched at a constant rate of $1.5$~mm/s; from $10$~s to $20$~s, the applied displacement is held fixed to observe relaxation; in the final $5$~s, the stretched ends are compressed at the same rate. This loading protocol results in a constant stretch rate of 0.05~s$^{-1}$ and a maximum stretch ratio of $1.5$. To reduce the computational cost, one-eighth of the plate is considered by leveraging symmetry boundary conditions. The mesh size is $15\times 15\times 5$, and mesh independence of the results has been verified. Symmetry boundary conditions are applied on the three mid-planes; traction-free boundary conditions are applied on the two lateral surfaces; on the loaded surface, a prescribed displacement is applied as
\begin{align*}
G_3(t) = 
\begin{cases}
C_7 t/T, & 0 \leq t/T \leq 0.4, \\
C_8, & 0.4 < t/T \leq 0.8, \\
C_8 - C_7(t/T - 0.8), & 0.8 < t/T \leq 1,
\end{cases}
\end{align*}
with $C_7=37.5$~mm and $C_8=15$~mm.

\begin{figure}[h]
\centering
\includegraphics[angle=0, trim=40 270 60 280, clip=true, scale = 0.65]{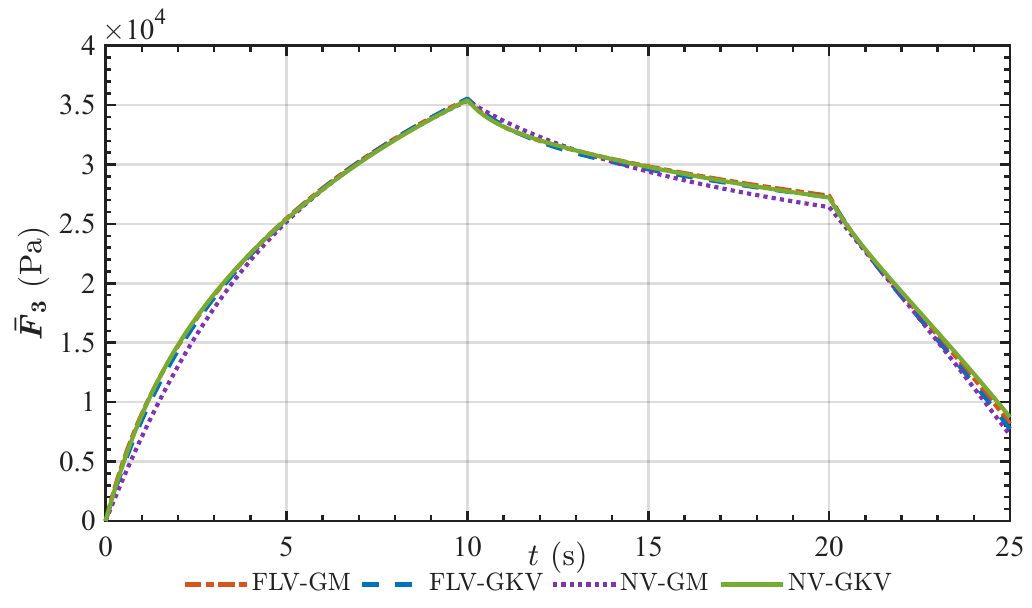}
\caption{The time-dependent traction responses predicted by the four different models.}
\label{fig:time_traction}
\end{figure}

\begin{figure}[h]
\centering
\begin{minipage}{0.25\textwidth}
\centering
\textbf{FLV-GM}
\end{minipage}%
\begin{minipage}{0.24\textwidth}
\centering
\textbf{NV-GM}
\end{minipage}%
\begin{minipage}{0.27\textwidth}
\centering
\textbf{FLV-GKV}
\end{minipage}
\begin{minipage}{0.23\textwidth}
\centering
\textbf{NV-GKV}
\end{minipage}
\\
\hfill \\
\begin{minipage}{0.25\textwidth}
\centering
\includegraphics[width=0.4\linewidth, trim=1130 0 1130 0, clip]{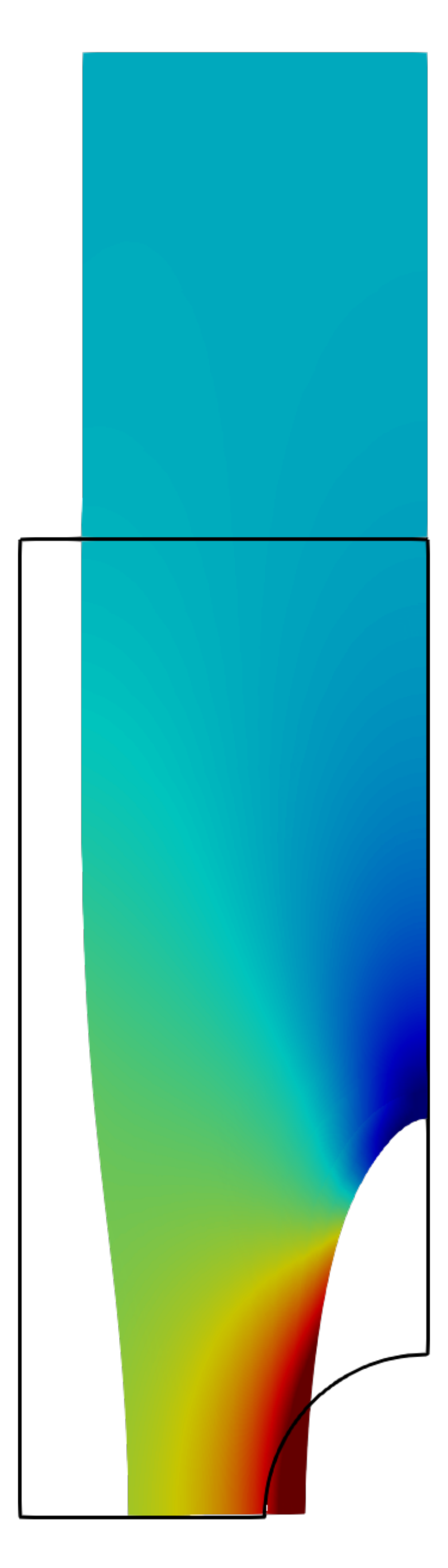}
\end{minipage}%
\begin{minipage}{0.25\textwidth}
\centering
\includegraphics[width=0.4\linewidth, trim=1130 0 1130 0, clip]{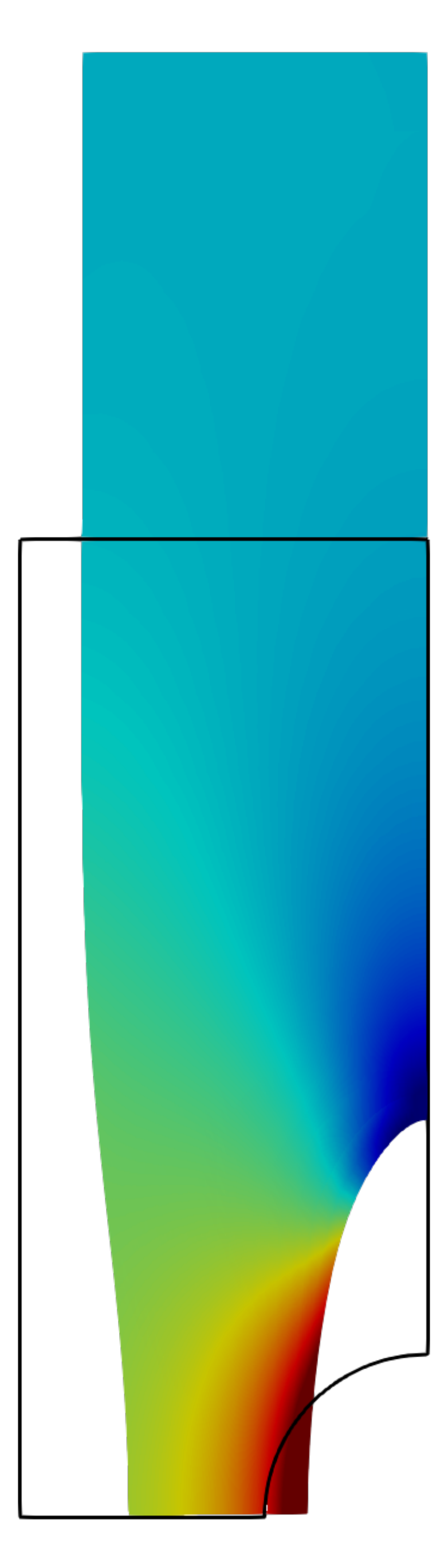}
\end{minipage}%
\begin{minipage}{0.25\textwidth}
\centering
\includegraphics[width=0.4\linewidth, trim=1130 0 1130 0, clip]{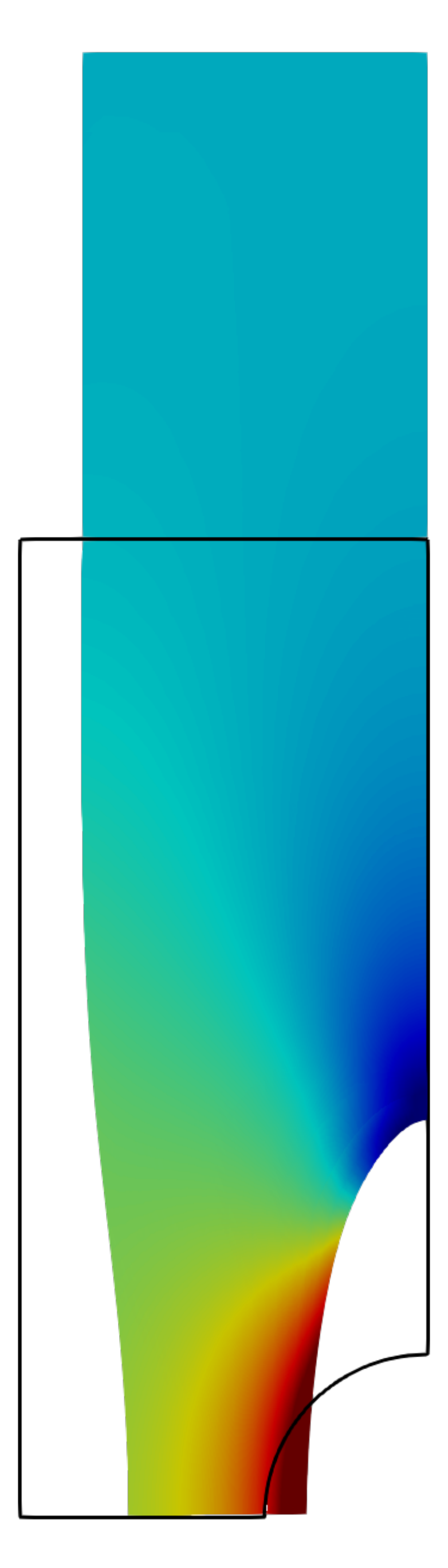}
\end{minipage}%
\begin{minipage}{0.25\textwidth}
\centering
\includegraphics[width=0.4\linewidth, trim=1130 0 1130 0, clip]{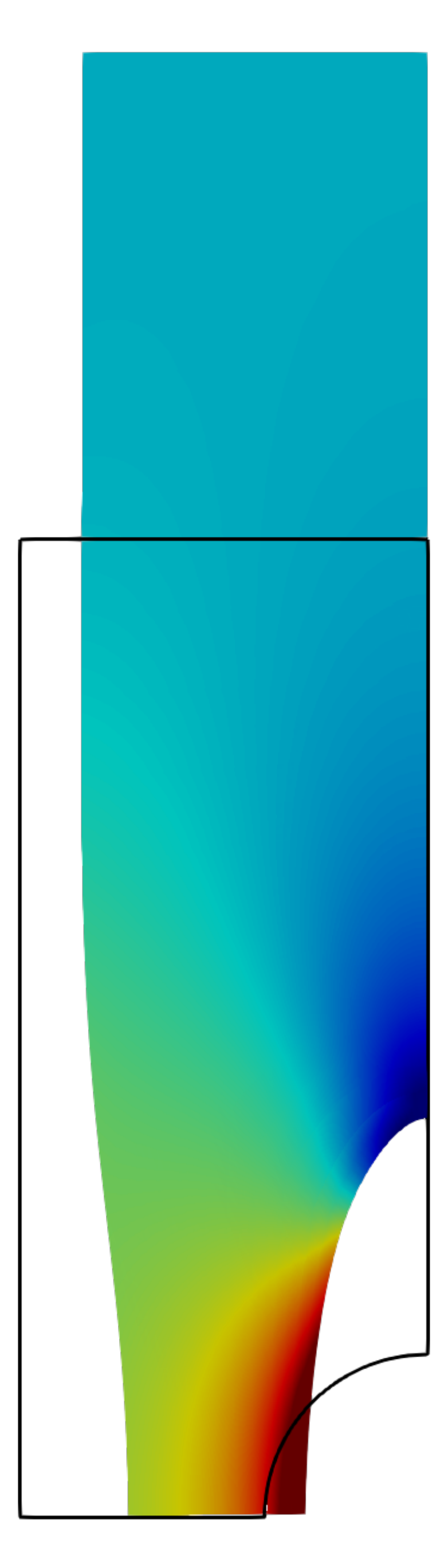}
\end{minipage}%
\\

\begin{minipage}{0.25\textwidth}
\centering
\includegraphics[width=0.4\linewidth, trim=1130 0 1130 0, clip]{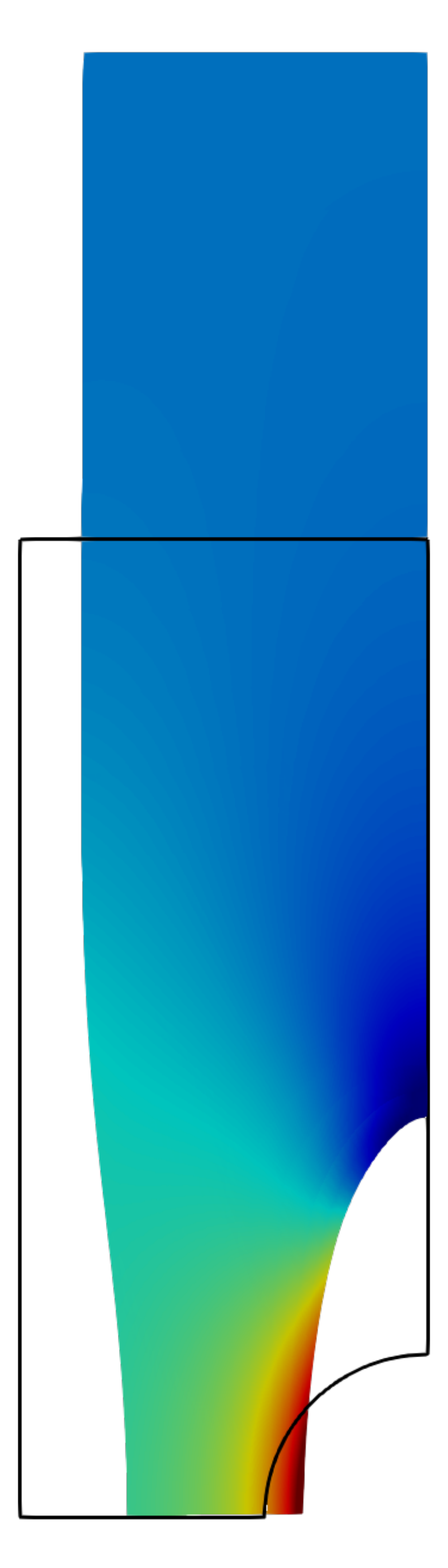}
\end{minipage}%
\begin{minipage}{0.25\textwidth}
\centering
\includegraphics[width=0.4\linewidth, trim=1130 0 1130 0, clip]{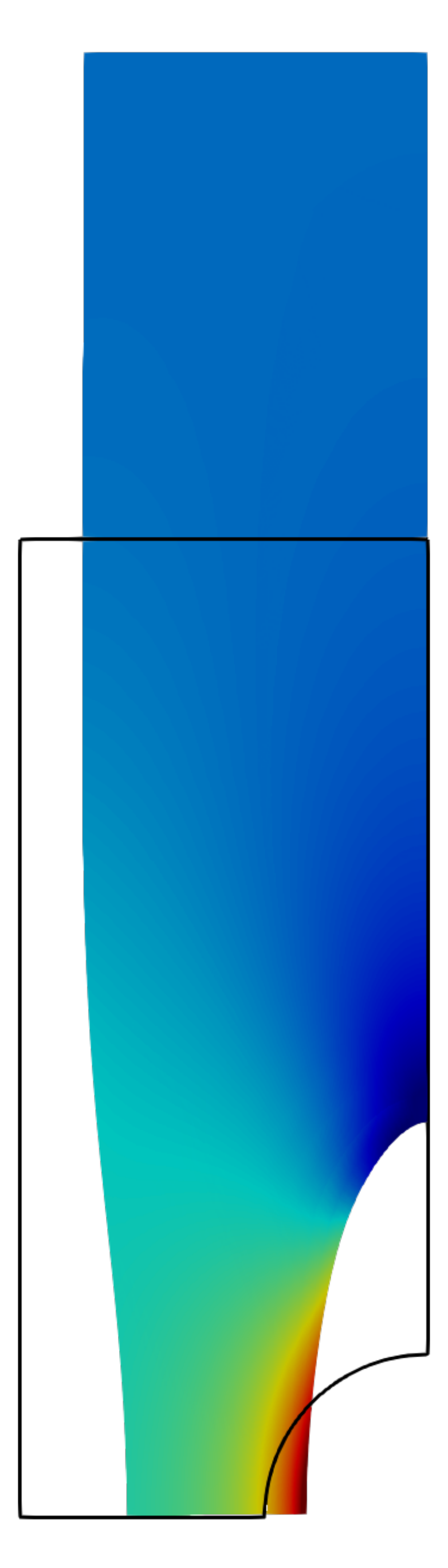}
\end{minipage}%
\begin{minipage}{0.25\textwidth}
\centering
\includegraphics[width=0.4\linewidth, trim=1130 0 1130 0, clip]{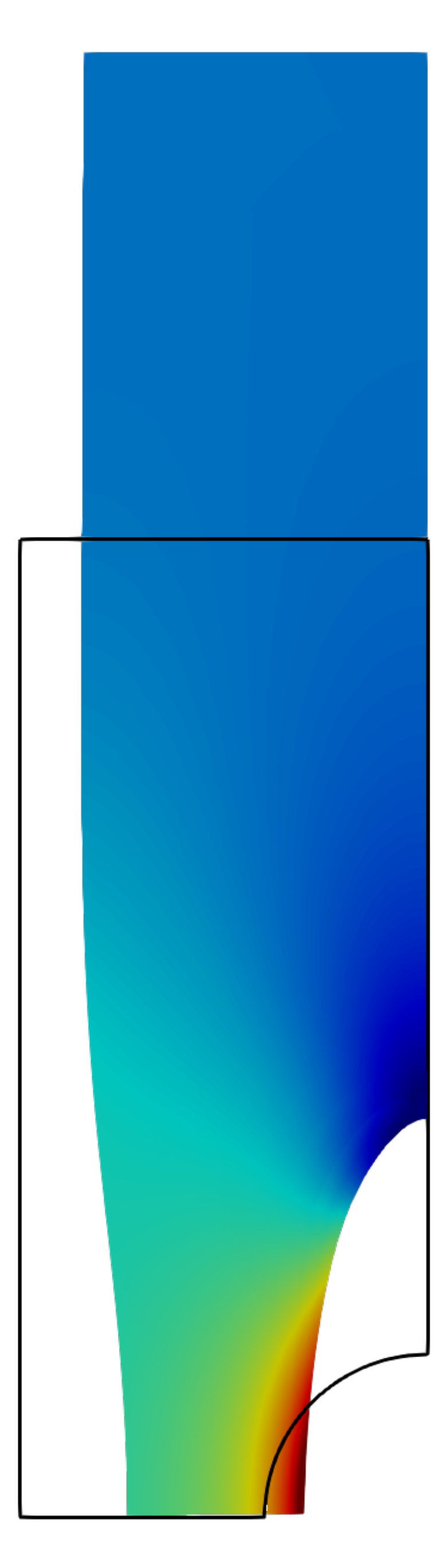}
\end{minipage}%
\begin{minipage}{0.25\textwidth}
\centering
\includegraphics[width=0.4\linewidth, trim=1130 0 1130 0, clip]{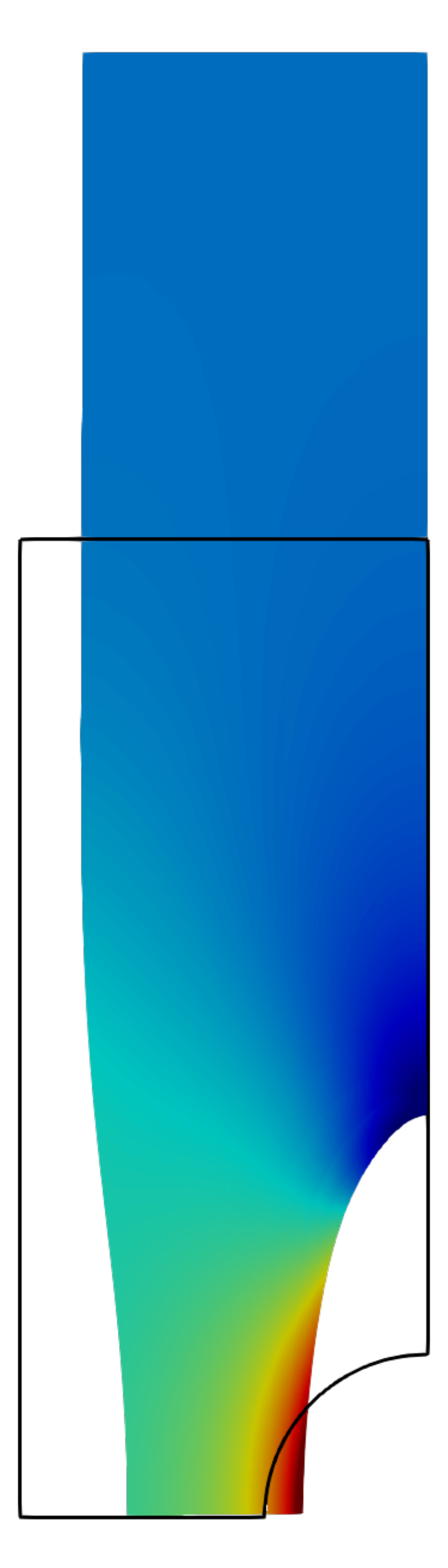}
\end{minipage}%
\\

\begin{minipage}{0.25\textwidth}
\centering
\includegraphics[width=0.4\linewidth, trim=1130 0 1130 200, clip]{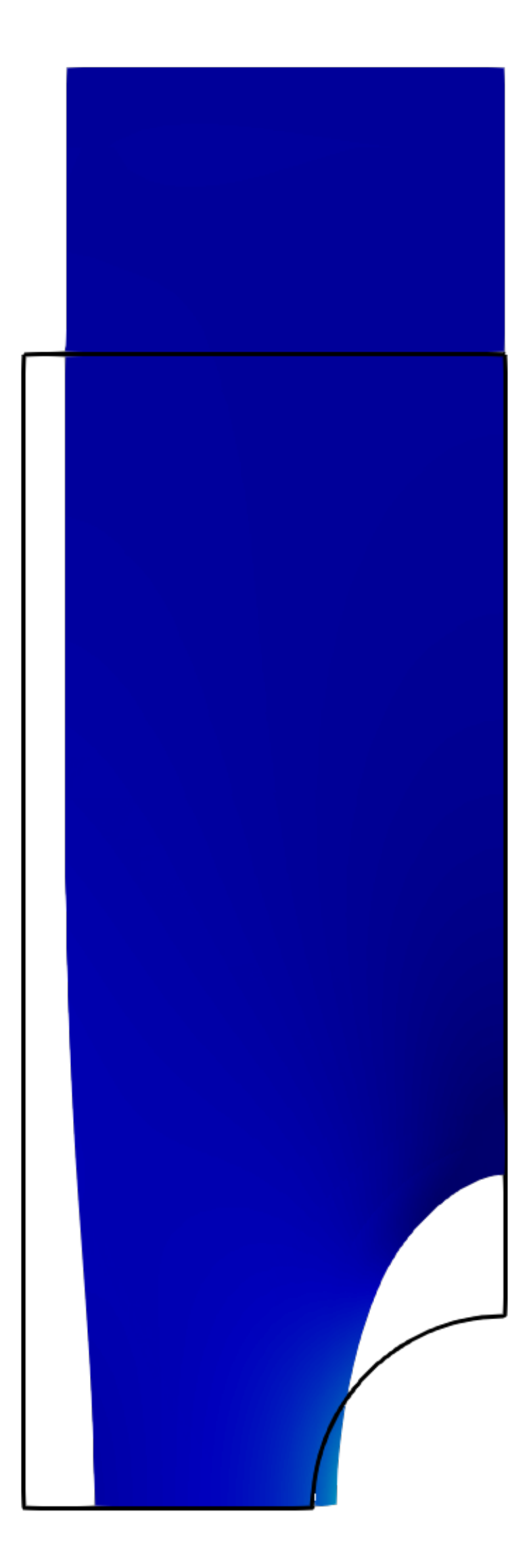}
\end{minipage}%
\begin{minipage}{0.25\textwidth}
\centering
\includegraphics[width=0.4\linewidth, trim=1130 0 1130 200, clip]{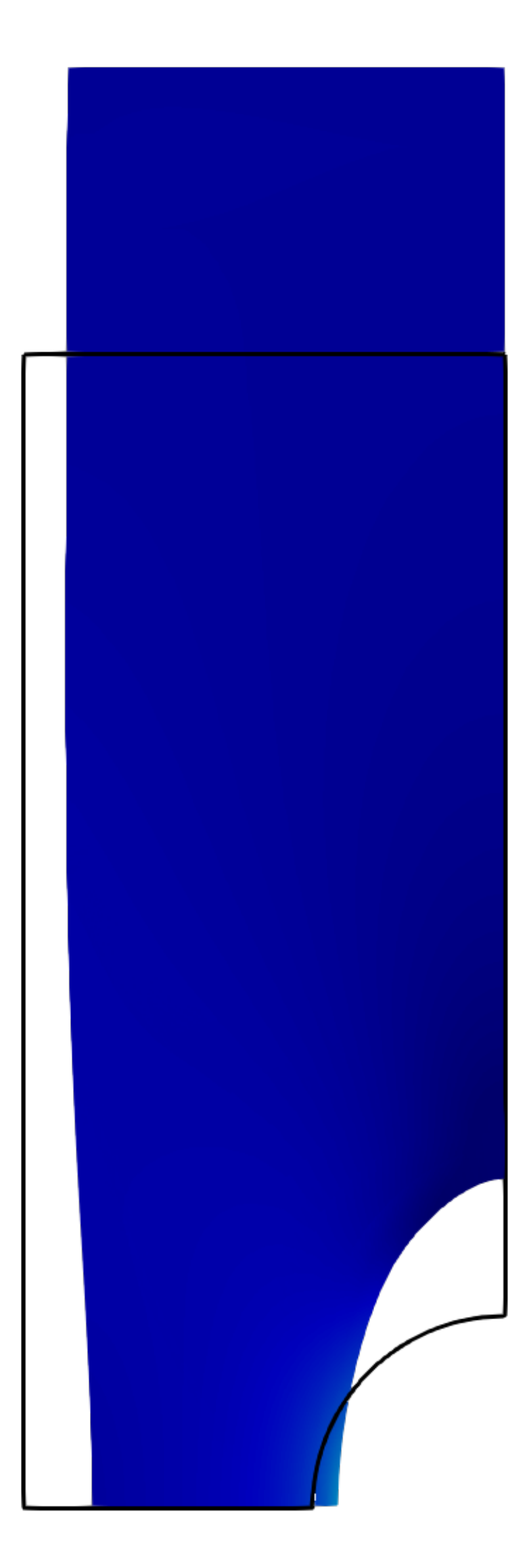}
\end{minipage}%
\begin{minipage}{0.25\textwidth}
\centering
\includegraphics[width=0.4\linewidth, trim=1130 0 1130 200, clip]{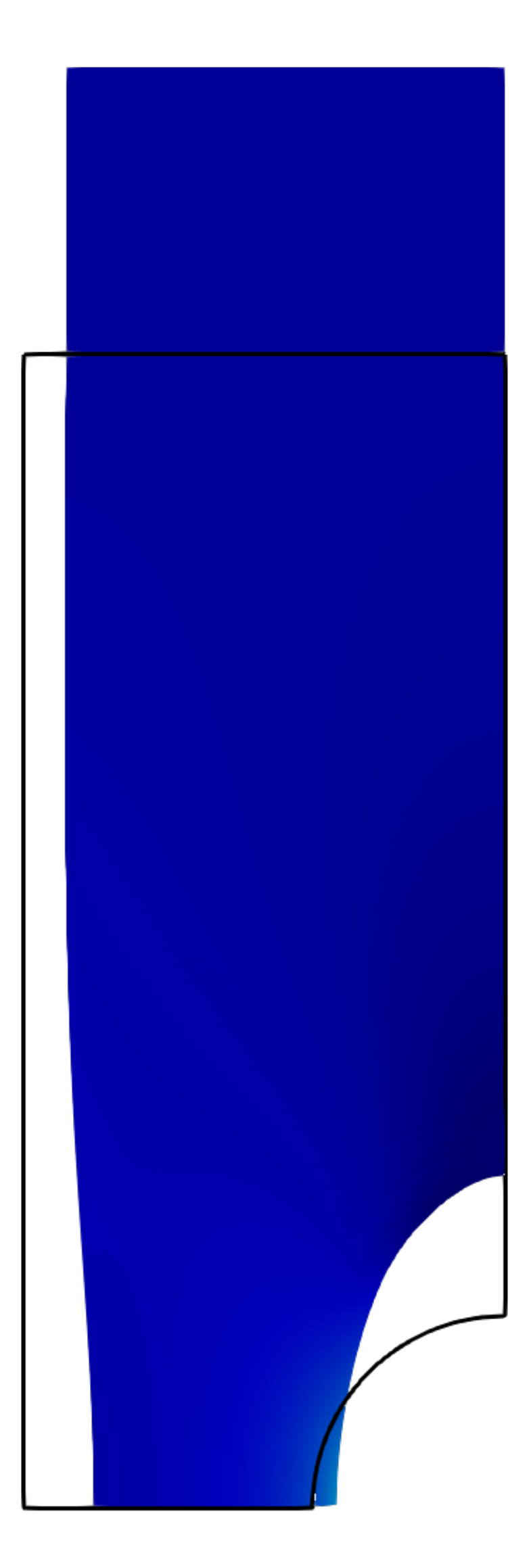}
\end{minipage}%
\begin{minipage}{0.25\textwidth}
\centering
\includegraphics[width=0.4\linewidth, trim=1130 0 1130 200, clip]{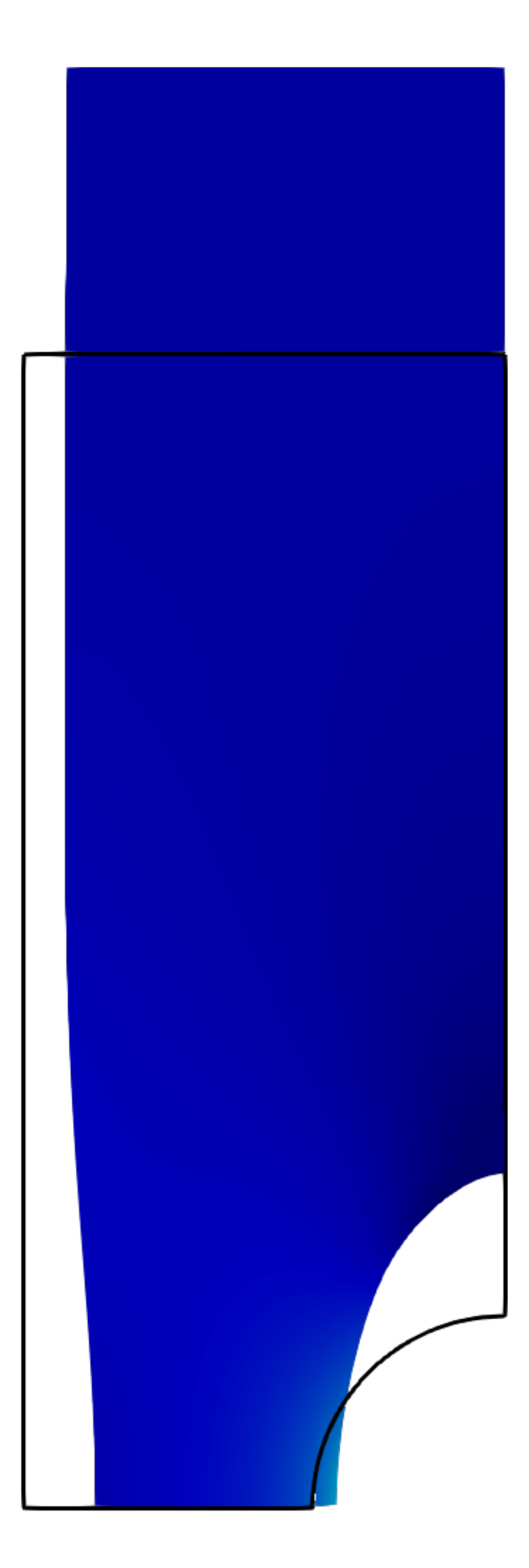}
\end{minipage}%
\\
\begin{minipage}{1.0\textwidth}
\centering
\includegraphics[width=0.5\linewidth, trim=0 1210 1920 0, clip]{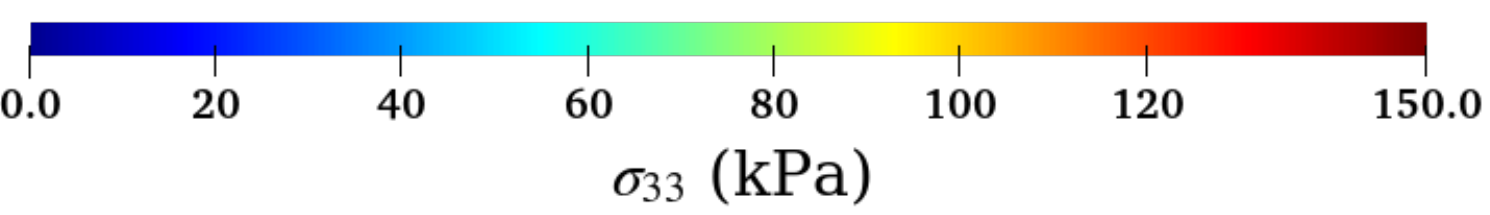}
\end{minipage}%
\caption{Cauchy stress distributions at $t = 10$, $20$, and $25$~s (top to bottom rows).}
\label{fig:stress_distribution}
\end{figure}

\begin{figure}[h]
\centering
\begin{minipage}{0.25\textwidth}
\centering
\textbf{FLV-GM}
\end{minipage}%
\begin{minipage}{0.24\textwidth}
\centering
\textbf{NV-GM}
\end{minipage}%
\begin{minipage}{0.27\textwidth}
\centering
\textbf{FLV-GKV}
\end{minipage}
\begin{minipage}{0.23\textwidth}
\centering
\textbf{NV-GKV}
\end{minipage}
\\
\hfill \\
\begin{minipage}{0.25\textwidth}
\centering
\includegraphics[width=0.4\linewidth, trim=1130 0 1130 0, clip]{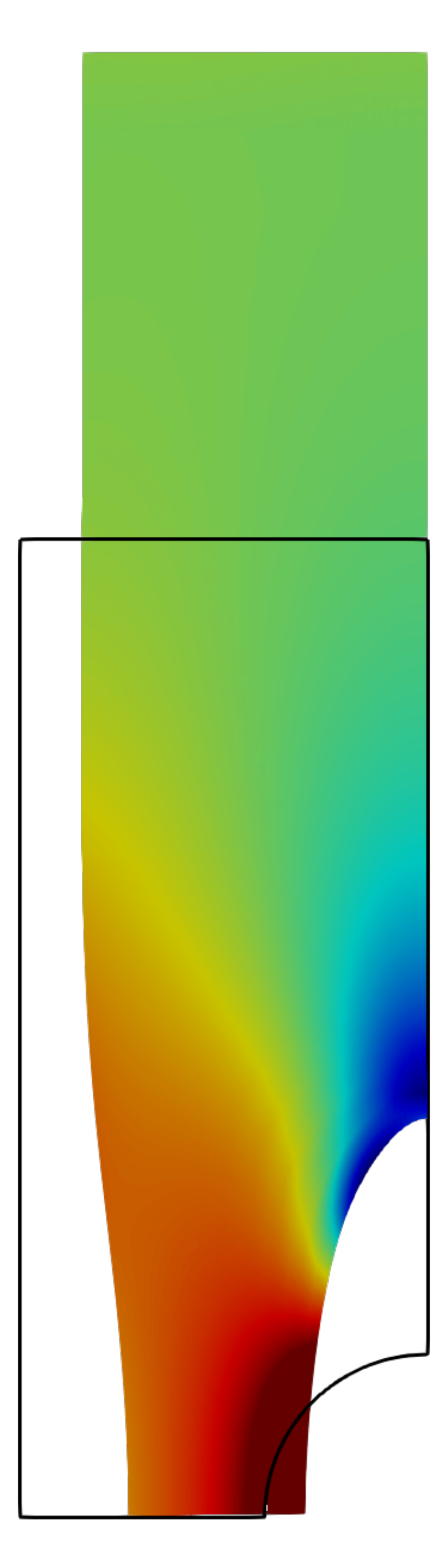}
\end{minipage}%
\begin{minipage}{0.25\textwidth}
\centering
\includegraphics[width=0.4\linewidth, trim=1130 0 1130 0, clip]{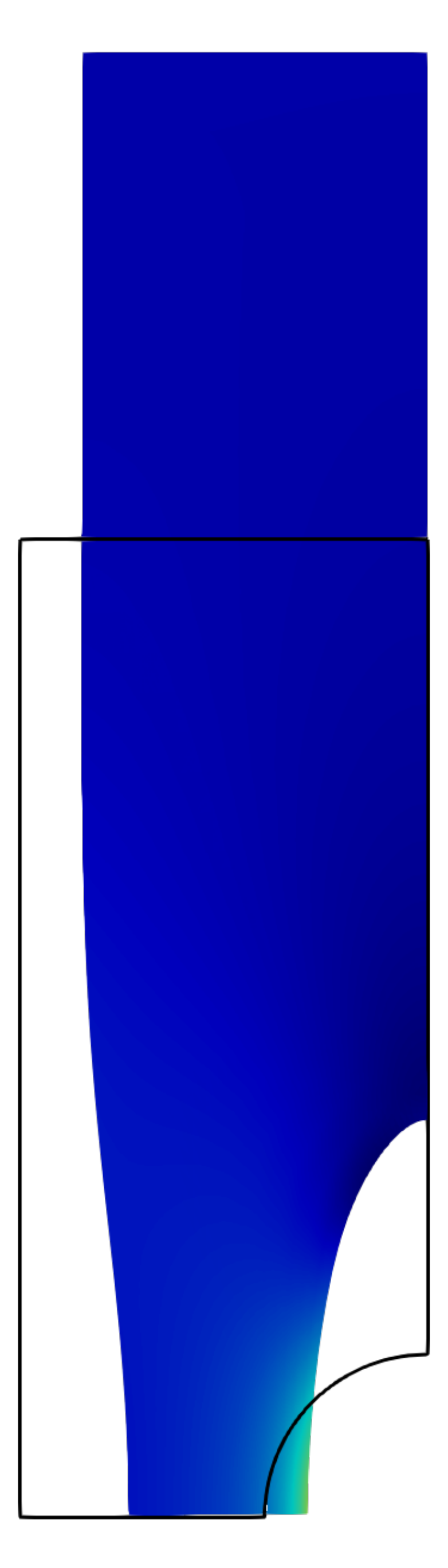}
\end{minipage}%
\begin{minipage}{0.25\textwidth}
\centering
\includegraphics[width=0.4\linewidth, trim=1130 0 1130 0, clip]{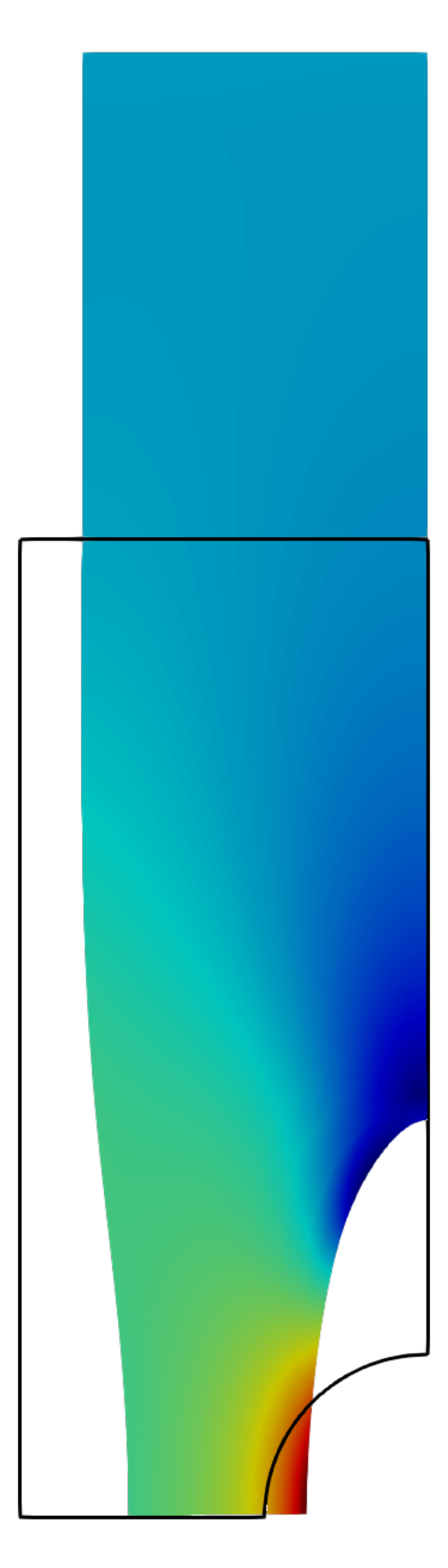}
\end{minipage}%
\begin{minipage}{0.25\textwidth}
\centering
\includegraphics[width=0.4\linewidth, trim=1130 0 1130 0, clip]{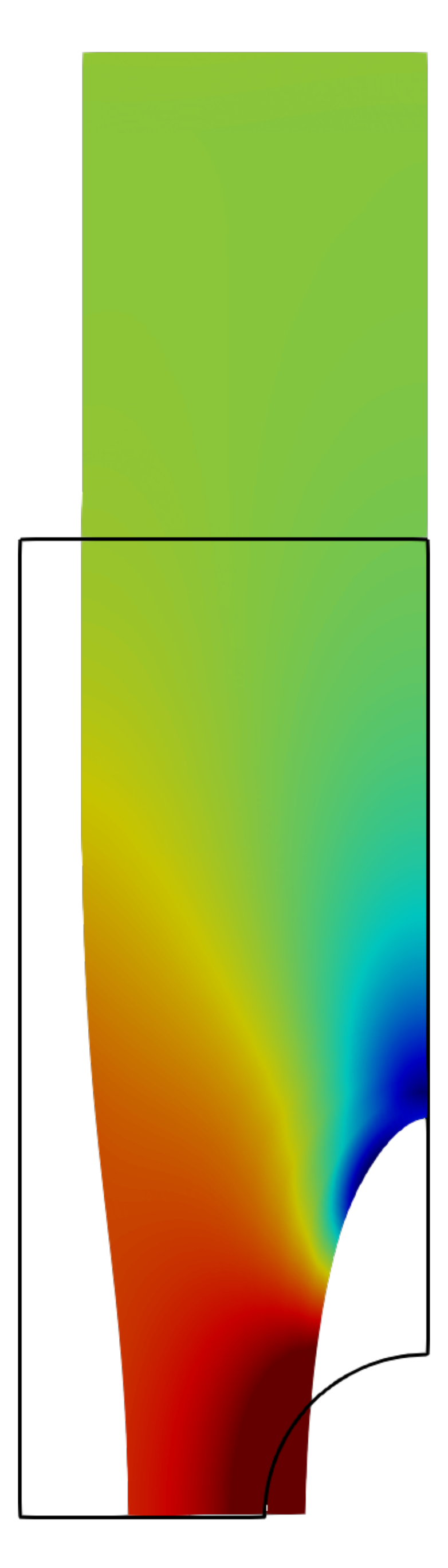}
\end{minipage}%
\\
	
\begin{minipage}{0.25\textwidth}
\centering
\includegraphics[width=0.4\linewidth, trim=1130 0 1130 0, clip]{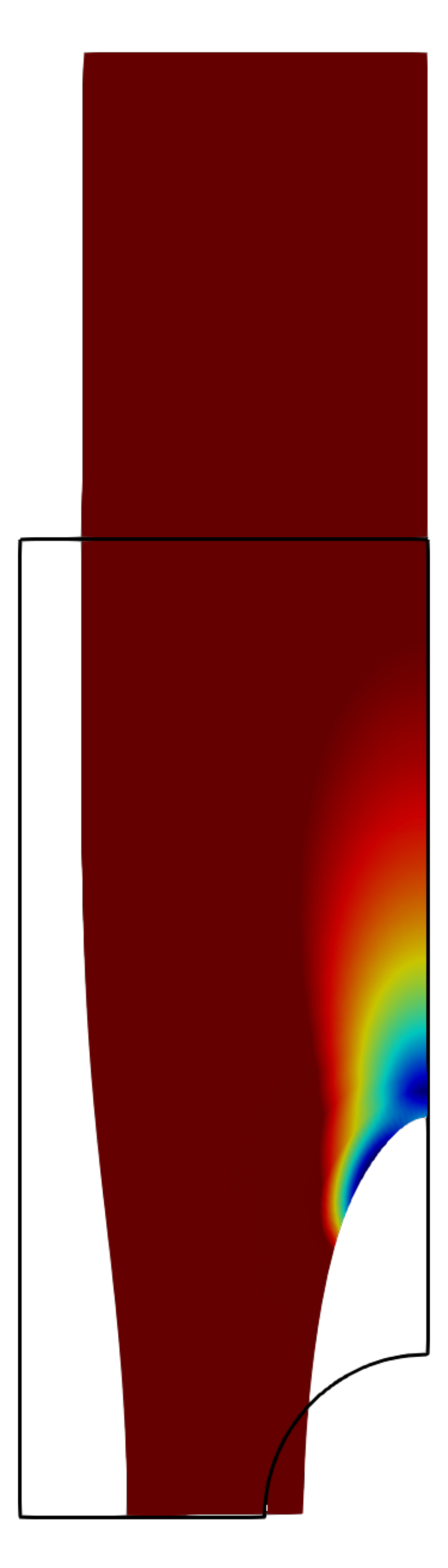}
\end{minipage}%
\begin{minipage}{0.25\textwidth}
\centering
\includegraphics[width=0.4\linewidth, trim=1130 0 1130 0, clip]{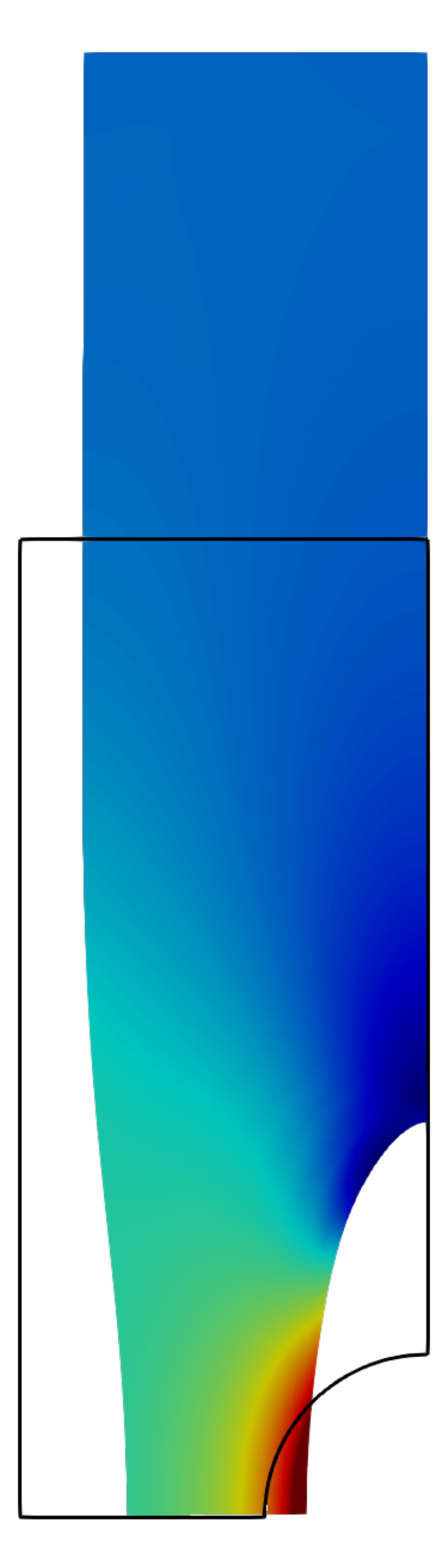}
\end{minipage}%
\begin{minipage}{0.25\textwidth}
\centering
\includegraphics[width=0.4\linewidth, trim=1130 0 1130 0, clip]{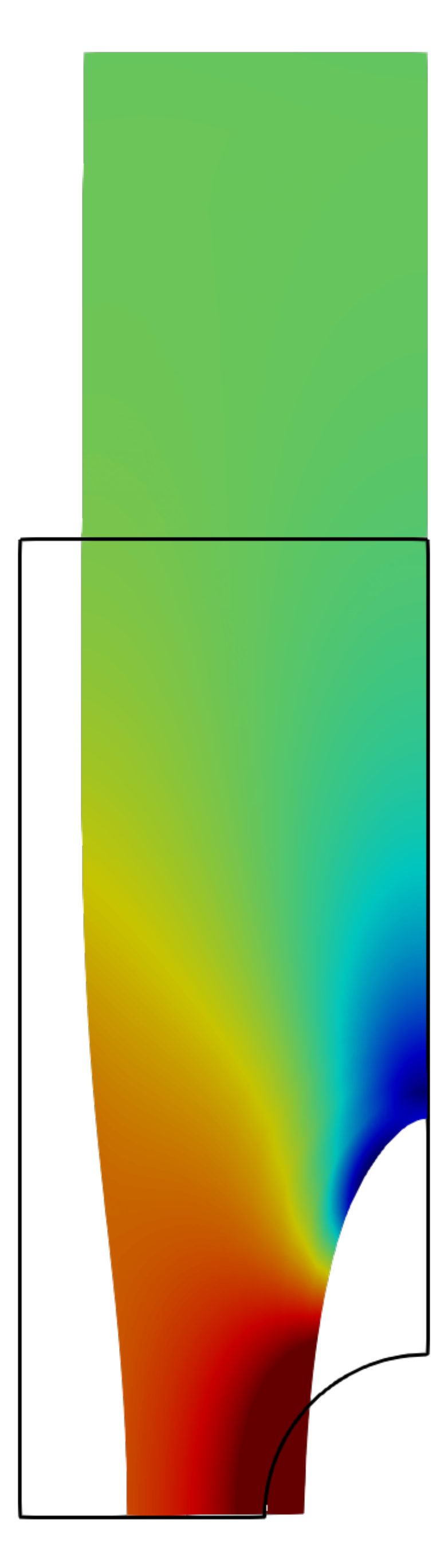}
\end{minipage}%
\begin{minipage}{0.25\textwidth}
\centering
\includegraphics[width=0.4\linewidth, trim=1130 0 1130 0, clip]{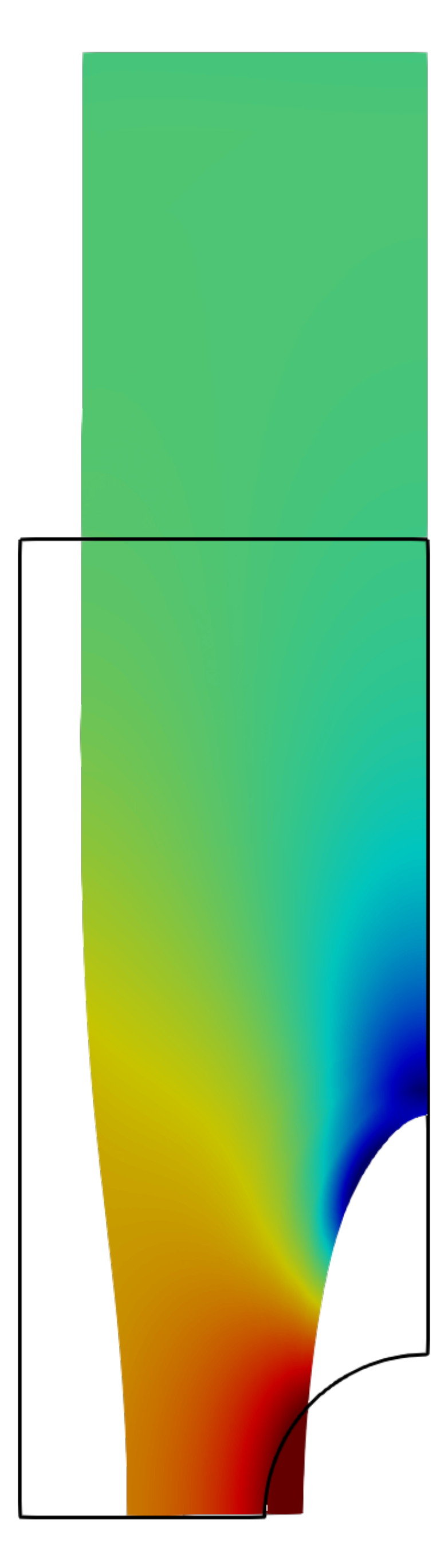}
\end{minipage}%
\\
	
\begin{minipage}{0.25\textwidth}
\centering
\includegraphics[width=0.4\linewidth, trim=1130 0 1130 200, clip]{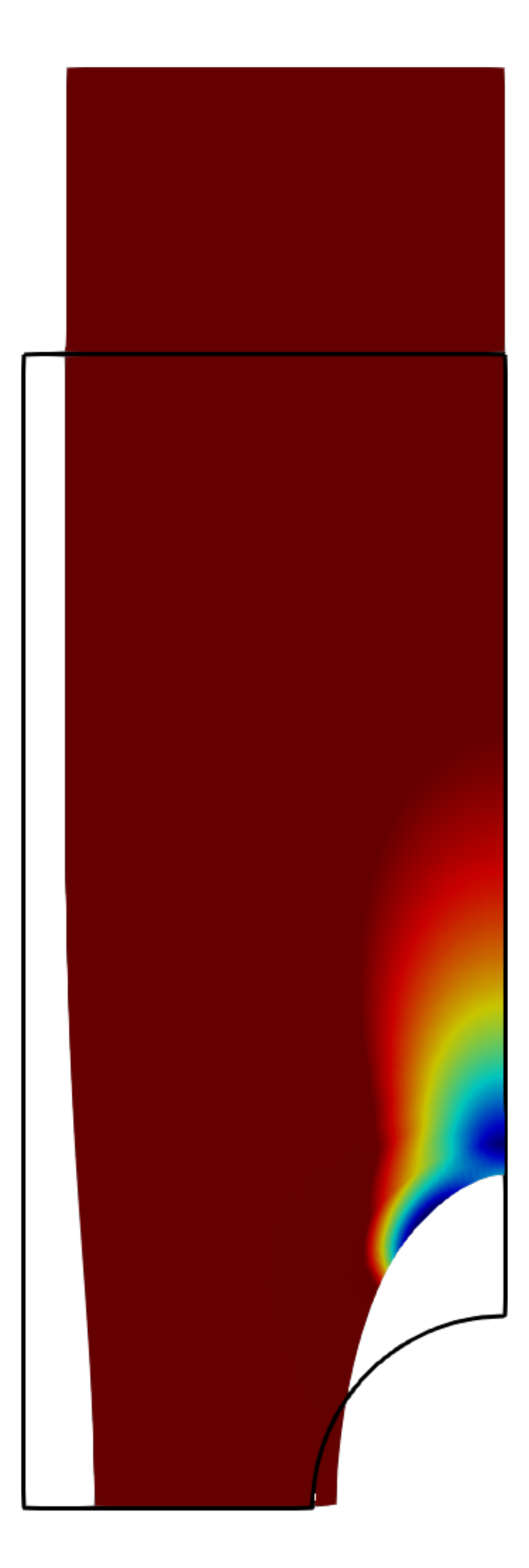}
\end{minipage}%
\begin{minipage}{0.25\textwidth}
\centering
\includegraphics[width=0.4\linewidth, trim=1130 0 1130 200, clip]{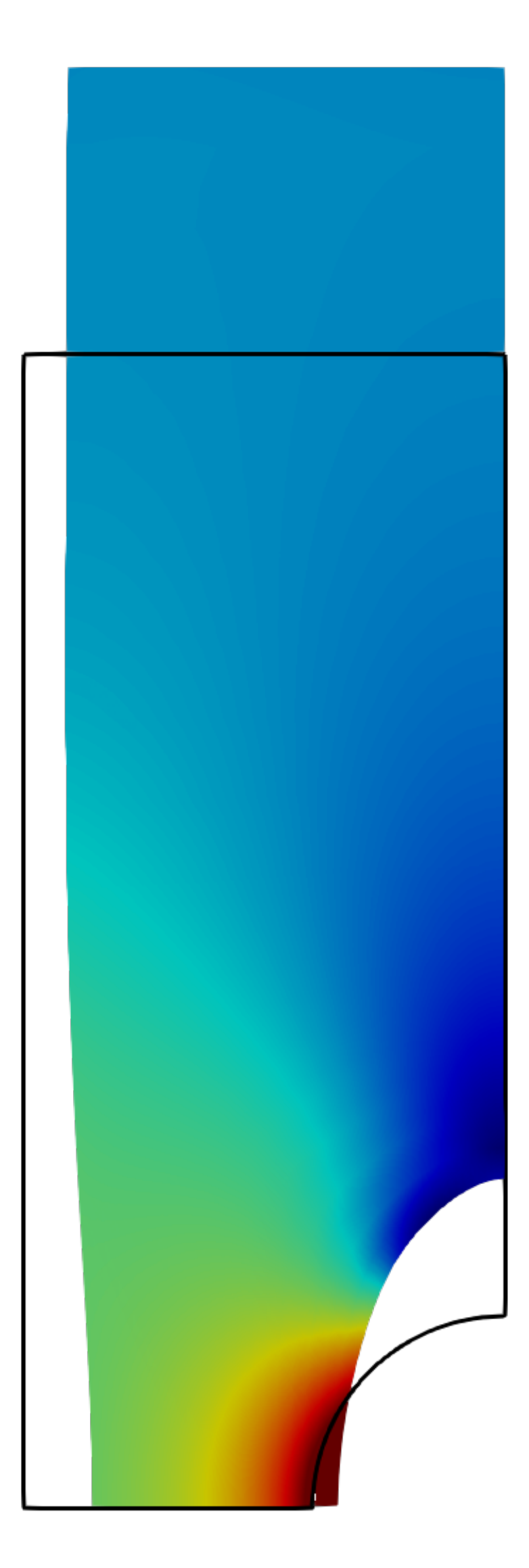}
\end{minipage}%
\begin{minipage}{0.25\textwidth}
\centering
\includegraphics[width=0.4\linewidth, trim=1130 0 1130 200, clip]{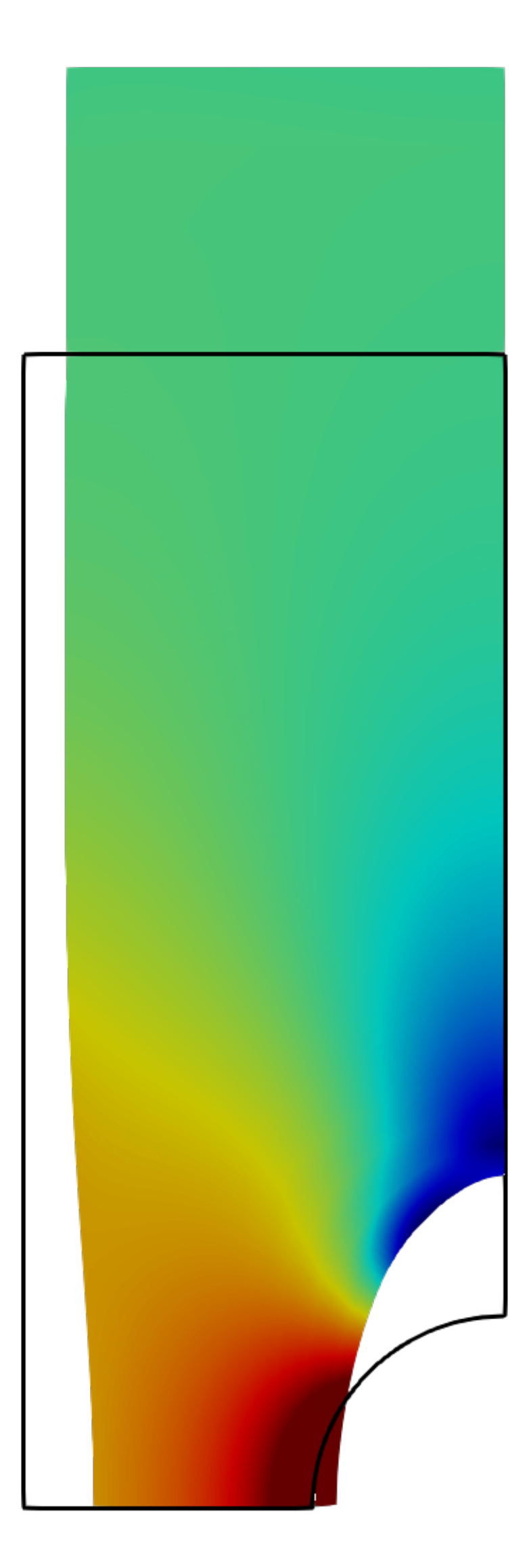}
\end{minipage}%
\begin{minipage}{0.25\textwidth}
\centering
\includegraphics[width=0.4\linewidth, trim=1130 0 1130 200, clip]{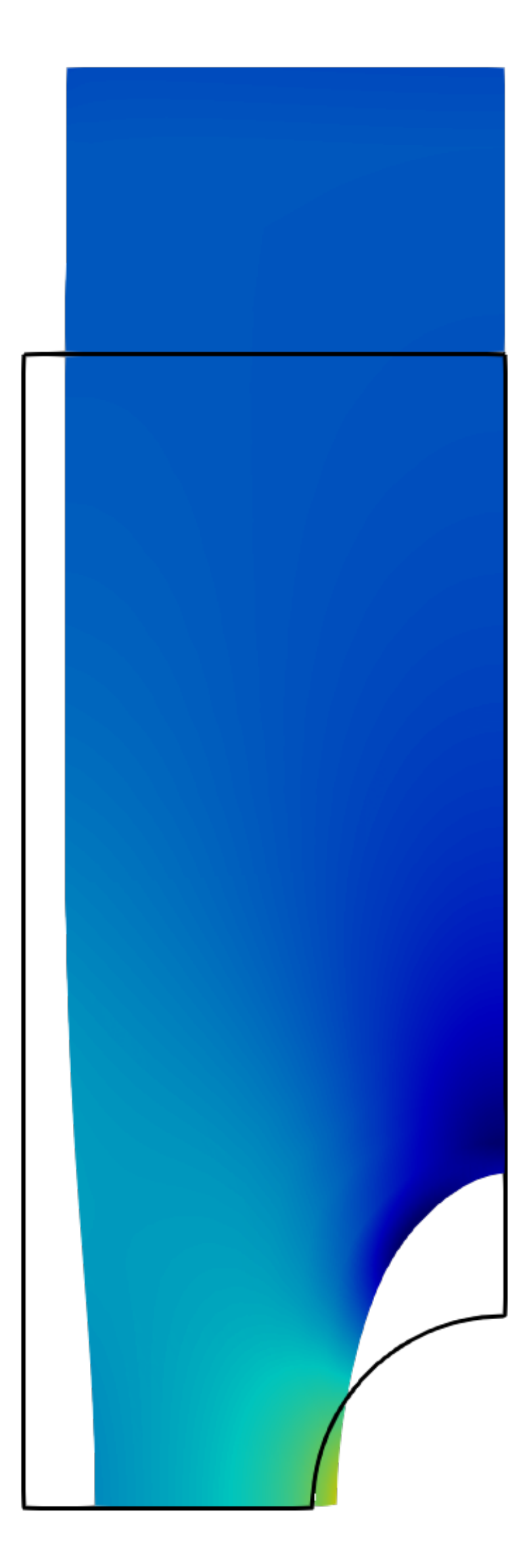}
\end{minipage}%
\\
\begin{minipage}{1.0\textwidth}
\centering
\includegraphics[width=0.35\linewidth, trim=0 1210 2160 0, clip]{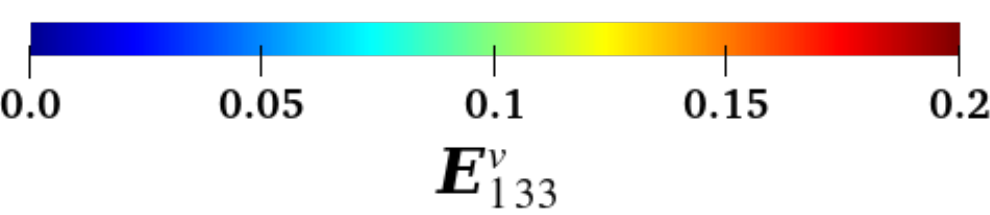}
\end{minipage}%
\caption{The distributions of $\bm E^{\mathrm v}_{1\:33}$ at $t = 10$, $20$, and $25$~s (top to bottom rows).}
\label{fig:Ev1_distribution}
\end{figure}

\begin{figure}[h]
\centering
\begin{minipage}{0.25\textwidth}
\centering
\textbf{FLV-GM}
\end{minipage}%
\begin{minipage}{0.24\textwidth}
\centering
\textbf{NV-GM}
\end{minipage}%
\begin{minipage}{0.27\textwidth}
\centering
\textbf{FLV-GKV}
\end{minipage}
\begin{minipage}{0.23\textwidth}
\centering
\textbf{NV-GKV}
\end{minipage}
\\
\hfill \\
\begin{minipage}{0.25\textwidth}
\centering
\includegraphics[width=0.4\linewidth, trim=1130 0 1130 0, clip]{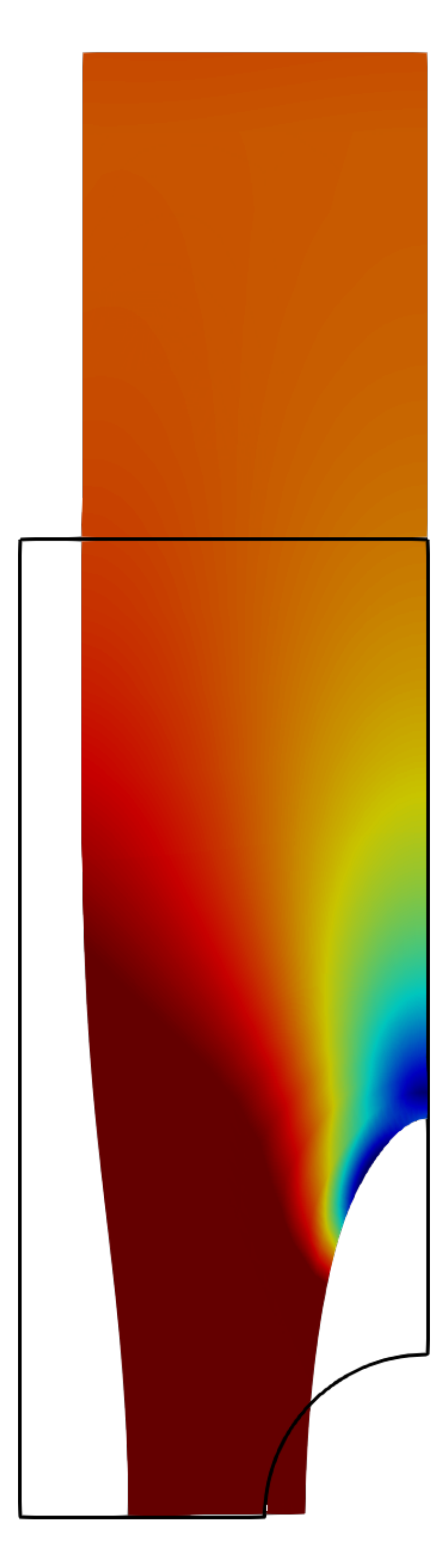}
\end{minipage}%
\begin{minipage}{0.25\textwidth}
\centering
\includegraphics[width=0.4\linewidth, trim=1130 0 1130 0, clip]{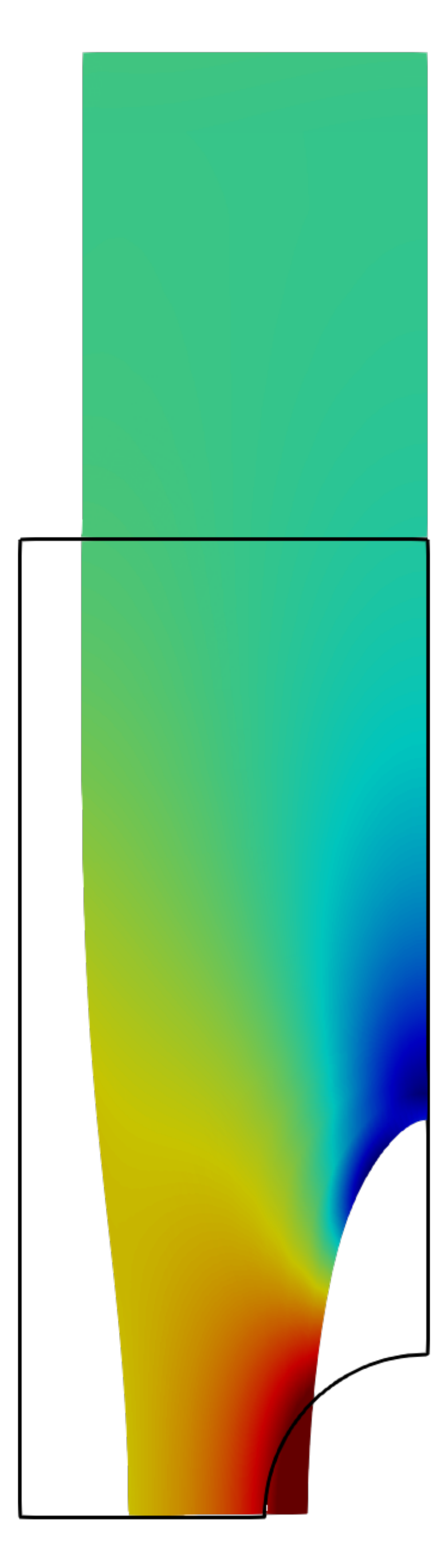}
\end{minipage}%
\begin{minipage}{0.25\textwidth}
\centering
\includegraphics[width=0.4\linewidth, trim=1130 0 1130 0, clip]{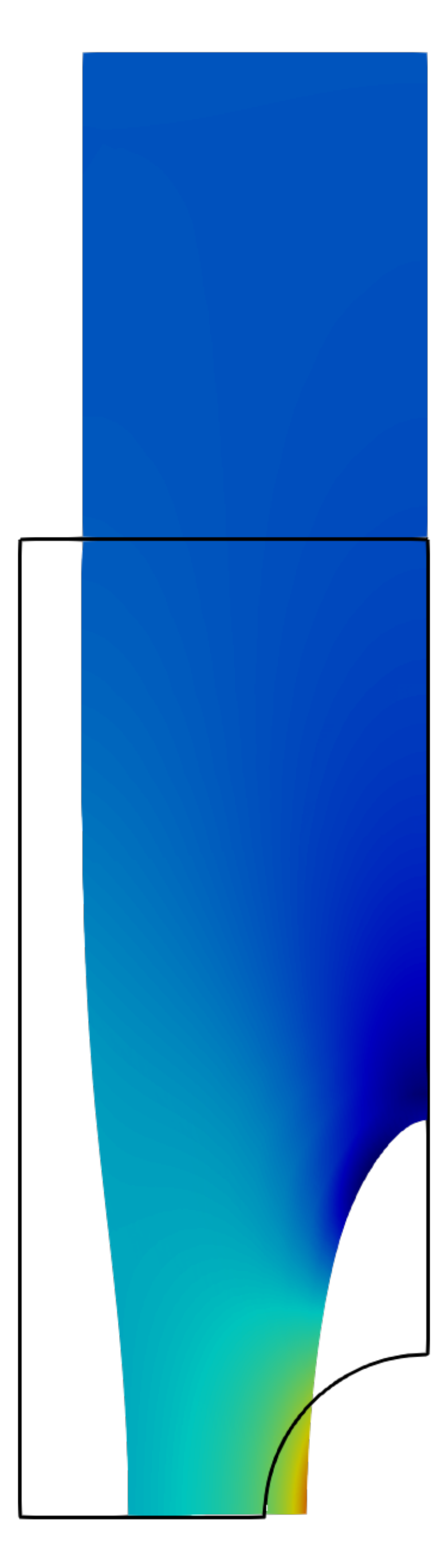}
\end{minipage}%
\begin{minipage}{0.25\textwidth}
\centering
\includegraphics[width=0.4\linewidth, trim=1130 0 1130 0, clip]{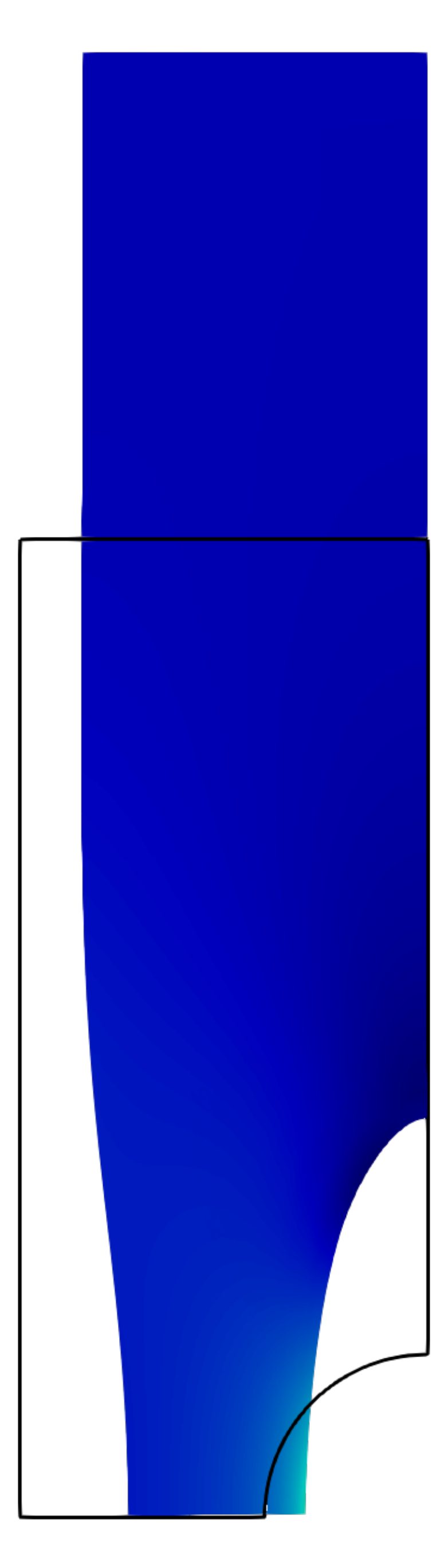}
\end{minipage}%
\\
	
\begin{minipage}{0.25\textwidth}
\centering
\includegraphics[width=0.4\linewidth, trim=1130 0 1130 0, clip]{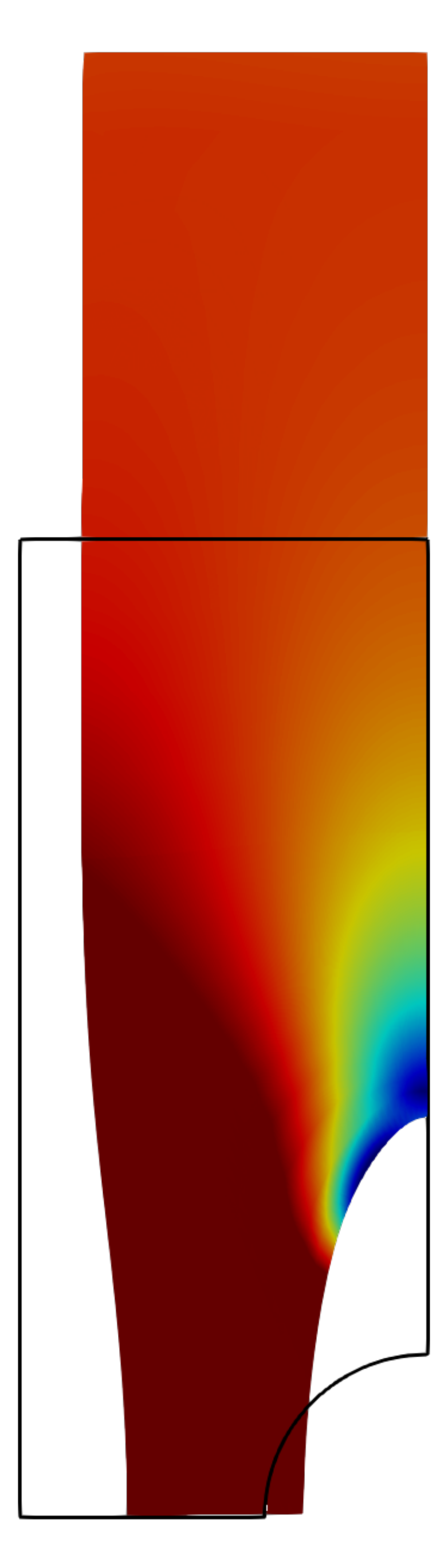}
\end{minipage}%
\begin{minipage}{0.25\textwidth}
\centering
\includegraphics[width=0.4\linewidth, trim=1130 0 1130 0, clip]{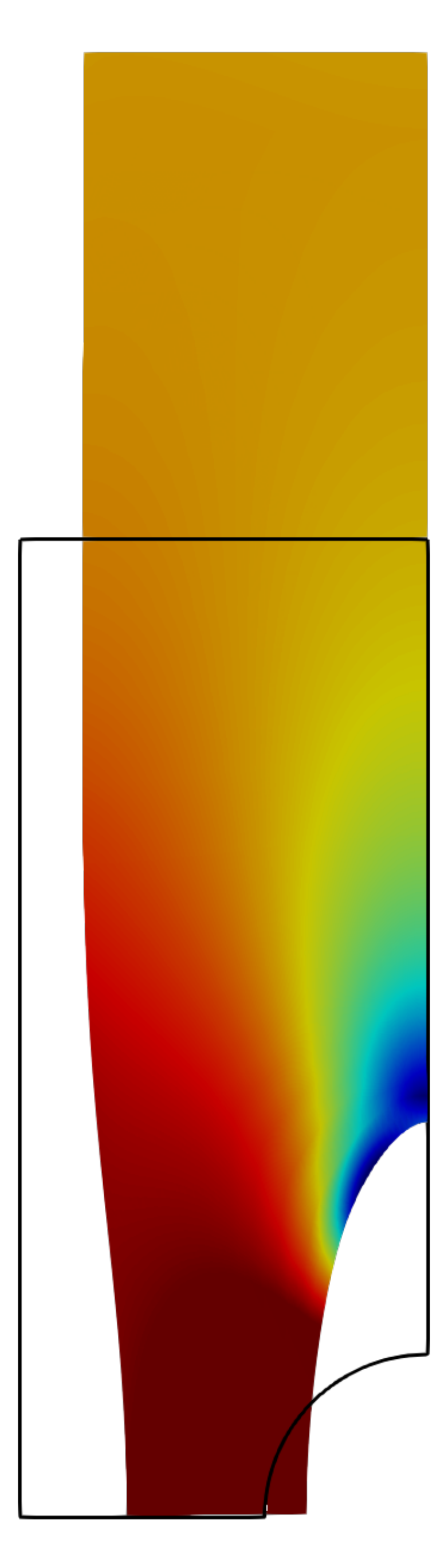}
\end{minipage}%
\begin{minipage}{0.25\textwidth}
\centering
\includegraphics[width=0.4\linewidth, trim=1130 0 1130 0, clip]{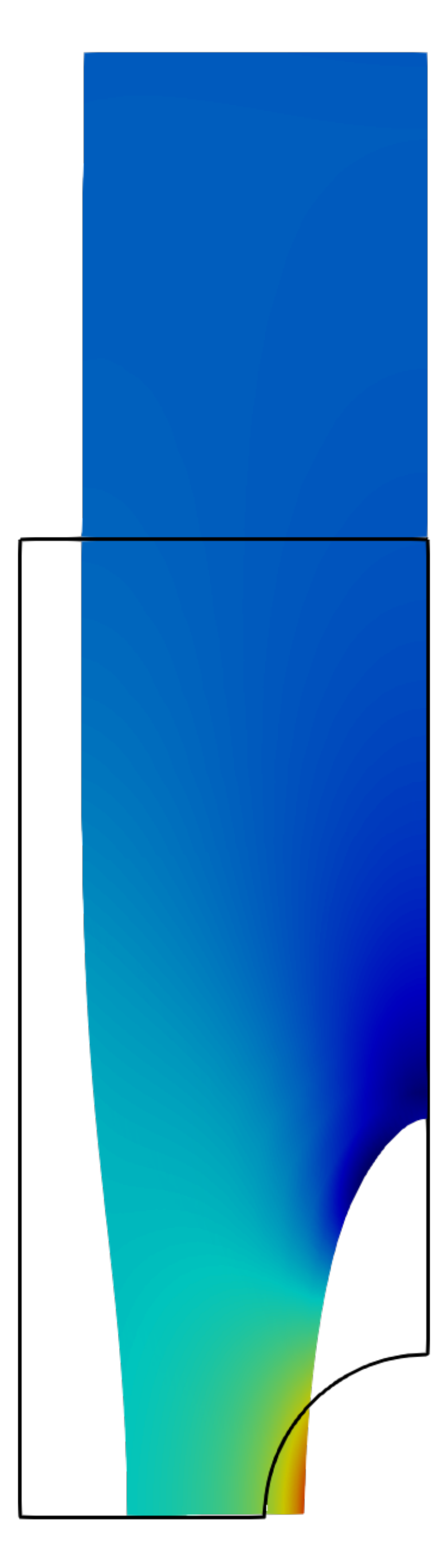}
\end{minipage}%
\begin{minipage}{0.25\textwidth}
\centering
\includegraphics[width=0.4\linewidth, trim=1130 0 1130 0, clip]{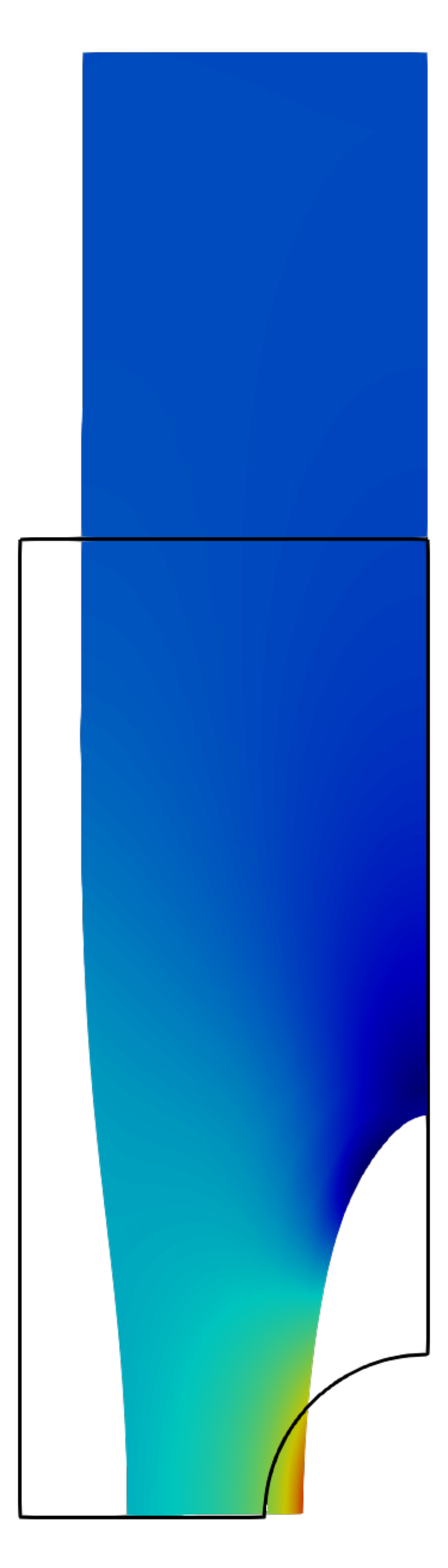}
\end{minipage}%
\\
	
\begin{minipage}{0.25\textwidth}
\centering
\includegraphics[width=0.4\linewidth, trim=1130 0 1130 200, clip]{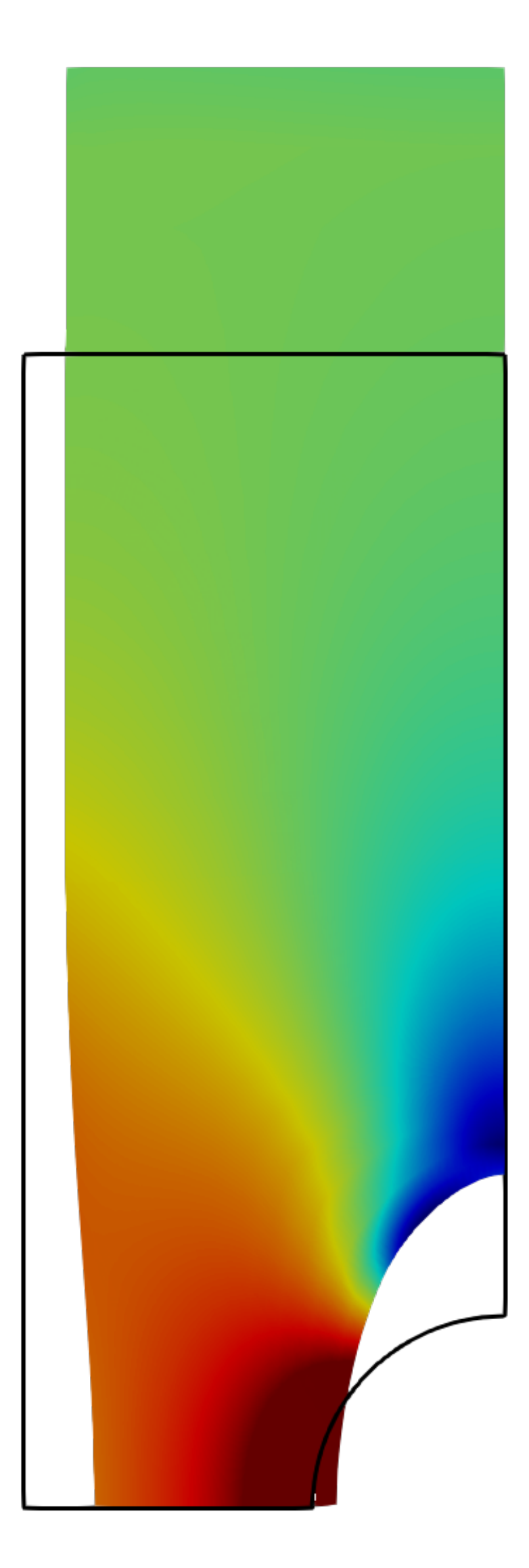}
\end{minipage}%
\begin{minipage}{0.25\textwidth}
\centering
\includegraphics[width=0.4\linewidth, trim=1130 0 1130 200, clip]{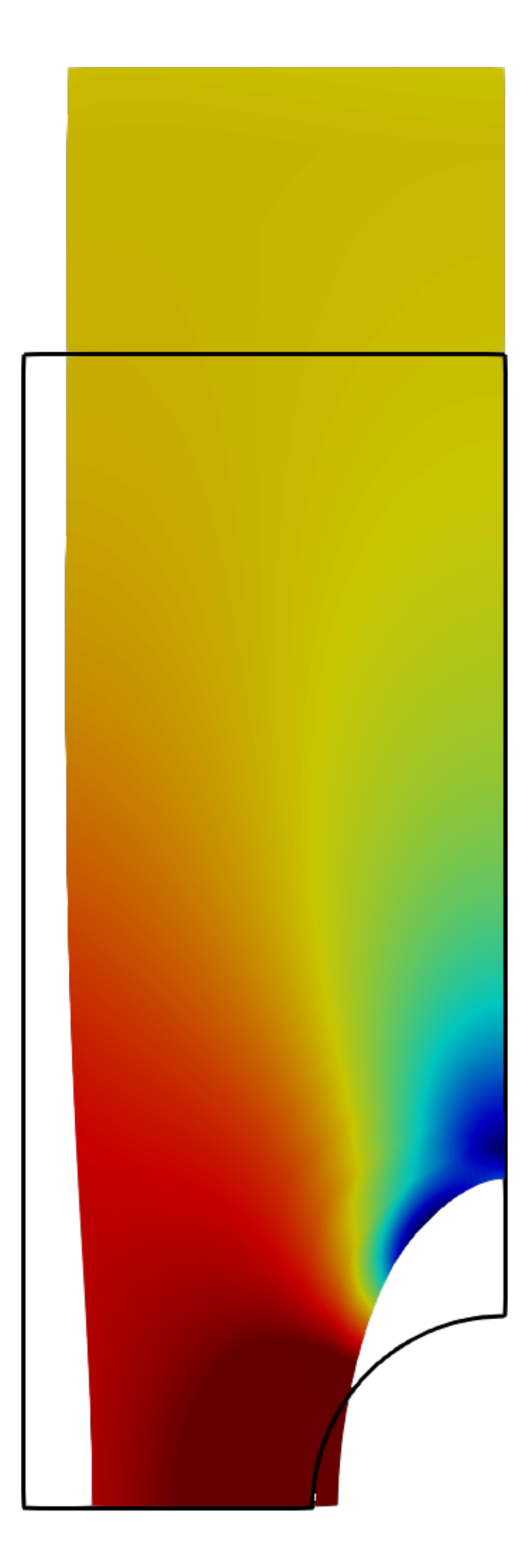}
\end{minipage}%
\begin{minipage}{0.25\textwidth}
\centering
\includegraphics[width=0.4\linewidth, trim=1130 0 1130 200, clip]{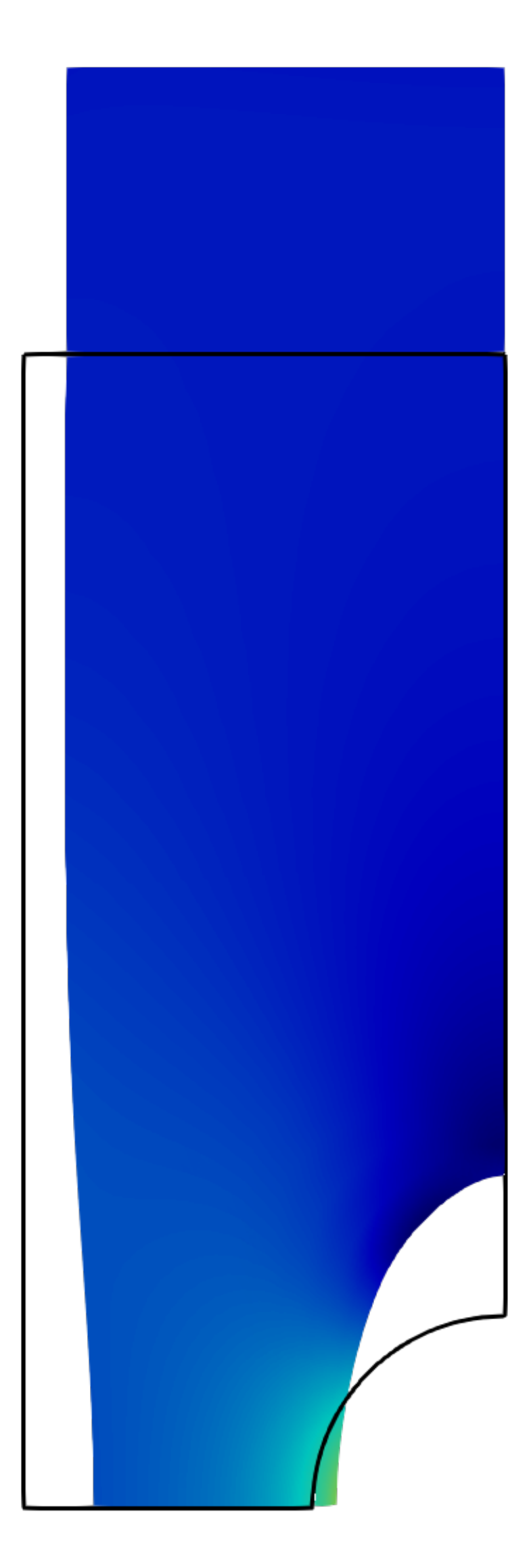}
\end{minipage}%
\begin{minipage}{0.25\textwidth}
\centering
\includegraphics[width=0.4\linewidth, trim=1130 0 1130 200, clip]{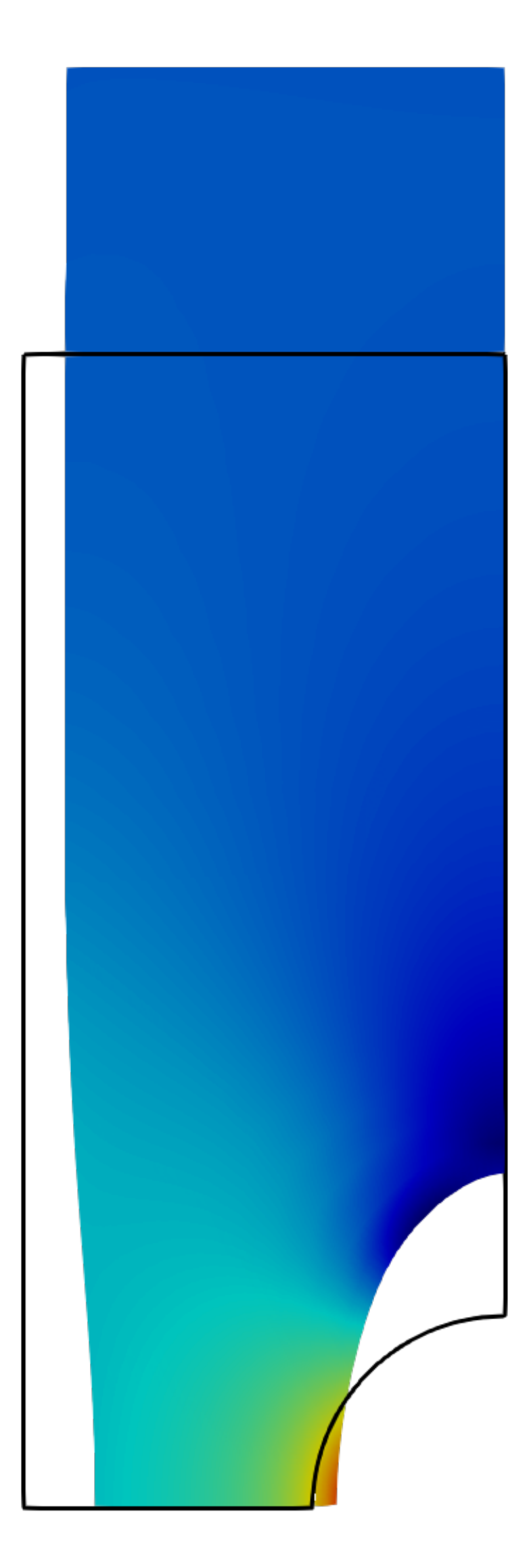}
\end{minipage}%
\\
\begin{minipage}{1.0\textwidth}
\centering
\includegraphics[width=0.3\linewidth, trim=0 1200 2250 0, clip]{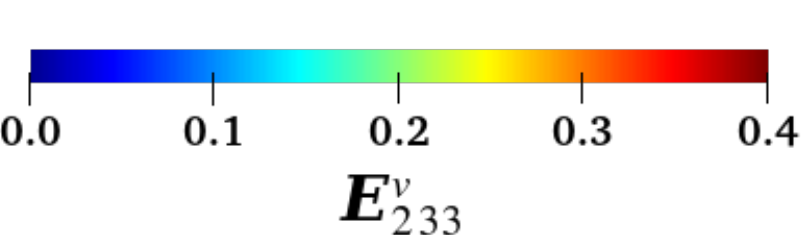}
\end{minipage}%
\caption{The distributions of $\bm E^{\mathrm v}_{2\:33}$ at $t = 10$, $20$, and $25$~s (top to bottom rows).}
\label{fig:Ev2_distribution}
\end{figure}

We monitor the total traction along the stretching direction on the top surface (Figure \ref{fig:time_traction}). All four models exhibit a similar traction response over time. We also examine the distribution of the Cauchy stress in the loading direction by depicting $\bm \sigma_{33}$ across the entire plate (Figure \ref{fig:stress_distribution}). The stress fields are evaluated at $t = 10$, $20$, and $25$~s, corresponding to the times of the maximum stretch, the end of the relaxation phase, and the completion of the compression phase, respectively. The predicted stress distributions from the four models exhibit good overall consistency, with only minor discrepancies in the vicinity of the hole. These results indicate that the models calibrated to the same experimental data yield consistent mechanical responses.

Moreover, we examine the distributions of two internal variable components in the loading direction, specifically $\bm E^\mathrm{v}_{1\:33}$ and $\bm E^\mathrm{v}_{2\:33}$, across the entire plate, as shown in Figure~\ref{fig:Ev1_distribution} and Figure~\ref{fig:Ev2_distribution}, respectively. Although the finite linear and nonlinear viscoelastic models under both rheological frameworks produce nearly identical stress responses, the evolution behaviors of their internal variables differ significantly. For the generalized Maxwell models, the evolution of the internal variables exhibits notably different intensities between the finite linear and nonlinear models. This discrepancy can be partly attributed to the difference in their relaxation or retardation times: the FLV-GM model uses $\eta_1/\mu_1 = 13.42$~s and $\eta_2/\mu_2 = 0.78$~s, whereas the NV-GM model yields much larger effective relaxation times, $\eta_1/\mu_1 = 315.88$~s and $\eta_2/\mu_2 = 14.12$~s. In contrast, for the generalized Kelvin-Voigt models, the distributions and magnitudes of the internal variable components are very similar, and become nearly identical after the relaxation phase. This similarity may be attributed to the serially connected architecture of the generalized Kelvin-Voigt model.

\section{Conclusion}
\label{sec:conclusion}
In this work, we have developed a constitutive framework for finite viscoelasticity, formulated for two representative rheological representations. In contrast to the generalized Maxwell model, which has been extensively studied in finite viscoelasticity, the generalized Kelvin-Voigt model has received limited attention in the literature. Its configuration with multiple elements connected in series poses inherent complexities in modeling and computation. The framework employs a strain-like tensor in $\mathrm{Sym}(3)$ as the internal variable, which conceptually represents the viscous strain. Distinct kinematic assumptions are proposed for the two configurations. The formulation unifies a broad range of canonical models, including those of Simo \cite{Simo1987,Liu2021b}, Green and Tobolsky \cite{Green1946}, Lubliner \cite{Lubliner1985}, and Miehe and Keck \cite{Miehe2000}, as special cases. 

Employing coercive strains enables a consistent construction of elastic and viscous deformation tensors, thereby permitting the incorporation of more general hyperelastic energy designs. In this study, the micromechanically motivated eight-chain model \cite{Arruda1993,Bischoff2001} is incorporated to construct nonlinear viscoelastic models. A key feature of our approach is that the strain parameters can be calibrated from experimental data. As a result, the kinematic decomposition is not imposed a priori but is instead adapted to the specific material. Our calibration for VHB 4910 demonstrates that the flexible kinematic decomposition enables the finite linear viscoelastic models to deliver competitive performance. In the meantime, the nonlinear models achieve further improvements in accuracy, especially in the large-strain regime.

The formulation is derived within a thermodynamically consistent framework, using the dissipation potential and Ziegler’s principle of maximum dissipation to characterize the dissipative mechanisms. The resulting constitutive relations exhibit distinctive features, reflecting different underlying rheological architectures. For the generalized Kelvin–Voigt model, the series-connected configuration leads to a coupled system of evolution equations. In its constitutive integration, the size of the local system grows proportionally with the number of non-equilibrium processes. Our analysis reveal that the associated matrix is a rank-one modification of a diagonal matrix, which enables an efficient inversion by utilizing the Sherman–Morrison–Woodbury formula \cite{Sherman1950,Hager1989}. This strategy extends to tensorial equations, which facilitates handling nonlinear evolution equations efficiently. As a result, the computational cost of the generalized Kelvin–Voigt model scales linearly with the number of non-equilibrium processes, matching that of the generalized Maxwell model. Numerical results show that, although the generalized Maxwell and generalized Kelvin–Voigt models may produce similar or even identical stress responses under certain parameter conditions, the evolution of their internal variables differs markedly. This behavior reflects the fundamentally distinct microscopic mechanisms of parallel and series arrangements. Moreover, clear differences arise between the linear and nonlinear evolution equations, with the nonlinear models generally exhibiting stronger dissipation.

The two considered rheological architectures serve as prototypes for more elaborate models, providing a foundation for the systematic incorporation of multiple inelastic mechanisms. The use of the dissipation potential further paves the way for extending the framework to non-Newtonian viscous behaviors. Finally, the adopted kinematic assumption eliminates the need for the troublesome intermediate configuration inherent in the traditional multiplicative decomposition and offers a promising route for effectively addressing anisotropy.

\section*{Acknowledgements}
This work is supported by the National Natural Science Foundation of China [Grant Numbers 12172160, 12472201], Shenzhen Science and Technology Program [Grant Number JCYJ20220818100600002], Southern University of Science and Technology [Grant Number Y01326127], and the Department of Science and Technology of Guangdong Province [2021QN020642]. Computational resources are provided by the Center for Computational Science and Engineering at the Southern University of Science and Technology.

\bibliographystyle{elsarticle-num}
\bibliography{viscoelasticity-theory}
\appendix
\end{document}